%% file: main.tex
\documentclass[12pt,reqno]{amsart}
\pdfoutput=1 
\usepackage{latexsym}
\usepackage{amsfonts}
\usepackage{amsthm,amssymb,mathtools,amsfonts,amsmath}
\usepackage{color}
\usepackage{graphicx}
\usepackage{url}
\usepackage{enumerate}
\usepackage{tikz-cd}
\usepackage{hyperref}
\usepackage[foot]{amsaddr}
\usepackage{stackrel}
\usepackage{hyperref}
\usepackage{cancel}
\usepackage{braket}
\usetikzlibrary{decorations.pathmorphing}
\usepackage[a4paper, margin=2cm]{geometry}
\usepackage{enumitem}
\usepackage{booktabs}
\usepackage{caption} 
\usepackage{arydshln}
\DeclareMathOperator{\tr}{tr}
\usepackage{comment}
\usepackage{dsfont}
\usepackage{nicematrix}
\usepackage{mathrsfs}  

\usepackage{ytableau}

\ytableausetup{
  mathmode,
  boxsize=0.9em,
  centertableaux
}

\usepackage{tikz}
\usetikzlibrary{arrows.meta}

\usepackage{thmtools}
\declaretheorem[parent=section]{theorem}
\declaretheorem[parent=section,numbered=no,name=Theorem]{theorem*}
\declaretheorem[numberlike=theorem, name=Proposition]{proposition}
\declaretheorem[name=Proposition, numbered=no]{proposition*}
\declaretheorem[numberlike=theorem, name=Remark]{remark}

\declaretheorem[numberlike=theorem, name=Lemma]{lemma}
\declaretheorem[numberlike=theorem, name=Definition]{definition}
\declaretheorem[numberlike=theorem, name=Corollary]{corollary}

\DeclareMathOperator{\id}{id}

\def\rs#1{{\color{black}#1}}%
\def\rausdamit#1{{\color{lightgray}#1}}%
\def\rausdamit#1{{}}%

\title[DIQKD in the commuting operator framework]{\large Device-independent Quantum Key Distribution in the commuting operator framework}
\author{Gereon Koßmann$^{1}$}
\address{$^1$ Institute for Quantum Information, RWTH Aachen University, Aachen, Germany}
\author{René Schwonnek$^{2}$}
\address{$^2$ Leibniz Universität Hannover, Hannover, Germany}
\author{Po-Chieh Liu$^{3}$ \and Hao-Chung Cheng$^{3}$}
\address{$^3$ Department of Electrical Engineering and Graduate Institute of Communication Engineering,\\
National Taiwan University, Taipei 106, Taiwan (R.O.C.)\\
Department of Mathematics, National Taiwan University\\
Center for Quantum Science and Engineering, National Taiwan University\\
Hon Hai (Foxconn) Quantum Computing Center, New Taipei City 236, Taiwan (R.O.C.)\\
Physics Division, National Center for Theoretical Sciences, Taipei 10617, Taiwan (R.O.C.)}

\bigskip
\def\wech#1{}%

\begin{document}

\maketitle

\vspace{-1cm}
\input{abstract2}

\section{Introduction}

Quantum key distribution (QKD) has evolved from an idea that started the field of quantum information theory, trough 
an motivation for experiments that demonstrate fundamental quantum effects into a cryptographic technology that is partially approaching real-world practicality. 
Along this path our collection of available security proof techniques had to keep up as well and is still far from being completed. 
The central promise of QKD is to base cryptographic security on the fundamental structures of quantum theory itself, in particular on the monogamy constraints imposed on quantum correlations \cite{Ekert1991,Ekert2014}. 
Over the past decades, such security ideas \cite{Ekert1991,Bennett1992a,Bennett1992b,Bennett_2014} have been refined, linked to experiments \cite{Jennewein2000,Tittel2000,Naik2000}, and placed in an information-theoretic framework central to modern QKD \cite{RENNER2008,tomamichel2013frameworknonasymptoticquantuminformation}.

A technology, that just has started to leave the stage of a purely theoretical idea \cite{Mayers,Barrett2005,Acn2006_1,Acn2006_2,Acn2007,Pironio2009} in order to reach the stage of real experiments \cite{Liu2022,Nadlinger2022,Zhang2022,lu2026device} is Device-independent QKD (DIQKD). Here the fundamental challenge of modeling imperfect devices is adressed by  reducing the assumptions made in a modeling process to a minimum.
In order to achieve this, DIQKD experiments essentially turn experimental setups designed for a loophole-free Bell test into  cryptographic devices. Those are then  only treated through their classical input-output behavior, while the security analysis relies, at the device level, only on the validity of quantum theory and the availability of local private randomness \cite{Mayers}. This makes DIQKD the most stringent form of QKD security, but also creates a clear mathematical requirement: if a proof is meant to be device-independent, then its mathematical model of the devices should be as general as the most general formulation of quantum theory allows. Providing such a formulation is the aim of this work.

Naively, one might try to describe a device by the collection of all positive operator-valued measure (POVM) families on all Hilbert spaces that could realize its input-output behavior. Taken literally, without imposing further structure, this collection is already too large to be a set and in any case  not a well behaving object for a security proof.
Instead, natural, well understood, and indeed also well behaving objects for a general description are provided by the commuting-operator framework \cite{Tsirelson1993,Junge2010,Junge2011}. Here the language of universal $C^*$-algebras is used for modeling quantum systems. The beauty of this approach is that we can recover any possible implementation of an experiment by representations of this algebra. This is, 
every concrete implementation of devices corresponds to a representation of this algebra, and conversely every representation gives such an implementation.

Throughout this work we will consider bipartite protocols involving two parties Alice and Bob, see \autoref{fig:diqkd_sketch}. A generalization to protocols with more parties \cite{bluhm2026device,kossmann2025routed,Lobo2024} is straight forward. The local measurement device of a party, say Alice,  is modeled by assigning an abstract effect  $M_{a|x}$ to each event with input $x$ and output $a$. Demanding that $M_{a|x}\geq 0$ and $\sum_a M_{a|x}=1$ is then already sufficient to define \cite{Blackadar2006} the abstract $C^*$ algebra $\mathcal{A}$   generated by the $M_{a|x}$. It contains all observables Alice could measure. The algebra $\mathcal{B}$ of Bob is defined correspondingly. A system that contains  Alice and Bob is then given by any algebra $\mathcal{C}_{AB}$ that contains $\mathcal{A}$ and $\mathcal{B}$. The fact that Alice and Bob have independent laboratories is modeled by imposing that $\mathcal{A}$ and $\mathcal{B}$ commute. This is where the name of the framework originates.  
\input{figure-qkd}. Now, a representation of an algebra $\mathcal{C}_{AB}$ is in essence a  linear map $\pi:\mathcal{C}_{AB} \to \mathcal{B}(\mathcal{H})$ that assigns bounded operators on a concrete Hilbertspace $\mathcal{H}$ to the abstract observables from before. This is done in a way that preserves the algebra structure. A sufficiently large collection of representations can be obtained by the GNS-construction, which assigns a concrete Hilbertspace realization to every state $\omega$. 

What this ansatz not contains is the guarantee that, for an representation on $\mathcal{H}$, we can always find local Hilbertspaces $\mathcal{H}_A$, $\mathcal{H}_B$, and an environment $\mathcal{H}_E$ such that we could decompose  $\mathcal{H}=\mathcal{H}_A\otimes \mathcal{H}_B\otimes\mathcal{H}_{E}$. So in essence, we have that tensor products of Hilbert spaces do not suffice to describe all valid actions and attacks on an DIQKD system held by Alice Bob. This point is important as here most of the existing DIQKD security proofs more or less outspokenly deviate from the general commuting operator model by restricting to situations in which such a tensor product decomposition exists \cite{ArnonFriedman2020,Tan_2021,Brown_2024}. A major reason for that might be that in this case modern finite dimensional building blocks like  left over hashing  \cite{Ren05, Dupuis2023, regula2026rethinkingquantumsmoothentropies} or entropy accumulation \cite{ArnonFriedman2020,Metger2024,Metger2024} and its corresponding entropy definitions can be more or less smoothly be taken over. Moreover, worst case models of an attacker, i.e.~the Hilbertspace $\mathcal{H}_{E}$ and the often made simplification that all measurements can w.l.o.g.~be assumed to be projective, build on Stinespring dilation arguments \cite{Tan_2021} that, up to now, demanded the tensor product structure. 

The question whether this restriction is only a technicality or reflects a genuine loss of generality that is observable in an experiment is precisely the content of the various forms of Tsirelson's problem and Cones embedding conjecture \cite{Kirchberg1993,Tsirelson1993,OZAWA2004,PrezGarca2008,Junge2010,Paulsen2016,Slofstra_2019,ji2022mipre}. This question is now known to have a negative answer \cite{SLOFSTRA_2019_not_closed,Slofstra_2019,MIPstar=RE}.  The question whether this has relevant cryptographic implications is however still unclear. An operational characterization of states that admit a tensor product representation was given in \cite{van2024schmidt} in terms of a generalized Schmidt rank. From this perspective the assumption of a tensorproduct structure can be seen as assuming only a finite amount of entanglement between Alice and Bob, which is a clear deviation from full generality when considering an attacker Eve only limited by fundamental laws of quantum physics.  

Lastly a technical problem arises in the process of numerically computing secure key rates. For protocols beyond the CHSH setting, the standard path \cite{Tan_2021,Brown_2024,kossmann2025reliableentropyestimationobserved} for this is to reduce this task to an non-commutative polynomial optimization problem which is the relaxed by the Navascues--Pironio--Acín (NPA) hierarchy \cite{Navascus2007,Navascus2008}, one of the main numerical tools in device-independent quantum information.  These hierarchies however attain their limit in the commuting-operator model and not in a fixed  tensor product model. 

In this work, we close these structural gaps for an AEP based secrecy analysis of DIQKD. We formulate Alice's and Bob's devices by universal $C^*$-algebras generated by their measurement effects and describe the joint Alice--Bob observable algebra using the maximal tensor product. After fixing a state compatible with the observed statistics, the GNS construction gives a represented von Neumann algebra for Alice and Bob. Eve's side information is then not introduced as an external tensor factor, but as the commutant of this represented algebra. We next prove the corresponding POVM-to-PVM reduction. In a Hilbert space model one would normally appeal to Naimark dilation, but in the commuting-operator setting the dilation has to preserve not only Alice and Bob's observed statistics but also all correlations with arbitrary operators in Eve's commutant. We show that this can be done, so that local POVMs may be replaced by PVMs without loss of generality for the secrecy relevant optimization. Since Eve's algebra may be an arbitrary von Neumann algebra, the relative entropy, which is the central quantity of a security analysis, has to be expressed using modular operators \cite{kosaki_1986_variationalformula} rather than by the explicit formula $\tr(\rho( \log \rho-\log \sigma))$ \cite{umegaki1962conditional}. A core technical tool, is the extension of Frenkel's integral representation \cite{Frenkel2023} to von Neumann algebras. For which we have to generalize the operator layer cake representation of the Gateaux derivative of the logarithm \cite{liu2025layercakerepresentationsquantum} accordingly. Combining this new entropy formula with the asymptotic equipartition theorem in von Neumann algebras \cite{Berta_2015} and the techniques of \cite{kossmann2025reliableentropyestimationobserved} yields a noncommutative polynomial optimization problem for the amount of randomness that can be securely extracted in the presence of an attacker. This formulation is fully compatible with underlying algebraic structure, and therefore makes NPA-type relaxations rigorous within the commuting-operator framework. In this sense, we provide a representation independent route from observed statistics to an AEP based security statement, with all objects formulated inside the commuting-operator framework. 

This paper is structures as follows. In the subsequent  \autoref{sec:results}  we briefly comment on two theorems  that build the technical core of this work and might be from general interest. In \autoref{sec:prelim} we briefly review the  basic algebraic concepts needed to follow the content of this work. In \autoref{sec:qkd} the basic concepts and challenges of a DIQKD security proof are outlined. In \autoref{sec:commop} a formulation of DIQKD in the commuting operator framework is worked out in detail. In \autoref{sec:operator_layer_cake} our main technical tool, the von Neumann version of the operator layer cake theorem is provided. In \autoref{sec:approx} we collect the numerical tools needed for computing key rates in a fully DI setting and proof their convergence.  

\section{Challenges and Results}\label{sec:results}

Considering DIQKD from the perspective of quantum games turns out to be a very powerful and mathematically rigorous point of view for the analysis of the underlying experiment, which we already sketched in \autoref{fig:diqkd_sketch}. However, several serious challenges arise the moment we leave the well-established perspective of Hilbert space based quantum theory. Indeed, the definition of security of a DIQKD protocol needs to be adapted, and most of the tools required to claim security in a rigorous mathematical manner need to be developed within the commuting operator framework, which is the correct framework for analyzing quantum games. In this work, we do not aim to lift all techniques available in the Hilbert space world to the commuting operator framework, as this would go strictly beyond the scope of this work. To nevertheless deliver all results with mathematical rigor, we use a leftover hashing result from \cite{Berta_2015}, which turns a one-shot security statement into an entropic statement, and we apply the simplest form of a reduction theorem for an $n$-round procedure of a DIQKD protocol sketched in \autoref{fig:diqkd_sketch}, namely the asymptotic equipartition (AEP) theorem, which is available on von Neumann algebras \cite{Fawzi_2025}. Even though both results are not state of the art in the device-dependent/Hilbert space world (cf.~\cite{Dupuis2023} for a state-of-the-art privacy amplification result and \cite{Metger2024} for an entropy accumulation theorem), both results are formulated in the commuting operator framework and thus fit into our point of view.
After collecting statistics in an $n$-round procedure of \autoref{fig:diqkd_sketch}, and even assuming that the rounds are independent and identically distributed such that we can apply the AEP from \cite{Fawzi_2025}, it has been an open question in the commuting operator framework whether we are allowed to assume that the single-round reduction from the AEP can be considered with projective measurements instead of general effects (cf.~e.g.~\cite{Tan_2021}). Roughly speaking, the reason for this is that we are not allowed to assume that the purification of any state lives in a fixed additional Hilbert space, which could simply be tensored on as one would do in a Naimark dilation argument in Hilbert spaces (see also \autoref{subsec:randomness_extraction_DIQKD}). Our first result is a dilation theorem showing that, in a certain sense, an embedding of the commutants of the Gelfand-Naimark-Segal representation of the underlying state on the effect algebra into the algebra spanned by projection-valued measurements is possible.
\begin{theorem}[Dilation theorem, informal]
    Let $\psi$ be a state on the POVM algebra with GNS triple with purification $\zeta_\psi$. Then there exist
    \begin{enumerate}
        \item[(a)] a Hilbert space $\mathcal{H}$,
        \item[(b)] a unit vector $\xi \in \mathcal{H}$,
        \item[(c)] a unital $\star$-representation from the PVM-algebra into $\mathcal{B}\!\left(\mathcal{H}\right)$,
        \item[(d)] an injective normal unital $\star$-homomorphism connecting the commutants,
    \end{enumerate}
    such that $\zeta_\psi$ behaves on the effect algebra and its commutant as $\xi$ behaves on the PVM-algebra and its commutant.
\end{theorem}
Given this dilation theorem, we can show that the optimization problems of the single-round entropy can be computed without loss of generality with the assumption that Alice and Bob have PVM's at their disposal (cf.~\autoref{cor:povm_pvm_equality}). 

Randomness is quantified by entropic quantities, and thus the natural objective function in the optimization problem resulting from the AEP reduction to a single round involves, in this case, the conditional von Neumann entropy on a general von Neumann algebra. As we are not allowed to assume that Eve always admits a type-$I$ factor description, we must assume that the adversary's information is correctly described only by a general, possibly type-$III$, von Neumann algebra. However, the techniques available for the optimization of the conditional von Neumann entropy for DIQKD assumed at some point a Hilbert space model and tensor products and did not consider the general case \cite{Brown_2024,kossmann2025reliableentropyestimationobserved}. To overcome this issue and to show that the technique from \cite{kossmann2025reliableentropyestimationobserved} leads to the correct optimization problem in the commuting operator framework, we provide in this work a proof of the integral formula by Frenkel \cite{Frenkel2023} in the case of von Neumann algebras. To achieve this, we choose the path via the operator layer cake theorem introduced in \cite{cheng2025errorexponentsquantumpacking} for matrices. To be concrete, we prove the following main theorem, which is an integral representation for the Gateaux derivative of the operator logarithm.
\begin{theorem}[Operator layer cake]
Let $\mathcal{H}$ be a separable Hilbert space and let $A,B\in\mathcal{B}(\mathcal{H})$ with $A\geq \delta I$ for some $\delta > 0$ and $B=B^\ast$.
Then the following identity holds in strong operator topology:
\begin{align}
\operatorname{D}\log(A)[B]
=
\int_{0}^{\infty} H(B-uA) \, du
-
\int_{-\infty}^{0} \left(I-H(B-uA)\right)du,
\end{align}
where $H(\cdot)$ denotes the Heaviside function defined in \eqref{eq:def_heaviside} below. 
\end{theorem}
With some technical effort, we are able to show the integral formula by Frenkel for the relative entropy in its general form, establishing another integral formula for the relative entropy next to the Kosaki formula \cite{kosaki_1986_variationalformula}. 
\begin{theorem}[Frenkel formula for general von Neumann algebras]
Let $\mathcal{M}$ be a von Neumann algebra and let $\rho,\sigma\in\mathcal{M}_+^\star$ \rs{ with support $s_{\mathcal{M}}(\rho)\leq s_{\mathcal{M}}(\sigma)$}. Then
\begin{align}
D(\rho\Vert\sigma)
=
\left(\rho-\sigma\right)(1)
+
\int_0^1 \frac{1}{s}\,\left(s\sigma-\rho\right)_+(1)\,ds
+
\int_1^\infty \frac{1}{s}\,\left(\rho-s\sigma\right)_+(1)\,ds,
\end{align}
with both sides possibly infinite.
\end{theorem}

Given the integral formula for the relative entropy, it turns out that the discretization techniques from \cite{kossmann2025reliableentropyestimationobserved} can be adopted to the setup of von Neumann alegbras and deliver for fixed states $\rho$ and $\sigma$ an approximation. However, what is a priori not clear is that this is compatible with the Navascues-Pironio-Acin (NPA)-hierarchy \cite{Navascus2007,Navascus2008} as the NPA-hierarchy optimizes over the maximal tensor product of universal $\mathcal{C}^\star$-algebras as shown in \cite{Ligthart_2023}. For this purpose we show that this step is indeed justified and the optimization of the conditional von Neumann entropy can be indeed executed in the framework of the NPA-hierachy and a corresponding universal $\mathcal{C}^\star$-algebra (cf.~\autoref{thm:equivalence_convergence}).

\section{Preliminaries on operator algebras}\label{sec:prelim}
\rs{Before turning to the device-independent cryptographic setting, we recall the operator-algebraic notions used throughout the paper. The aim of this section is not to give a self-contained introduction to the theory of von Neumann algebras, but rather to collect the concepts needed by a reader with a background in quantum information theory to follow the later construction. We begin with basic $C^*$-algebraic terminology, including positivity, states, representations, and maximal tensor products, since these provide the language in which commuting-operator models of Bell experiments are formulated. We then recall the corresponding von Neumann algebraic notions: normal states, traces, type decompositions at the level needed here, standard forms, support projections, and Araki relative entropy. Finally, we record two technical tools used later in the paper: Haagerup reduction, which allows certain arguments to be reduced to the finite tracial case, and the functional calculus for the operator logarithm, whose directional derivative is the starting point for the layer-cake representation of relative entropy.}

In this work we consider $\mathcal{C}^\star$-algebras $\mathcal{A}$, which are Banach-algebras equipped with a conjugate-linear isometric anti-automorphism of order two \cite{Blackadar2006}, such that the norm $\Vert \cdot\Vert$ satisfies $\Vert a^\star a\Vert = \Vert a\Vert^2$ for all elements $a\in \mathcal{A}$. An element $a \in \mathcal{A}$ is called positive, if there exists $b\in \mathcal{A}$ such that $a = b^\star b$. The closed convex cone of positive elements is denoted with $\mathcal{A}_+$. A representation of a $\mathcal{C}^\star$-algebra $\mathcal{A}$ is a $\star$-homomorphism from $\mathcal{A}\to \mathcal{B}\!\left(\mathcal{H}\right)$ for some Hilbert space $\mathcal{H}$. A representation $\pi$ is called non-degenerate if $\mathcal{H} = \overline{\{\pi(x)\xi \ \vert \ x \in \mathcal{A},\xi \in \mathcal{H}\}}^{\Vert \cdot \Vert}$. Let $\mathcal{A}$ and $\mathcal{B}$ be $\mathcal{C}^\star$-algebras. The \emph{algebraic tensor product} $\mathcal{A} \odot \mathcal{B}$ is the algebraic tensor product of vector spaces equipped with the $\star$-algebra structure
\begin{align}
(a_1 \otimes b_1)(a_2 \otimes b_2) \coloneqq  (a_1a_2)\otimes(b_1b_2),
\qquad
(a \otimes b)^* \coloneqq  a^\star \otimes b^\star.
\end{align}
For $x \in \mathcal{A} \odot \mathcal{B}$ define the \emph{maximal $\mathcal{C}^\star$-seminorm} by
\begin{align}
\Vert x\Vert_{\max} \coloneqq  \sup \left\{ \Vert (\pi_{\mathcal{A}} \cdot \pi_{\mathcal{B}})(x)\Vert 
\ \vert \  (\pi_{\mathcal{A}},\pi_{\mathcal{B}})\ \text{commuting $\star$-representations on some $\mathcal{H}$}\right\},
\end{align}
where $(\pi_\mathcal{A} \cdot \pi_{\mathcal{B}})(a \otimes b) \coloneqq  \pi_{\mathcal{A}}(a)\pi_{\mathcal{B}}(b)$. 
Then the \emph{maximal tensor product} $\mathcal{A} \otimes_{\max} \mathcal{B}$ is the completion of 
$\mathcal{A} \odot \mathcal{B} / \{x \ \vert \ \Vert x\Vert_{\max}=0\}$ with respect to $\Vert \cdot\Vert_{\max}$. 

We refer to the free product "$\star$" of $\mathcal{C}^\star$-algebras as the universal $\mathcal{C}^\star$-algebras by its subalgebras and relations. As we just consider algebras generated by positive operator valued measure constraints, it becomes clear how to build the universal $\mathcal{C}^\star$-algebras. We furthermore denote by $\ell^\infty(\Omega)$ for a countable set $\Omega$, the commutative $\mathcal{C}^\star$-algebra on the symbols of $\Omega.$
A von Neumann algebra $\mathcal{M} \subseteq \mathcal{B}(\mathcal{H})$ is a $\mathcal{C}^\star$-algebra, which satisfies $\mathcal{M}^{\prime \prime} = \mathcal{M}$. For von Neumann algebras $\mathcal{M}_1 \subseteq \mathcal{B}(\mathcal{H}_1)$ and $\mathcal{M}_2 \subseteq \mathcal{B}(\mathcal{H}_2)$, the spatial tensor product is defined as the weak$^\star$-closure of the algebraic tensor product $\mathcal{M}_1 \odot \mathcal{M}_2 \subseteq \mathcal{B}(\mathcal{H}_1 \otimes \mathcal{H}_2)$, which is denoted as $\mathcal{M}_1 \overline{\otimes} \mathcal{M}_2$.  
\subsection{States}
As $\mathcal{C}^\star$-algebras are Banach spaces, we can define its canonical dual space as the set of linear functionals $\omega :\mathcal{A} \to \mathbb{C}$, which are continuous with respect to the  $\mathcal{C}^\star$-norm $\Vert \cdot\Vert$. In particular positive functionals, i.e.~$\omega(a)\geq 0$ for all $a\in \mathcal{A}_+$, are denoted with $\mathcal{A}_+^\star$ and states are positive functionals with $\Vert \omega\Vert =1$, whereby
\begin{align}
    \Vert \omega \Vert \coloneqq \sup_{x \in \mathcal{A}, \ \Vert x\Vert\leq 1} \vert\omega(x)\vert.
\end{align}
For a von Neumann algebra $\mathcal{M}$ we denote by $\mathcal{M}_+^\star$ the set of normal linear functionals, i.e.~linear functionals, which are also weak$^\star$-continuous and $\mathcal{S}(\mathcal{M})$ the set of normal states.  

\subsection{Types and traces}

A \emph{trace} on a von Neumann algebra $\mathcal{M}$ is a map $\tau \colon \mathcal{M}_+ \to [0,\infty]$
which is additive, positively homogeneous, and satisfies the tracial property
\begin{align}
\tau(x^*x)=\tau(xx^*)
\qquad \text{for all } x\in \mathcal{M}.
\end{align}
A trace is called \emph{faithful} if $\tau(x)=0$ implies $x=0$ for every $x\in \mathcal{M}_+$, \emph{normal} if it is continuous with respect to increasing limits of positive operators, and \emph{semi-finite} if for every $x\in \mathcal{M}_+$ there exists $0\neq y\in \mathcal{M}_+$ with $y\leq x$ and $\tau(y)<\infty$.

The theory of von Neumann algebras is partly characterized by types, which, in a rough sense, classify von Neumann algebras according to the existence of a trace. For the purpose of this work, it is enough to introduce three classes:
\begin{enumerate}
    \item \emph{finite.} A finite von Neumann algebra admits a normal faithful trace which is finite, i.e.~$\tau(1)=1$.
    \item \emph{semi-finite.} A semi-finite von Neumann algebra admits a normal faithful semi-finite trace.
    \item \emph{type-III.} A type-III von Neumann algebra is not semi-finite.
\end{enumerate}
The simplest example of finite, non-commutative von Neumann algebras are full matrix algebras, which we denote with $\mathbb{M}_n(\mathbb{C})$ for $n\times n$-dimensional matrices. 

\subsection{Standard form and relative entropy}
The standard form of a von Neumann algebra \cite{Takesaki2003} $\mathcal{M}$ is defined by a quadruple $(\mathcal{M},\mathcal{H},J,P)$, whereby $\mathcal{H}$ is a Hilbert space, $J$ is an injective anti-linear isometry on $\mathcal{H}$ and $P \subseteq \mathcal{H}$ is a self-dual cone and  with an injective $\star$-homomorphism $\pi:\mathcal{M} \to \mathcal{B}(\mathcal{H})$, such that 
\begin{enumerate}
\item $J^2 = 1$, $J\mathcal{M}J = \mathcal{M}^\prime$,
\item $JaJ = a^\star$ for $a \in \mathcal{M} \cap \mathcal{M}^\prime$,
\item $J\xi = \xi$, $\xi \in P$,
\item $aJaJP \subseteq P$ for $a \in \mathcal{M}$.
\end{enumerate}
Given the standard form, there exists a unique vector $\xi_\phi \in P$ for each positive normal functional $\phi\in\mathcal{M}_+^\star$ implementing $\phi$ via
\begin{align}
    \phi(x) = \langle\xi_\phi\vert x \xi_\phi\rangle,\qquad x \in \mathcal{M}.
\end{align}
For positive normal linear functions $\phi \in \mathcal{M}_+^\star$ (and indeed even positive forms) there exists a unique projection $s_{\mathcal{M}}(\phi)$ with the property 
\begin{align}
    \phi(x) = \phi(xs_{\mathcal{M}}(\phi)) = \phi(s_{\mathcal{M}}(\phi)x), \quad x \in \mathcal{M}.
\end{align}
This projection is called \emph{support projection} of $\phi$ (cf.~\cite[sec.~5.15]{Stratila2019}). 

Let $\psi,\phi \in \mathcal{M}_+^\star$ with vector representations $\xi_\psi,\xi_\phi \in P$. Then define
\begin{align}
    S_{\psi,\phi}^0 (a \xi_\phi + \eta) = s_{\mathcal{M}}(\phi)a^\star \xi_\psi, \qquad a \in \mathcal{M}, \ \eta \in (1-s_{\mathcal{M}^\prime}(\phi))\mathcal{H}.
\end{align}
This is a closable and anti-linear operator, whereby we denote with $S_{\psi,\phi}$ its closure. The relative modular operator is defined as 
\begin{align}
    \Delta(\psi,\phi) \coloneqq S_{\psi,\phi}^\star S_{\psi,\phi}.
\end{align}

For elements $\rho \in \mathcal{S}_\leq(\mathcal{M})$ and $\sigma \in \mathcal{M}_+^\star$ with $s_{\mathcal{M}}(\rho)\leq s_{\mathcal{M}}(\sigma)$, the Umegaki relative entropy is defined as (cf.~\cite{Araki1975}) 
\begin{align}\label{eq:def_araki_relative_entropy}
    D(\rho \Vert \sigma) \coloneqq -\langle \xi_\rho,\log  \Delta(\sigma,\rho)\xi_\rho\rangle.
\end{align}

We have the following properties of the Umegaki relative entropy, which readily follow from \eqref{eq:def_araki_relative_entropy}.
\begin{lemma}\label{lem:araki_order_scaling}
Let $\rho\in \mathcal S_{\le}(\mathcal M)$ and let
$\sigma,\sigma_1,\sigma_2\in \mathcal M_+^\ast$ with
$s_{\mathcal M}(\rho)\le s_{\mathcal M}(\sigma)$.

\begin{enumerate}
    \item[(i)] If $\sigma_1\le \sigma_2$ and
    $s_{\mathcal M}(\rho)\le s_{\mathcal M}(\sigma_1)$, then
    \begin{align}
        D(\rho\Vert \sigma_2)\le D(\rho\Vert \sigma_1).
    \end{align}

    \item[(ii)] For every $\lambda>0$,
    \begin{align}
        D(\rho\Vert \lambda \sigma)
        =
        D(\rho\Vert \sigma)-\rho(1)\log \lambda.
    \end{align}
    In particular, if $\rho$ is a state, then
    \begin{align}
        D(\rho\Vert \lambda \sigma)
        =
        D(\rho\Vert \sigma)-\log \lambda.
    \end{align}
\end{enumerate}
\end{lemma}
It is well-known that, in the case where the von Neumann algebra $(\mathcal{M},\tau)$ is semi-finite with trace $\tau$, the Umegaki relative entropy can be written as
\begin{align}\label{eq:def_relative_entropy_trace}
D(\rho\Vert\sigma)
\coloneqq 
\begin{cases}
\tau\left(\rho(\log\rho-\log\sigma)\right), & \text{if } s(\rho)\leq s(\sigma),\\[0.3em]
+\infty, & \text{otherwise},
\end{cases}
\end{align}
where $s(\rho)$ and $s(\sigma)$ denote the support projections of $\rho$ and $\sigma$, respectively.
\subsection{Haagerup reduction}
The Haagerup reduction technique \cite{Haagerup2009} is a powerful tool in the theory of von Neumann algebras, as it allows one to approximate the noncommutative $L_p$-spaces of a von Neumann algebra $\mathcal{M}$ by an increasing family of subalgebras $(\mathcal{M}_n)_{n\in \mathbb{N}}$. We recall the formulation given in \cite[Thm.~2.1, Prop.~2.2]{Fawzi_2025}, since it directly highlights the properties used in this work.

\begin{proposition}[Haagerup reduction]\label{prop:Haagerup_reduction}
    Let $\mathcal{M}$ be a von Neumann algebra. Then there exists a sequence of finite von Neumann algebras with normal faithful trace $(\mathcal{M}_n,\tau_n)$ such that, for each pair of states $\rho,\sigma \in \mathcal{S}(\mathcal{M})$, one can find, after appropriate identification, sequences of states $(\rho_n)$ and $(\sigma_n)$ with $\rho_n,\sigma_n \in \mathcal{S}(\mathcal{M}_n)$ such that
    \begin{align}
        \lim_{n\to \infty} D(\rho_n \Vert \sigma_n) = D(\rho\Vert \sigma).
    \end{align}
\end{proposition}

\subsection{Operator logarithm and its directional derivative}

Let $\mathcal{H}$ be a separable Hilbert space and $\mathcal{B}(\mathcal{H})$ the algebra of bounded operators.
For a bounded operator $A\in\mathcal{B}(\mathcal{H})$ with
\begin{align}
\sigma(A)\subset \mathbb C\setminus (-\infty,0]
\end{align}
we denote by $\log$ the principal branch of the complex logarithm on $\mathbb C\setminus (-\infty,0]$ and define the operator logarithm by the holomorphic functional calculus
\begin{align}\label{eq:def_operator_log}
\log(A)\coloneqq \frac{1}{2\pi i}\int_{\Gamma_A} \log(z)\,(zI-A)^{-1}\,dz,
\end{align}
where $\Gamma_A$ is a positively oriented, simple smooth closed curve enclosing $\sigma(A)$ and lying in $\mathbb C\setminus (-\infty,0]$.

If $A>0$ is boundedly invertible and $B=B^\ast\in\mathcal{B}(\mathcal{H})$, then $A+tB$ stays invertible for $|t|$ small enough, and we define the \emph{directional derivative} (G\^{a}teaux derivative) of $\log$ at $A$ in direction $B$ by
\begin{align}\label{eq:def_DLog}
\operatorname{D}\log(A)[B]\coloneqq \left.\frac{d}{dt}\right|_{t=0}\log(A+tB),
\end{align}
where the derivative is taken in the strong operator topology (equivalently, pointwise on vectors).
It is standard that $\operatorname{D}\log(A)[B]$ exists and is bounded, and one may also express it via the resolvent integral
\begin{align}\label{eq:resolvent_derivative}
\operatorname{D}\log(A)[B]=\frac{1}{2\pi i}\int_{\Gamma_A} \log(z)\,(zI-A)^{-1}B(zI-A)^{-1}\,dz,
\end{align}
again in strong operator topology.

\section{Challenges in Device Independent Quantum Key Distribution} \label{sec:qkd}

The overarching purpose of quantum key distribution (QKD) is to establish an information-theoretically secure key between distant users, usually referred to as Alice and Bob, which is provably known only to Alice and Bob themselves, up to a security parameter $\varepsilon\!>\!0$ that quantifies an arbitrarily small error parameter \cite{RENNER2008}. As was recognized in the early days of QKD, security cannot be guaranteed unless all imperfections and all mismatches between theory and experiment are taken into account \cite{Bennett1992a,Bennett1992b}. It has become apparent throughout the different stages of development that experimental mismatches are arguably difficult to characterize, and it remains unclear at what level of detailed experimental description it is sufficient to claim security \cite{Scarani2009}.

To overcome the challenge of describing an experiment with sufficient accuracy to claim security, device-independent quantum key distribution (DIQKD) was developed and is regarded as the gold standard of security techniques in the quantum world, at least from a theoretical perspective \cite{Mayers} (cf.~\cite{Ekert2014} for a discussion of fundamental limitations of mathematical modeling in a physical world). The goal of a (DI)QKD protocol is to establish a secure key relying only on the observed statistics, i.e.~based on a Bell-experiment \cite{Bell1964} (cf.~\autoref{fig:diqkd_sketch}). To achieve this goal, quantum theory makes it possible to prove security based on fundamental physical principles \cite{Acn2006_1,Acn2006_2,Pironio2009}. Standard DIQKD protocols are based on the Bell-type QKD experiment introduced in \cite{Ekert1991} and first demonstrated as a Bell experiment in \cite{Hensen2015}. The fundamental principle underlying DIQKD security is the monogamy of entanglement: roughly speaking, strong correlations obtained in a quantum experiment cannot be shared among an arbitrary number of parties. Since this property is fundamental and does not depend on a specific implementation or detailed modeling of the experiment itself, the notion of device independence has been established. 

In most earlier works (cf.~\cite{Barrett2005,Acn2006_1,Acn2006_2,Navascus2007,Navascus2008,Pironio2009,Vazirani2014} and subsequent developments), device independence has been formulated through the mathematical model in which the underlying quantum system admits a Hilbert-space representation with a trace-class operator $\rho_{AB} \in \mathcal{S}(\mathcal{H}_A \otimes \mathcal{H}_B)$, measurement operators $\{M_{a\vert x}\}$ on Alice's side, and measurement operators $\{N_{b\vert y}\}$ on Bob's side. We will refer to such a tuple as a strategy and denote it by
\begin{align}\label{eq:general_quantum_model}
    S \coloneqq \left(\rho_{AB},\{M_{a\vert x}\},\{N_{b\vert y}\},\mathcal{H}_A,\mathcal{H}_B\right).
\end{align}
The observed statistics in an experiment of the form shown in \autoref{fig:diqkd_sketch} are then given by Born's rule,
\begin{align}\label{eq:conditional_distribution}
    p(a,b\vert x,y) \coloneqq \tr[\rho_{AB} M_{a\vert x} \otimes N_{b\vert y}].
\end{align}

The key difference between device-dependence and device independence is that one even neither specifies the Hilbert spaces $\mathcal{H}_A$ and $\mathcal{H}_B$  more the measurement operators. Rather, in DIQKD one assumes only the existence of some strategy of the form \eqref{eq:general_quantum_model} that gives rise to the observed conditional distribution \eqref{eq:conditional_distribution}. Accordingly, the typical questions in DIQKD ask whether a given property of the underlying experiment holds for all strategies compatible with an observed statistic $p(a,b\vert x,y)$ connected to the mathematical modeling with \eqref{eq:conditional_distribution}. Although this point of view has been tremendously successful in the device-dependent setting, where the Hilbert spaces and measurement operators are fixed and known, substantial obstacles arise in the device-independent setting.

\subsection{Quantum Games and DIQKD}\label{subsec:quantum_games_intro} 
\input{figure-games}

The more and less obvious limitations of the Hilbert space strategy-based perspective in DIQKD are perhaps best understood by viewing a DIQKD experiment as a quantum game, see \autoref{fig:quantum_game}. Since QKD ultimately amounts to collecting statistical data under the assumption that the underlying experiment is governed by the laws of quantum theory, it is natural to compare QKD with the task of playing a quantum game (cf.~\cite{Palazuelos2016} for a review). A bipartite quantum game $\mathcal{G}$ is specified by sets of questions $\mathcal{X},\mathcal{Y}$, sets of answers $\mathcal{A},\mathcal{B}$, a rule function
\begin{align}
    V:\mathcal{X}\times \mathcal{Y}\times \mathcal{A}\times \mathcal{B} \to \{0,1\},
\end{align}
and a referee who asks Alice and Bob questions $x \in \mathcal{X}$ and $y\in \mathcal{Y}$ according to a probability mass function $\pi$ on $\mathcal{X} \times \mathcal{Y}$. The goal of Alice and Bob is to agree in advance on a strategy, in the most general case given by a conditional distribution $p(a,b\vert x,y)$, so as to maximize the success probability
\begin{align}
    \omega(\mathcal{G}) \coloneqq \sup_{p\left(a,b\vert x,y\right) \ \text{feasible}} \sum_{a,b,x,y} \pi(x,y) V(a,b,x,y)p(a,b\vert x,y).
\end{align}
If the distribution $p(a,b\vert x,y)$ arises from a quantum experiment, for instance as in \eqref{eq:conditional_distribution} and corresponding to a strategy $S$ of the form \eqref{eq:general_quantum_model}, then the game is called a quantum game. Computing the optimal value $\omega(\mathcal{G})$ is therefore equivalent to optimizing a linear functional over a set of distributions determined by the physical power available to Alice and Bob. Although in DIQKD one is not directly interested in the success probability of a game, it has turned out that important cryptographic quantities can be derived from this value when the DIQKD protocol is based on a game \cite{Pironio2009,CerveroMartn2025}. In the realm of the quantum games it was a central task to distinguish different notions on the power of the strategies such as they can be described with finite-dimensional quantum systems, approximately finite-dimensional quantum systems, general limits of such strategies and strategies in the commuting operator framework (cf.~e.g.~\cite{Kirchberg1993,Tsirelson1993,OZAWA2004,PrezGarca2008,Junge2010,Paulsen2016,Slofstra_2019,ji2022mipre} for a by far not complete list of milestone results in the realm of quantum games). 

For DIQKD, it is an immediate consequence that the perspective of Hilbert space \emph{strategies} as in \eqref{eq:general_quantum_model} is incomplete, as this becomes apparent already if we decide for a specific class of strategies studied in the context of quantum games (cf.~\cite{SLOFSTRA_2019_not_closed}). In particular, these results show that, even if one restricts \eqref{eq:general_quantum_model} to finite-dimensional strategies, to which is often referred to in DIQKD, the corresponding set of strategies is not closed. Comparing this with one of the most important and successful tools in DIQKD, namely the Navascu\'es--Pironio--Ac\'in (NPA) hierarchy \cite{Navascus2007,Navascus2008}, which is known in the setting of quantum games to converge to the commuting-operator value of the game (cf.~\cite[Lem.~4]{Ligthart_2023} and \cite{Paulsen2016}), reveals a genuine discrepancy between the way DIQKD is modeled and the tools available to analyze it. This calls for a unified framework.

\subsection{Randomness extraction in DIQKD}\label{subsec:randomness_extraction_DIQKD}
A less obvious challenge in DIQKD lies in the fact that the central quantity to be quantified in QKD is \emph{randomness}. Modern security frameworks capture this through entropic quantities \cite{RENNER2008} (cf.~\cite{regula2026rethinkingquantumsmoothentropies} for a recent result on randomness extraction in device-dependent QKD). In cryptography, however, one is interested in the amount of randomness that remains secure in the presence of a potential eavesdropper. In the Hilbert-space model of quantum theory, this question can be addressed using purifications of quantum states and the fact that pure states are, in a precise sense, uncorrelated with any additional systems. More specifically, if $\rho_{AB} \in \mathcal{S}(\mathcal{H}_A\otimes \mathcal{H}_B)$, then we have for any extension $\rho_{ABEE^\prime}$
\begin{align}\label{eq:purification_property_finite_dimensional_spaces}
    \tr_{E^\prime}[\rho_{ABEE^\prime}]
    = \lvert \psi_{ABE} \rangle\!\langle \psi_{ABE} \rvert
    \quad \Rightarrow \quad
    \rho_{ABEE^\prime}
    = \lvert \psi_{ABE} \rangle\!\langle \psi_{ABE} \rvert \otimes \rho_{E^\prime},
\end{align}
which means precisely that the state is a product across the bipartition $ABE:E^\prime$. Moreover, it follows that all purifications can be modeled within a fixed quantum system $\mathcal{H}_A \otimes \mathcal{H}_B \otimes \mathcal{H}_E$, with $\mathcal{H}_E \cong \mathcal{H}_A \otimes \mathcal{H}_B$. By Uhlmann's theorem \cite{Uhlmann1976} and the invariance of conditional entropies under local isometries in the Hilbert-space framework (cf.~\cite{Tomamichel2016}), one may therefore fix, without loss of generality, a quantum system $\mathcal{H}_E \cong \mathcal{H}_A \otimes \mathcal{H}_B$ and solve the randomness-extraction problem in this setting. Unfortunately, this property fails in the general commuting operator framework, as was recently shown in \cite{vanLuijk2026}. Thus, combining the elegant perspective in \autoref{subsec:quantum_games_intro} with randomness extraction over side-information is at least challenged by an argumentation for a \emph{correct} notion of purification for the purpose of randomness extraction. 

\subsection{Turning experimental data into security statements}\label{subsec:turning_experimental_data}

The composable security framework established for QKD \cite{Pfitzmann2000} provides a way to turn experimental data, such as the statistics $p(a,b\vert x,y)$, together with a given Hilbert-space model of the underlying quantum experiment, into a security statement for the resulting key. For the purposes of this work, it is sufficient to distinguish the following three core security properties \cite{RENNER2008}:
\begin{enumerate}
    \item \emph{Completeness}: If both parties behave honestly and the protocol is executed under normal operating conditions, then the protocol should abort only with small probability.

    \item \emph{Correctness}: Conditioned on not aborting, the keys held by the legitimate parties should be identical except with small probability.

    \item \emph{Secrecy}: Conditioned on not aborting, the final key should be indistinguishable from an ideal uniformly random key that is independent of the adversary's information.
\end{enumerate}

Among these three properties, only \emph{secrecy} truly depends on the underlying quantum model, since it requires a characterization of the adversary's possible side information \cite{ferradini2025definingsecurityquantumkey}. By contrast, \emph{completeness} and \emph{correctness} are determined from the classically observed data and an honest implementation usually in finite dimensions, namely whether the protocol aborts and whether the legitimate users' classical keys agree. Thus, as this work focuses on how to turn data such as a distribution $p(a,b\vert x,y)$ into a  mathematical model for DIQKD, we are primarily interested in proving \emph{secrecy}. In device-dependent QKD, secrecy is defined by requiring that, conditioned on the protocol not aborting, the final key be close to an ideal uniformly random key that is independent of the adversary's side information (cf.~\cite[Rem~6.1.3]{RENNER2008}). We give a concrete definition next.

\begin{definition}[Secrecy in device-dependent QKD]\label{def:secrecy_ddqkd}
Let $\Omega$ denote the event that the protocol does not abort, and let $p(\Omega)$ be its probability. Let $\rho_{K_AE\mid\Omega}$ be the joint state of Alice's final key $K_A$ and Eve's side information $E$, conditioned on $\Omega$, where $E$ includes all information available to the adversary, in particular the public classical transcript. The protocol is called $\varepsilon_{\mathrm{sec}}$-secret if
\begin{align}\label{eq:secrecy_statement_ddqkd}
p(\Omega)\,\Vert \rho_{K_AE\vert\Omega} - \tau_{K_A}\otimes\rho_{E\vert \Omega}\Vert
\le \varepsilon_{\mathrm{sec}},
\end{align}
and $\tau_{K_A}=\frac{1}{2^\ell}\sum_{k\in\{0,1\}^\ell}\lvert k \rangle\!\langle k \rvert$ is the fully mixed state on the $\ell$-bit key space.
\end{definition}

To derive a secrecy statement from experimental data, a leftover-hash lemma \cite{Tomamichel2011,Berta_2015,Dupuis2023,regula2026rethinkingquantumsmoothentropies} converts the one-shot quantity in \eqref{eq:secrecy_statement_ddqkd} into an entropic quantity for a multi-round experiment. Building on this, entropy-accumulation theorems \cite{Dupuis_2019,Dupuis_2020,Metger2024}, which can be viewed as generalizations of the quantum asymptotic equipartition theorem \cite{Tomamichel2009,Fawzi_2025}, reduce the security analysis to a single optimization problem involving a single-round quantity, typically a conditional entropy $\mathbb{H}$ \cite{Tomamichel2016}. The resulting optimization program is then given by
\begin{equation}\label{eq:finite_dimensional_problem}
    \begin{aligned}
        \inf \ &\mathbb{H}(A\vert X=\tilde{x},E)_{\rho_{AE}} \\
        &\tr[\rho_{AB} M_{a\vert x} \otimes N_{b\vert y}] = p(a,b\vert x,y)\\
        &\rho_{ABE} \in \mathcal{S}(\mathcal{H}_A \otimes \mathcal{H}_B\otimes \mathcal{H}_E).
    \end{aligned}
\end{equation}
In the device-dependent case, the optimization problem \eqref{eq:finite_dimensional_problem} can be solved numerically for certain conditional entropies $\mathbb{H}$, given the distribution $p(a,b\vert x,y)$ (cf.~\cite{he2025qicsquantuminformationconic,Komann2026,Hu_2022}). In the general device-independent case, however, the discrepancy between the two perspectives discussed in \autoref{subsec:quantum_games_intro} and \autoref{subsec:randomness_extraction_DIQKD} becomes apparent. If one adopts the viewpoint of DIQKD as a quantum game, as in \autoref{subsec:quantum_games_intro}, then the discussion in \autoref{subsec:randomness_extraction_DIQKD} shows that one must carefully address the definition of purification in cryptographic settings\footnote{We remark that this was considered in \cite[Lem.~1]{Berta_2015} for the setting of a leftover-hash lemma.} and its relation to a single-round program such as \eqref{eq:finite_dimensional_problem} within the commuting-operator framework. Conversely, if one assumes that the model is described by a finite-dimensional quantum system, then the purification issues discussed in \autoref{subsec:randomness_extraction_DIQKD} simplify considerably, but the application of the NPA hierarchy and related tools from the commuting-operator framework may no longer faithfully describe the experiment under consideration as discussed in \autoref{subsec:quantum_games_intro}. As we show in the following section, a careful combination of known tools and statements from operator algebras yields, in principle, a fully rigorous DIQKD framework in the language of the commuting-operator setting. However, the resulting optimization problems have so far only been solved by methods \cite{Tan_2021,Brown_2024,kossmann2025reliableentropyestimationobserved} that rely on the availability of Hilbert spaces and trace-class operators, which makes it difficult to conclude that one has obtained a genuinely complete DIQKD model. In the subsequent sections, we show how to resolve these issues completely.

\section{The commuting operator framework for DIQKD} \label{sec:commop}

Given the discussion of \emph{device independence} and its consequences from point of view of quantum games, i.e.~\autoref{subsec:quantum_games_intro}, and randomness extraction, i.e.~\autoref{subsec:randomness_extraction_DIQKD}, there is a need for a precise mathematical description of its meaning in the realm of DIQKD. Following the framework of quantum games, we consider \autoref{fig:diqkd_sketch} simply as an experiment whose internal workings inside the boxes may be modeled by quantum theory, without to specify any further. This question was already addressed in the early days of axiomatic quantum theory and was solved by von Neumann and Segal \cite{BaezSegalZhou1992,vonNeumann1996}. The solution is given by the universal $\mathcal{C}^\star$-algebra generated by the effects for each party, together with the maximal tensor product of these local algebras. In this realm, the standard Bell scenario is specified by measurement sets $\mathcal{X}$ and $\mathcal{Y}$ for Alice and Bob, respectively. Each measurement is assumed to have a finite outcome set, denoted $\mathcal{A}_x$ for $x \in \mathcal{X}$ and $\mathcal{B}_y$ for $y \in \mathcal{Y}$. We first model local measurements as positive operator-valued measures (POVMs), i.e.\ each setting $x$ corresponds to a finite family of positive contractions $\{M_{a\vert x}\}_{a\in\mathcal{A}_x}$ summing to the identity, and similarly for Bob. This motivates the following universal construction. The \emph{universal POVM algebra} associated with POVMs $M_{a\vert x}$, $a \in \mathcal{A}_x$ for each $x \in \mathcal{X}$ is the universal unital C$^\star$-algebra
\begin{align}\label{eq:def_universal_povm_algebra}
\mathcal U_{A} \; \coloneqq \; \mathcal{C}^\star\!\left(\{M_{a \vert x}\}_{x\in \mathcal{X},\,a\in\mathcal{A}_x}\;\Big|\;
0 \le M_{a \vert x} \le 1,\ \sum_{a\in \mathcal{A}_x} M_{a \vert x} = 1 \right).
\end{align}
Thus $\mathcal U_{A}$ is generated by symbols $\{M_{a\vert x}\}$ subject only to the relations that they form POVMs. The universal $\mathcal{C}^\star$-algebra is well-defined: the relations imply $\Vert M_{a\vert x}\Vert \le 1$ for all generators, and the defining relations are consistent (e.g.\ realized in a commutative model). Existence of the corresponding universal $\mathcal{C}^\star$-algebra for bounded generators is standard (see \cite[II.8.3.1]{Blackadar2006}).

Given the local algebras, one may ask what is the appropriate way of concatenating them, i.e., what is the correct “tensor product” between local algebras. As we aim to remain as general as possible, we argue that the universal property of the maximal tensor product provides the correct physical modeling (cf.~Prop.~\ref{prop:universal_property_maxtensorproduct}): every realization in a commuting-operator framework -- operationally satisfying the axioms of quantum theory -- is included in this tensor product \cite{Blackadar2006}. In particular, since it is known that there is a distinction between the minimal and maximal tensor products \cite{ji2022mipre}, we decide for the maximal tensor product as there exists a surjection from the maximal onto the minimal tensor product.
\begin{proposition}[Universal property of $\otimes_{\max}$ {\cite[Prop.~IV.4.7]{Takesaki1979}}]\label{prop:universal_property_maxtensorproduct}
Let $\mathcal{A},\mathcal{B}$ be $\mathcal{C}^\star$-algebras. Then for every $\mathcal{C}^\star$-algebra $\mathcal{C}$ and every pair of 
$\star$-homomorphisms $\pi_A:\mathcal{A} \to \mathcal{C}$, $\pi_B:\mathcal{B} \to \mathcal{C}$ with commuting ranges 
(i.e.\ $[\pi_A(a),\pi_B(b)]=0$ for all $a\in \mathcal{A}$, $b\in \mathcal{B}$), 
there exists a unique $\star$-homomorphism
\begin{align}
\Phi: \mathcal{A} \otimes_{\max} \mathcal{B} \longrightarrow \mathcal{C}
\end{align}
such that 
\begin{align}
\Phi(a \otimes b) = \pi_A(a)\pi_B(b) \qquad \forall a\in \mathcal{A}, b\in \mathcal{B}.
\end{align}
\end{proposition}
\begin{proof}
See \cite{Takesaki1979}, Prop.\ IV.4.7.
\end{proof}
Thus, the observable algebra for \autoref{fig:diqkd_sketch} can be summarized as
\begin{align}\label{eq:observable_algebras_diqkd}
    \mathcal{U}_{A} \otimes_{\operatorname{max}} \mathcal{U}_{B}.
\end{align}

\subsection{From POVMs to PVMs in the commuting-operator model}

With \eqref{eq:observable_algebras_diqkd}, we conclude that the DIQKD experiment can be described with the universal $\mathcal{C}^\star$-algebra of POVM's on Alice's and Bob's side and that they are concatenated by the maximal tensor product. 

Moreover, from the perspective of DIQKD, we are ultimately interested only in the statistics of such an experiment and aim to show a Naimark-type argument in order to replace the effect algebra by the universal $\mathcal{C}^\star$-algebra of projective measurements. This is particularly advantageous in applications, since, for instance, the Navascués--Pironio--Acín (NPA) hierarchy \cite{Navascus2007,Navascus2008} works much better with projections.

For each $x\in \mathcal{X}$ define the unital $\mathcal{C}^\star$-algebra $l^\infty\!\left(\mathcal{A}_x\right)$ to be generated by symbols $\{p_{a|x}\}_{a\in\mathcal{A}_x}$ subject to
\begin{align}
p_{a|x}p_{a'|x}=p_{a,a'}\,p_{a|x},\qquad \sum_{a\in\mathcal{A}_x}p_{a|x}=1.
\end{align}
The \emph{universal PVM algebra} is the unital free product \cite[II.8.3.4]{Blackadar2006}
\begin{align}\label{eq:universal_pvm_algebra}
\mathcal P_A \;=\; *_{\,x\in \mathcal{X}} \, l^\infty\!\left(\mathcal{A}_x\right),
\end{align}
so $\mathcal{P}_A$ is generated by families of projections which for each fixed $x$ form a projective measurement $\{p_{a|x}\}_{a \in \mathcal{A}_x}$, with no relations across different $x\neq x'\in \mathcal{X}$. We recall the corresponding proposition from quantum games. 
\begin{proposition}[Bell statistics are unchanged by passing from POVMs to PVMs]
\label{prop:Bell-POVM-PVM-equivalence}
Let $\mathcal{U}_{A}$ and $\mathcal{U}_{B}$ be the universal POVM algebras for Alice and Bob from \autoref{eq:def_universal_povm_algebra} and $\mathcal P_A$ and $\mathcal P_B$ be the universal PVM algebras from \autoref{eq:universal_pvm_algebra}. 

Then the set of probability distributions
\begin{align}
\Bigl\{
\bigl(\omega(M_{a|x}\otimes N_{b|y})\bigr)_{a,b,x,y}
:\ \omega\in S(\mathcal U_A\otimes_{\max}\mathcal U_B)
\Bigr\}
\end{align}
coincides with
\begin{align}
\Bigl\{
\bigl(\widetilde\omega(p_{a|x}\otimes q_{b|y})\bigr)_{a,b,x,y}
:\ \widetilde\omega\in S(\mathcal P_A\otimes_{\max}\mathcal P_B)
\Bigr\}.
\end{align}
\end{proposition}
\begin{proof}
    cf.~e.g.~\cite[Prop.~3.4.~(b)]{FRITZ_2012}.
\end{proof}
As discussed in the previous \autoref{subsec:randomness_extraction_DIQKD}, Prop.~\ref{prop:Bell-POVM-PVM-equivalence} is not enough for DIQKD. We need to investigate also the relation of the POVM vs. PVM model with respect to quantum side-information. 
\subsection{Quantum-side information and DIQKD in the commuting operator model}

The discussion in \autoref{subsec:randomness_extraction_DIQKD} shows that, in the commuting-operator framework, one cannot a priori model the adversary by adjoining a fixed Hilbert space factor $\mathcal{H}_E$ as in the finite-dimensional setting. The reason is that the observable algebra of Alice and Bob is now given abstractly by a universal $\mathcal{C}^\star$-algebra, and its concrete realizations need not come with a distinguished tensor-product decomposition. Consequently, the correct notion of quantum side information has to be formulated representation-theoretically.

The guiding principle is that, once a state $\psi \in \mathcal{S}(\mathcal{A}\otimes_{\operatorname{max}} \mathcal{B})$ describing the statistics between Alice and Bob has been fixed, the GNS construction provides a canonical representation of the observable algebra as a von Neumann algebra, while the vector state induced by $\xi_\psi$ is a normal state on the corresponding von Neumann algebra, and the commutant of the represented Alice--Bob algebra plays the role of the complementary system. Thus, we consider from now on an arbitrary but fixed state $\psi \in \mathcal{S}(\mathcal{A}\otimes_{\operatorname{max}} \mathcal{B})$ and its GNS-construction $(\pi_\psi,\mathcal{H}_\psi,\xi_\psi)$. Particulary, we fix the corresponding von Neumann algebra
\begin{align}
    \mathcal{M}_{AB,\psi} \coloneqq \pi_\psi(\mathcal{A}\otimes_{\operatorname{max}} \mathcal{B})'', 
\end{align}
which will be from now on the object under consideration with 
\begin{align}
    \psi(a\otimes b) = \langle \xi_\psi, \pi_{\psi}(a\otimes b) \xi_\psi\rangle, \quad a\otimes b \in \mathcal{A}\otimes_{\operatorname{max}} \mathcal{B},
\end{align}
whereby the vector $\xi_{\psi}$ is cyclic for the representation $\pi_\psi$. This implies that this description contains all information needed about the state $\psi$ considered as action on $\mathcal{A}\otimes_{\operatorname{max}} \mathcal{B}$. Indeed, the relative entropy, which will become central in the following, can be defined with this approach on $\mathcal{C}^\star$-algebras as discussed in \cite{Ohya1993}.

This leads to the following notion of purification, following the operator-algebraic framework of \cite{Berta_2015} (cf.~also~\cite[Lem.~4.11]{Takesaki1979}).

\begin{definition}[Purification in the commuting-operator framework]
\label{def:purification_commuting_operator}
Let $\mathcal{M}$ be a von Neumann algebra and let $\omega \in \mathcal{S}(\mathcal{M})$. A \emph{purification} of $\omega$ is a triple
\begin{align}
(\pi,\mathcal{H},\xi),
\end{align}
where $\pi:\mathcal{M} \to \mathcal{B}(\mathcal{H})$ is a representation and $\xi \in \mathcal{H}$ is a unit vector such that
\begin{align}
\omega(x)=\langle \xi,\pi(x)\xi\rangle,
\qquad x\in \mathcal{M}.
\end{align}
We call $\pi(\mathcal{M})$ the \emph{relevant system} and $\pi(\mathcal{M})'$ the \emph{complementary system} of the purification.
\end{definition}

In particular, if $(\pi_\omega,\mathcal{H}_\omega,\xi_\omega)$ denotes the GNS triple of $\omega$, then
\begin{align}
\omega(x)=\langle \xi_\omega,\pi_\omega(x)\xi_\omega\rangle,
\qquad x\in \mathcal{M},
\end{align}
and hence the GNS construction yields a canonical purification of $\omega$.
It is worth mentioning that with this discussion, the problematic situation of non-unique purifications from \cite{vanLuijk2026} disappears in the sense of entropic quantities as shown in \cite[Lem.~1]{Berta_2015}. Moreover, in contrast to the finite-dimensional tensor-product picture (cf.~\eqref{eq:purification_property_finite_dimensional_spaces}), this notion of purification is formulated entirely in terms of the represented algebra $\pi(\mathcal{M})$ and its commutant. Thus, the complementary system is not introduced externally, but is determined by the chosen representation of the observable algebra itself. The following theorem shows how this point of view can be connected with a PVM model $\mathcal{P}_A \otimes_{\operatorname{max}} \mathcal{P}_B$ instead of $\mathcal{U}_A \otimes_{\operatorname{max}} \mathcal{U}_B$.
\begin{theorem}[Dilation theorem]\label{thm:dilation_theorem}
    Let $\psi \in \mathcal{S}\!\left(\mathcal{U}_A \otimes_{\operatorname{max}} \mathcal{U}_B\right)$ be a state with GNS triple $\left(\pi_\psi, \mathcal{H}_\psi, \xi_\psi\right)$. Then there exists
    \begin{enumerate}
        \item[(a)] a Hilbert space $\mathcal{H}$
        \item[(b)] a unit vector $\xi \in \mathcal{H}$
        \item[(c)] a unital $\star$-representation $\Pi : \mathcal{P}_A \otimes_{\operatorname{max}} \mathcal{P}_B \to \mathcal{B}\!\left(\mathcal{H}\right)$
        \item[(d)] an injective normal unital $\star$-homomorphism $\iota : \pi_\psi \left(\mathcal{U}_A \otimes_{\operatorname{max}} \mathcal{U}_B\right)^{\prime} \to \Pi\left(\mathcal{P}_A \otimes_{\operatorname{max}} \mathcal{P}_B\right)^\prime $
    \end{enumerate}
    such that for all $x,y,a,b$ and $E \in \pi_\psi \left(\mathcal{U}_A \otimes_{\operatorname{max}} \mathcal{U}_B\right)^{\prime}$
    \begin{align}\label{eq:intertwining_identity}
        \langle \xi_\psi, M_{a\vert x}N_{b \vert y} E\,\xi_\psi\rangle = \langle \xi,\Pi\left(p_{a\vert x}\otimes q_{b\vert y}\right) \iota\! \left(E\right)\xi \rangle.
    \end{align}
    In particular, $\iota$ is a normal $\star$-isomorphism onto its image.
\end{theorem}
\begin{proof}
    Let $\left(\pi_\psi, \mathcal{H}_\psi, \xi_\psi\right)$ be the GNS representation of $\psi$, $\mathcal{M}_{AB} \coloneqq \pi_\psi \left(\mathcal{U}_A \otimes_{\operatorname{max}} \mathcal{U}_B\right)^{\prime \prime}$ and 
    \begin{align}
        \iota_0 : \mathcal{M}_{AB}^\prime \hookrightarrow \mathcal{B}\!\left(\mathcal{H}_\psi\right)
    \end{align}
    the embedding of the complementary system. Then $\iota_0$ is an injective normal $\star$-homomorphism by construction. With abuse of notation we denote by $\{A_{a\vert x}\}_{a,x}$ and $\{B_{b\vert y}\}_{b,y}$ the images of the POVM elements from $\mathcal{U}_A \otimes_{\operatorname{max}} \mathcal{U}_B$ in $\mathcal{B}\!\left(\mathcal{H}_\psi\right)$, which are again POVM's and satisfy $\left[A_{a\vert x}, B_{b\vert y}\right] = 0$ as well as $A_{a\vert x}, B_{b\vert y} \in \mathcal{M}_{AB}$. We define in the following 
    \begin{align}\label{eq:def_NA}
        \mathcal{N}_A \coloneqq \left(\{B_{b\vert y}\}_{b,y} \cup \mathcal{M}_{AB}^\prime \right)^{\prime \prime}
    \end{align}
    which is by construction a von Neumann algebra. Moreover, each $A_{a\vert x}$, $a\in \mathcal{A}_x$ and $x \in \mathcal{X}$, belongs to $\mathcal{N}_A^\prime$ and we thus define
    \begin{align}
        \phi_x: l^{\infty}\!\left(\mathcal{A}_x\right) \to \mathcal{N}_A^\prime, \quad p_{a\vert x} \mapsto A_{a\vert x}. 
    \end{align}
    Since $\{A_{a\vert x}\}_{a\in \mathcal{A}_x} \subseteq \mathcal{B}\!\left(\mathcal{H}_\psi\right)$ is a POVM, each $\phi_x$ is a unital and completely positive map (cf.~\cite[Prop.~3.1.~(b)]{FRITZ_2012}). By the free-product extension theorem for unital completely positive maps \cite[Thm.~3.1]{Davidson2018} (cf.~also~\cite{Boca1991}), there exists a unital completely positive map $\Phi_A:\mathcal{P}_A\to \mathcal{N}_A^\prime$ such that
\begin{align}\label{eq:def_PhiA}
\Phi_A(p_{a\vert x})=A_{a\vert x}, \quad x\in\mathcal X,\ a\in\mathcal A_x.
\end{align}
Now let $\left(\rho_A,\mathcal{H}_A,V_A\right)$ be a minimal Stinespring dilation of $\Phi_A$ \cite[Thm.~4.1]{Paulsen2003} (cf.~also~\cite{Stinespring1955}). Then we have $V_A^\star \rho_A(u)V_A = \Phi_A(u)$ for all $u \in \mathcal{P}_A$ and by minimality $\overline{\operatorname{span}}\bigl(\rho_A(\mathcal P_A)V_A\mathcal H_\psi\bigr)=\mathcal H_A$. Since $\Phi_A$ is unital, $V_A^\star V_A = \Phi_A\!\left(1\right) = 1$, i.e.~$V_A$ is an isometry. Furthermore, we observe that $\rho_A\left(\mathcal{P}_A\right) \subseteq \mathcal{B}\left(\mathcal{H}_A\right)$ is a unital $\mathcal{C}^\star$-subalgebra. By Arveson's theorem \cite[Thm.~12.7]{Paulsen2003} there exists for all $T \in \left(V_A^\star  \rho_A\left(\mathcal{P}_A\right) V_A\right)^\prime$ an element $T_1 \in \rho_A\left(\mathcal{P}_A\right)^\prime$ such that $V_AT=T_1V_A$. Indeed the map 
\begin{align}
\kappa_A:\left(V_A^\star  \rho_A\left(\mathcal{P}_A\right) V_A\right)^\prime \to \rho_A\left(\mathcal{P}_A\right)^\prime \cap \{V_AV_A^\star \}^\prime, T \mapsto \kappa_A(T) \coloneqq T_1 
\end{align}
is a normal $\star$-homormorphism (for normality check \cite[Thm.~1.3.1]{Arveson1969Subalgebras}). Since $\Phi_A\left(\mathcal{P}_A\right) \subseteq \mathcal{N}_A^\prime$, we have $V_A^\star  \rho_A\left(\mathcal{P}_A\right) V_A \subseteq \mathcal{N}_A^\prime.$ This implies $\mathcal{N}_A \subseteq \left(V_A^\star  \rho_A\left(\mathcal{P}_A\right) V_A\right)^\prime$ and thus there exists for each $T\in \mathcal{N}_A$ an operator $\kappa_A\left(T\right) \in \rho_A\left(\mathcal{P}_A\right)^\prime$ such that 
\begin{align}\label{eq:intertwine_VA}
    \kappa_A\left(T\right) V_A = V_A T.
\end{align}
Moreover, $\kappa_A$ is injective: if $\kappa_A\left(T\right) = 0$, then $V_AT = 0$ by \autoref{eq:intertwine_VA} and hence $T = V_A^\star V_A T = 0$, since $V_A$ is an isometry. We define $\xi_A \coloneqq V_A \xi_\psi \in \mathcal{H}_A$, which is a unit vector since $V_A$ is an isometry. Then we have for $T \in \mathcal{N}_A$, $x \in \mathcal{X}$, $a \in \mathcal{A}_x$
\begin{equation}\label{eq:alice_identity}
    \begin{aligned}
        \langle \xi_A, \rho_A\left(p_{a\vert x}\right) \kappa_A\left(T\right) \xi_A\rangle &= \langle \xi_A, \rho_A\left(p_{a\vert x}\right) \kappa_A\left(T\right) V_A \xi_\psi \rangle \\
        &=\langle \xi_\psi,V_A^\star \rho_A\left(p_{a\vert x}\right) V_A T \xi_\psi \rangle \\
        &= \langle \xi_\psi, \Phi_A\left(p_{a\vert x}\right) T\xi_\psi\rangle \\
        &=\langle \xi_\psi, A_{a\vert x} T\xi_\psi\rangle,
    \end{aligned}
\end{equation}
where we have used \autoref{eq:intertwine_VA} and \autoref{eq:def_PhiA}. We turn to Bob's side and define
\begin{align}\label{eq:def_NB}
    \mathcal{N}_B \coloneqq \left(\rho_A\!\left(\mathcal{P}_A\right) \cup \kappa_A\!\left(\mathcal{M}_{AB}^\prime\right)\right)^{\prime \prime} \subseteq \mathcal{B}\!\left(\mathcal{H}_A\right),
\end{align}
which is again a von Neumann algebra. Since $B_{b \vert y} \in \mathcal{N}_A$ for all $b,y$, the operators $\kappa_A\!\left(B_{b\vert y}\right)$ are well-defined and we claim $\kappa_A\!\left(B_{b\vert y}\right) \in \mathcal{N}_B^\prime$. Indeed, $\left[\kappa_A\!\left(B_{b\vert y}\right), \rho_A\!\left(p_{a\vert x}\right)\right] = 0$ since $\kappa_A\!\left(\mathcal{N}_A\right) \subseteq \rho_A\!\left(\mathcal{P}_A\right)^\prime$, and for $E \in \mathcal{M}_{AB}^\prime$ we have $\left[B_{b\vert y}, E\right] = 0$ since $B_{b\vert y} \in \mathcal{M}_{AB}$, whence
\begin{align}
    \left[\kappa_A\!\left(B_{b\vert y}\right), \kappa_A\!\left(E\right)\right] = \kappa_A\!\left(\left[B_{b\vert y}, E\right]\right) = 0,
\end{align}
because $\kappa_A$ is a $\star$-homomorphism. We thus define for each $y \in \mathcal{Y}$
\begin{align}
    \phi_y : l^\infty\!\left(\mathcal{B}_y\right) \to \mathcal{N}_B^\prime, \quad q_{b \vert y} \mapsto \kappa_A\!\left(B_{b\vert y}\right).
\end{align}
Since $\kappa_A$ is a unital $\star$-homomorphism, $\{\kappa_A\!\left(B_{b\vert y}\right)\}_{b \in \mathcal{B}_y}$ is again a POVM and each $\phi_y$ is a unital and completely positive map (cf.~\cite[Prop.~3.1.~(b)]{FRITZ_2012}). By the free-product extension theorem \cite[Thm.~3.1]{Davidson2018} (cf.~also~\cite{Boca1991}), there exists a unital completely positive map $\Phi_B : \mathcal{P}_B \to \mathcal{N}_B^\prime$ such that
\begin{align}\label{eq:def_PhiB}
    \Phi_B\!\left(q_{b\vert y}\right) = \kappa_A\!\left(B_{b\vert y}\right), \quad y \in \mathcal{Y},\ b \in \mathcal{B}_y.
\end{align}
Now let $\left(\rho_B, \mathcal{H}_B, V_B\right)$ be a minimal Stinespring dilation of $\Phi_B$. Then $V_B^\star \rho_B\!\left(v\right) V_B = \Phi_B\!\left(v\right)$ for all $v \in \mathcal{P}_B$ and $V_B$ is an isometry since $\Phi_B$ is unital. Since $\Phi_B\!\left(\mathcal{P}_B\right) \subseteq \mathcal{N}_B^\prime$, we have $V_B^\star \rho_B\!\left(\mathcal{P}_B\right) V_B \subseteq \mathcal{N}_B^\prime$ and hence $\mathcal{N}_B \subseteq \left(V_B^\star \rho_B\!\left(\mathcal{P}_B\right) V_B\right)^\prime$. By the same application of Arveson's theorem \cite[Thm.~12.7]{Paulsen2003} as above, there exists a normal $\star$-homomorphism
\begin{align}
    \kappa_B : \left(V_B^\star \rho_B\!\left(\mathcal{P}_B\right) V_B\right)^\prime \to \rho_B\!\left(\mathcal{P}_B\right)^\prime \cap \{V_BV_B^\star\}^\prime
\end{align}
such that for each $S \in \mathcal{N}_B$
\begin{align}\label{eq:intertwine_VB}
    \kappa_B\!\left(S\right) V_B = V_B S,
\end{align}
and $\kappa_B$ is injective by the same argument as for $\kappa_A$, since $V_B$ is an isometry. We now set
\begin{align}\label{eq:def_global}
    \mathcal{H} \coloneqq \mathcal{H}_B, \quad \xi \coloneqq V_B \xi_A = V_B V_A \xi_\psi, \quad \Pi_A \coloneqq \kappa_B \circ \rho_A, \quad \Pi_B \coloneqq \rho_B, \quad \iota \coloneqq \kappa_B \circ \kappa_A \circ \iota_0.
\end{align}
The vector $\xi$ is a unit vector since $V_A$ and $V_B$ are isometries. The compositions are well-defined: by \autoref{eq:def_NB} we have $\rho_A\!\left(\mathcal{P}_A\right) \subseteq \mathcal{N}_B$ and $\kappa_A\!\left(\mathcal{M}_{AB}^\prime\right) \subseteq \mathcal{N}_B$, both contained in the domain of $\kappa_B$. Furthermore, $\Pi_A$ and $\Pi_B$ are unital $\star$-representations with commuting ranges, since
\begin{align}
    \Pi_A\!\left(\mathcal{P}_A\right) = \kappa_B\!\left(\rho_A\!\left(\mathcal{P}_A\right)\right) \subseteq \rho_B\!\left(\mathcal{P}_B\right)^\prime = \Pi_B\!\left(\mathcal{P}_B\right)^\prime.
\end{align}
By the universal property of the maximal tensor product (cf.~\autoref{prop:universal_property_maxtensorproduct}) there exists a unique unital $\star$-representation
\begin{align}\label{eq:def_Pi}
    \Pi : \mathcal{P}_A \otimes_{\operatorname{max}} \mathcal{P}_B \to \mathcal{B}\!\left(\mathcal{H}\right), \quad \Pi\!\left(u \otimes v\right) = \Pi_A\!\left(u\right)\Pi_B\!\left(v\right),
\end{align}
which proves (a)--(c). For (d) we first observe that $\iota$ is a normal unital $\star$-homomorphism as a composition of such maps, and injective as a composition of injective maps. It remains to show $\iota\!\left(\mathcal{M}_{AB}^\prime\right) \subseteq \Pi\!\left(\mathcal{P}_A \otimes_{\operatorname{max}} \mathcal{P}_B\right)^\prime$, i.e.~commutation with $\Pi_A\!\left(\mathcal{P}_A\right)$ and $\Pi_B\!\left(\mathcal{P}_B\right)$. For $E \in \mathcal{M}_{AB}^\prime$ we have $\iota\!\left(E\right) = \kappa_B\!\left(\kappa_A\!\left(E\right)\right) \in \kappa_B\!\left(\mathcal{N}_B\right) \subseteq \rho_B\!\left(\mathcal{P}_B\right)^\prime = \Pi_B\!\left(\mathcal{P}_B\right)^\prime$. Moreover, $\left[\kappa_A\!\left(E\right), \rho_A\!\left(p_{a\vert x}\right)\right] = 0$ since $\kappa_A\!\left(\mathcal{N}_A\right) \subseteq \rho_A\!\left(\mathcal{P}_A\right)^\prime$, and therefore
\begin{align}
    \left[\iota\!\left(E\right), \Pi_A\!\left(p_{a\vert x}\right)\right] = \kappa_B\!\left(\left[\kappa_A\!\left(E\right), \rho_A\!\left(p_{a\vert x}\right)\right]\right) = 0,
\end{align}
because $\kappa_B$ is a $\star$-homomorphism. Since $\iota$ is an injective normal $\star$-homomorphism between von Neumann algebras, it is a normal $\star$-isomorphism onto its image. Finally, let $x \in \mathcal{X}$, $y \in \mathcal{Y}$, $a \in \mathcal{A}_x$, $b \in \mathcal{B}_y$ and $E \in \mathcal{M}_{AB}^\prime$. Then
\begin{equation}
    \begin{aligned}
        \langle \xi, \Pi\!\left(p_{a\vert x} \otimes q_{b\vert y}\right) \iota\!\left(c\right) \xi\rangle 
        &= \langle V_B\xi_A, \kappa_B\!\left(\rho_A\!\left(p_{a\vert x}\right)\right) \rho_B\!\left(q_{b\vert y}\right) \kappa_B\!\left(\kappa_A\!\left(E\right)\right) V_B \xi_A\rangle \\
        &= \langle V_B\,\rho_A\!\left(p_{a\vert x}\right)\xi_A,\, \rho_B\!\left(q_{b\vert y}\right) V_B\, \kappa_A\!\left(E\right) \xi_A\rangle \\
        &= \langle \rho_A\!\left(p_{a\vert x}\right)\xi_A,\, V_B^\star\rho_B\!\left(q_{b\vert y}\right) V_B\, \kappa_A\!\left(E\right) \xi_A\rangle \\
        &= \langle \xi_A,\, \rho_A\!\left(p_{a\vert x}\right) \kappa_A\!\left(B_{b\vert y}\right) \kappa_A\!\left(E\right) \xi_A\rangle \\
        &= \langle \xi_A,\, \rho_A\!\left(p_{a\vert x}\right) \kappa_A\!\left(B_{b\vert y}\,E\right) \xi_A\rangle \\
        &= \langle \xi_\psi,\, A_{a\vert x}\,B_{b\vert y}\,E\,\xi_\psi\rangle,
    \end{aligned}
\end{equation}
where the second equality uses \autoref{eq:intertwine_VB} for $\rho_A\!\left(p_{a\vert x}\right) \in \mathcal{N}_B$ and $\kappa_A\!\left(E\right) \in \mathcal{N}_B$, the fourth equality uses $V_B^\star \rho_B\!\left(q_{b\vert y}\right) V_B = \Phi_B\!\left(q_{b\vert y}\right)$ together with \autoref{eq:def_PhiB}, the fifth equality holds since $\kappa_A$ is a $\star$-homomorphism and $B_{b\vert y}\,E \in \mathcal{N}_A$, and the last equality is \autoref{eq:alice_identity} applied with $T = B_{b\vert y}\,E$. This proves \autoref{eq:intertwining_identity}.
\end{proof}

\subsection{Secrecy proofs}

As already emphasized in the \autoref{subsec:turning_experimental_data}, the genuinely quantum part of a security proof is privacy amplification: once Alice has produced a classical raw key variable from her key-generation measurement, the task is to quantify how much randomness can still be extracted against arbitrary quantum side information compatible with the observed data. In the present framework, such side information is not modeled by adjoining a fixed tensor factor, but rather through the commutant of a purification as in Def.~\ref{def:purification_commuting_operator}. 

For secrecy, fix a key-generation setting $\tilde{x}\in\mathcal X$, let $\mathcal Z\coloneqq \mathcal A_{\tilde{x}}$, and write $\{P_z\}_{z\in\mathcal Z}$ for the corresponding measurement on Alice's side. Let $(\pi,\mathcal H,\xi)$ be a purification of the relevant Alice--Bob state in the sense of Def.~\ref{def:purification_commuting_operator}, and set
\begin{align}
\mathcal M_E\coloneqq \pi(\mathcal M_{AB})'
\end{align}
for the adversary's system. The key-generation measurement is then described in the Heisenberg picture by the normal unital $^\ast$-homomorphism
\begin{align}
\mathsf M_Z:\ell^\infty(\mathcal Z)\to \pi(\mathcal M_{AB}),
\qquad
\mathsf M_Z(f)=\sum_{z\in\mathcal Z} f(z)\,P_z.
\end{align}
Since $\mathcal M_E$ commutes with $\pi(\mathcal M_{AB})$, this induces a normal state
\begin{align}
\omega_{ZE}\in \mathcal S\bigl(\ell^\infty(\mathcal Z)\,\overline\otimes\,\mathcal M_E\bigr)
\end{align}
given by
\begin{align}
\omega_{ZE}(f\otimes E)
=
\langle \xi,\mathsf M_Z(f)\,E\,\xi\rangle,
\qquad
f\in \ell^\infty(\mathcal Z),\ E\in \mathcal M_E.
\end{align}
Equivalently, if
\begin{align}
p(z)\coloneqq \langle \xi,P_z\xi\rangle,
\end{align}
then
\begin{align}
\omega_{ZE} = \sum_{z\in\mathcal Z} p(z)\,p_z\otimes \omega_E^z, \quad \omega_E^z(E)\coloneqq \frac{\langle \xi,P_zE\,\xi\rangle}{p(z)} \quad \text{for }p(z)>0,
\end{align}
so $\omega_{ZE}$ is precisely the classical--quantum post-measurement state of Alice's raw key variable and Eve's side information. The following definition is from \cite{Berta_2015}.

\begin{definition}[Secrecy for privacy amplification in the commuting-operator model]
\label{def:secrecy_commuting_operator}
Let $\mathcal{K}$ be a finite key alphabet and let
\begin{align}
    \omega_{KE}\in \mathcal S_{\le}\!\bigl(\ell^\infty(\mathcal{K})\,\overline\otimes\,\mathcal{M}_E\bigr)
\end{align}
be a classical--quantum state. We call $\omega_{KE}$ $\varepsilon$-secret if
\begin{align}\label{eq:secrecy_statement_commuting_operator}
    \Bigl\|
    \omega_{KE}
    -
    \frac{1}{|\mathcal K|}\tau_{\mathcal K}\otimes \omega_E
    \Bigr\|
    \le \varepsilon,
\end{align}
where $\tau_{\mathcal{K}}$ denotes the canonical trace on $\ell^\infty(\mathcal{K})$, $\omega_E$ is the marginal of $\omega_{KE}$ on $\mathcal{M}_E$.
\end{definition}

Def.~\ref{def:secrecy_commuting_operator} is the operator-algebraic analogue of the usual secrecy requirement, that is, the ideal key is the uniformly distributed classical state on $\ell^\infty(\mathcal K)$, independent of the adversary's observable algebra. In particular, once the abort event and the public classical transcript are included in the standard way, \eqref{eq:secrecy_statement_commuting_operator} reduces to the same privacy-amplification task that appears in the composable definitions of secrecy discussed in \autoref{subsec:turning_experimental_data}. 

The quantity governing privacy amplification is the smooth conditional min-entropy. Since $\ell^\infty(\mathcal Z)$ is finite-dimensional and canonically embeds as the diagonal algebra in $B(\mathbb C^{|\mathcal Z|})$, the definition \cite{Berta_2015} applies verbatim in our setting. For a classical--quantum state $\omega_{ZE}\in \mathcal S_{\le}(\ell^\infty(\mathcal Z)\,\overline\otimes\,\mathcal M_E)$ we define
\begin{align}\label{eq:conditional_min_entropy_commuting_operator}
    H_{\min}(Z|E)_\omega
    \coloneqq
    -\log
    \inf_{\sigma_E\in \mathcal N^+(\mathcal M_E)}
    \Bigl\{
        \sigma_E(1)
        :
        \tau_{\mathcal Z}\otimes \sigma_E \ge \omega_{ZE}
    \Bigr\},
\end{align}
where $\tau_{\mathcal Z}$ is the trace on $\ell^\infty(\mathcal Z)$. Writing
\begin{align}
    \mathcal B_\varepsilon(\omega_{ZE})
    \coloneqq
    \Bigl\{
        \widetilde\omega_{ZE}\in
        \mathcal S_{\le}\!\bigl(\ell^\infty(\mathcal Z)\,\overline\otimes\,\mathcal M_E\bigr)
        :
        P(\widetilde\omega_{ZE},\omega_{ZE})\le \varepsilon
    \Bigr\},
\end{align}
for the $\varepsilon$-ball with respect to the purified distance (cf.~\cite[Def.~7]{Berta_2015}), we set
\begin{align}\label{eq:smooth_min_entropy_commuting_operator}
    H_{\min}^\varepsilon(Z|E)_\omega
    \coloneqq
    \sup_{\widetilde\omega_{ZE}\in \mathcal B_\varepsilon(\omega_{ZE})}
    H_{\min}(Z|E)_{\widetilde\omega}.
\end{align}
Equivalently, in terms of the smooth max-relative entropy,
\begin{align}\label{eq:smooth_min_entropy_as_dmax}
    H_{\min}^\varepsilon(Z|E)_\omega
    =
    - \inf_{\sigma_E\in \mathcal S(\mathcal M_E)}
    D_{\max}^\varepsilon\!\bigl(
        \omega_{ZE}\,\big\|\,\tau_{\mathcal Z}\otimes \sigma_E
    \bigr).
\end{align}

At this point it is important that the quantity in \eqref{eq:smooth_min_entropy_commuting_operator} is well defined from the cryptographic point of view. For any purification, Alice's key measurement induces a post-measurement state on $ZB$, whose complementary system is the adversary's algebra. After embedding $\ell^\infty(\mathcal Z)$ diagonally into a matrix algebra, \cite[Lem.~1 and Lem.~3]{Berta_2015} imply that $H_{\min}^{\varepsilon}(Z\mid E)_\omega$ is independent of the chosen purification. Hence \eqref{eq:smooth_min_entropy_commuting_operator} is an intrinsic quantity of the commuting-operator model and thus the correct one-shot measure for privacy amplification.

The left-over-hash lemma \cite[Prop.~21]{Berta_2015} reduces secrecy to a lower bound on $H_{\min}^\varepsilon(Z|E)_\omega$ in the following sense. Let $\mathcal Z,\mathcal K$ be finite sets with $|\mathcal K|\le |\mathcal Z|$,
and let $\{\mathcal F,P_{\mathcal F}\}_{\mathcal Z,\mathcal K}$ be a two-universal family of hash functions $f:\mathcal Z\to \mathcal K$. For each $f\in\mathcal F$, let
\begin{align}
    T_f:\ell^\infty(\mathcal Z)\to \ell^\infty(\mathcal K)
\end{align}
denote the induced classical post-processing map. Then for every $\varepsilon\ge 0$,
\begin{align}\label{eq:privacy_amplification_commuting_operator}
    \mathbb E_{f\sim P_{\mathcal F}}
    \Bigl\|
        (T_f\otimes \id_E)(\omega_{ZE})
        -
        \frac{1}{|\mathcal K|}\tau_{\mathcal K}\otimes \omega_E
    \Bigr\|
    \le
    \sqrt{|\mathcal K|\,2^{-H_{\min}^\varepsilon(Z|E)_\omega}}
    +4\varepsilon.
\end{align}

In the asymptotic i.i.d.\ regime, this one-shot quantity \eqref{eq:privacy_amplification_commuting_operator} is governed by the conditional von Neumann entropy. To make this precise, recall that \cite{Fawzi_2025} formulates the asymptotic equipartition theorem in terms of the smooth max-relative entropy and the Araki relative entropy on general von Neumann algebras. For the present classical--quantum situation, this leads naturally to the definition
\begin{align}\label{eq:conditional_von_neumann_entropy_commuting_operator}
    H(Z|E)_\omega
    \coloneqq
    - \inf_{\sigma_E\in \mathcal S(\mathcal M_E)}
    D\!\bigl(
        \omega_{ZE}\,\big\|\,\tau_{\mathcal Z}\otimes \sigma_E
    \bigr).
\end{align}
Indeed, writing the classical--quantum state in block form as
\begin{align}
    \omega_{ZE}=\sum_{z\in \mathcal Z} p_z\otimes \omega_E^z,
    \qquad
    \omega_E=\sum_{z\in \mathcal Z}\omega_E^z,
\end{align}
the direct-sum structure of $\ell^\infty(\mathcal Z)\,\overline\otimes\,\mathcal M_E$
implies
\begin{align}
    D\!\bigl(\omega_{ZE}\,\big\|\,\tau_{\mathcal Z}\otimes \sigma_E\bigr)
    =
    \sum_{z\in\mathcal Z}
    D\!\bigl(\omega_E^z\,\big\|\,\sigma_E\bigr)
\end{align}
for every $\sigma_E\in \mathcal S(\mathcal M_E)$. Applying the corresponding identity for Araki relative entropy, which is a direct consequence of the definition, yields
\begin{align}
    \sum_{z\in\mathcal Z}
    D\!\bigl(\omega_E^z\,\big\|\,\sigma_E\bigr)
    =
    \sum_{z\in\mathcal Z}
    D\!\bigl(\omega_E^z\,\big\|\,\omega_E\bigr)
    +
    D\!\bigl(\omega_E\,\big\|\,\sigma_E\bigr).
\end{align}
Since $D(\omega_E\|\sigma_E)\ge 0$, with equality if and only if
$\sigma_E=\omega_E$, the infimum in
\eqref{eq:conditional_von_neumann_entropy_commuting_operator} is attained at
the marginal $\omega_E$. Therefore
\begin{align}
    H(Z|E)_\omega
    =
    -D\!\bigl(
        \omega_{ZE}\,\big\|\,\tau_{\mathcal Z}\otimes \omega_E
    \bigr).
\end{align}

Combining \cite[Prop.~21]{Berta_2015} with \cite[Thm.~1.1]{Fawzi_2025}, we arrive at the same conceptual conclusion as in the finite-dimensional discussion. Up to finite-size corrections, the secrecy of a DIQKD protocol is determined by a single-round optimization of a conditional von Neumann entropy over all commuting-operator models compatible with the observed statistics $p(a,b\vert x,y)$. The only difference is that the optimization now has to be carried out in the universal operator-algebraic model and with side information represented via the commutant of a purification.

More concretely, let $\mathcal{U}_A$ and $\mathcal{U}_B$ be the universal POVM algebras from \autoref{eq:def_universal_povm_algebra}, let $\tilde x\in \mathcal{X}$ be the key-generation input on Alice's side, and let $M_{z|\tilde x}\in \mathcal{U}_A$ denote the corresponding measurement operators. The asymptotic secrecy problem is then reduced to
\begin{equation}\label{eq:commuting_operator_entropy_program}
    \begin{aligned}
        \inf\ & H(Z|E)_{\omega_{ZE}} \\
        \text{s.t. }&
        \psi \in \mathcal{S}\!\left(\mathcal{U}_A\otimes_{\max}\mathcal{U}_B\right),\\
        &\psi\!\left(M_{a|x}\otimes N_{b|y}\right)=p(a,b|x,y),
        \qquad
        a,b,x,y,\\
        &(\pi,\mathcal{H},\xi)\ \text{is a purification of }\psi,\\
        &\mathcal{M}_E \coloneqq \pi\!\left(\mathcal{U}_A\otimes_{\max}\mathcal{U}_B\right)',\\
        &\omega_{ZE}\in
        \mathcal{S}\!\left(\ell^\infty(\mathcal{Z})\,\overline\otimes\,\mathcal{M}_E\right),\\
        &\omega_{ZE}(p_z\otimes E)
        =
        \left\langle \xi,\,
        \pi\!\left(M_{z|\tilde x}\otimes 1\right)E\,\xi
        \right\rangle,
        \qquad
        z\in\mathcal{Z},\ E\in \mathcal{M}_E.
    \end{aligned}
\end{equation}
Here $\{p_z\}_{z\in\mathcal{Z}}$ denotes the canonical basis of $\ell^\infty(\mathcal{Z})$. Comparing this with \eqref{eq:finite_dimensional_problem}, we conclude a similar structure in the commuting operator framework.

Moreover, there is no loss of generality in fixing the GNS purification of $\psi$. Indeed, for the one-shot quantities relevant to privacy amplification, purification-independence is exactly the content of \cite[Lem.~1 and Lem.~3]{Berta_2015}. Combining the argument in \cite[Lem.~3]{Berta_2015} and the asymptotic equipartition theorem \cite{Fawzi_2025} readily implies that also the conditional von Neumann entropy is invariant under the concrete purification. Hence the optimization over arbitrary purifications can be dropped, and one may, without loss of generality, work with the GNS representation throughout. Accordingly, if $(\pi_\psi,\mathcal{H}_\psi,\xi_\psi)$ denotes the GNS triple of $\psi$ and
\begin{align}
    \mathcal{M}_{AB,\psi}
    \coloneqq
    \pi_\psi\!\left(\mathcal{U}_A\otimes_{\max}\mathcal{U}_B\right)'',
\end{align}
then the asymptotic secrecy problem reduces to
\begin{equation}\label{eq:commuting_operator_entropy_program_gns_povm}
    \begin{aligned}
        \inf\ & H(Z|E)_{\omega_{ZE}} \\
        \text{s.t. }&
        \psi \in \mathcal{S}\!\left(\mathcal{U}_A\otimes_{\max}\mathcal{U}_B\right),\\
        &\psi\!\left(M_{a|x}\otimes N_{b|y}\right)=p(a,b|x,y),
        \qquad
        a,b,x,y,\\
        &\mathcal{M}_E \coloneqq \mathcal{M}_{AB,\psi}',\\
        &\omega_{ZE}\in
        \mathcal{S}\!\left(\ell^\infty(\mathcal{Z})\,\overline\otimes\,\mathcal{M}_E\right),\\
        &\omega_{ZE}(p_z\otimes E)
        =
        \left\langle \xi_\psi,\,
        \pi_\psi\!\left(M_{z|\tilde x}\otimes 1\right)E\,\xi_\psi
        \right\rangle,
        \qquad
        z\in\mathcal{Z},\ E\in \mathcal{M}_E.
    \end{aligned}
\end{equation}
Here $\{p_z\}_{z\in\mathcal{Z}}$ denotes again the canonical basis of $\ell^\infty(\mathcal{Z})$. Comparing this with \eqref{eq:finite_dimensional_problem}, we conclude that the commuting-operator framework leads to the same structural optimization problem, except that the underlying model is now formulated canonically in the GNS representation and the adversary's side information is encoded by the commutant.

Finally, the dilation theorem \autoref{thm:dilation_theorem} allows us to replace the universal POVM algebras by the universal PVM algebras without changing the value of the program. Consider
\begin{equation}\label{eq:commuting_operator_entropy_program_gns}
    \begin{aligned}
        \inf\ & H(Z|E)_{\omega_{ZE}} \\
        \text{s.t. }&
        \psi \in \mathcal{S}\!\left(\mathcal{P}_A\otimes_{\max}\mathcal{P}_B\right),\\
        &\psi\!\left(p_{a|x}\otimes q_{b|y}\right)=p(a,b|x,y),
        \qquad
        a,b,x,y,\\
        &\mathcal{M}_E \coloneqq \mathcal{M}_{AB,\psi}',\\
        &\omega_{ZE}\in
        \mathcal{S}\!\left(\ell^\infty(\mathcal{Z})\,\overline\otimes\,\mathcal{M}_E\right),\\
        &\omega_{ZE}(p_z\otimes E)
        =
        \left\langle \xi_\psi,\,
        \pi_\psi(p_{z|\tilde x})\,E\,\xi_\psi
        \right\rangle,
        \qquad
        z\in\mathcal{Z},\ E\in \mathcal{M}_E,
    \end{aligned}
\end{equation}
where now $\mathcal{M}_{AB,\psi} \coloneqq \pi_\psi\!\left(\mathcal{P}_A\otimes_{\max}\mathcal{P}_B\right)''$ and we abbreviate $\pi_\psi(p_{z|\tilde x}) \coloneqq \pi_\psi\!\left(p_{z|\tilde x}\otimes 1\right)$.

\begin{corollary}\label{cor:povm_pvm_equality}
    The optimal values of \eqref{eq:commuting_operator_entropy_program_gns_povm} and \eqref{eq:commuting_operator_entropy_program_gns} coincide.
\end{corollary}
\begin{proof}
    Denote the value of \eqref{eq:commuting_operator_entropy_program_gns_povm} by $h_{\operatorname{POVM}}$ and the value of \eqref{eq:commuting_operator_entropy_program_gns} by $h_{\operatorname{PVM}}$.

    We first show $h_{\operatorname{POVM}} \leq h_{\operatorname{PVM}}$. Since the canonical projections $\{p_{a|x}\}_{a\in\mathcal{A}_x} \subseteq \mathcal{P}_A$ form POVMs, the universal property of $\mathcal{U}_A$ yields a unital $\star$-homomorphism $\theta_A : \mathcal{U}_A \to \mathcal{P}_A$ with $\theta_A\!\left(M_{a|x}\right) = p_{a|x}$, and similarly $\theta_B : \mathcal{U}_B \to \mathcal{P}_B$ with $\theta_B\!\left(N_{b|y}\right) = q_{b|y}$. Composed with the canonical inclusions into $\mathcal{P}_A\otimes_{\max}\mathcal{P}_B$, the ranges commute, so by the universal property of the maximal tensor product (cf.~\autoref{prop:universal_property_maxtensorproduct}) there exists a unital $\star$-homomorphism
    \begin{align}
        \Theta : \mathcal{U}_A\otimes_{\max}\mathcal{U}_B \to \mathcal{P}_A\otimes_{\max}\mathcal{P}_B, \quad \Theta\!\left(M_{a|x}\otimes N_{b|y}\right) = p_{a|x}\otimes q_{b|y},
    \end{align}
    which is surjective, since it contains the generators $p_{a|x}\otimes q_{b|y}$. Now let $\left(\psi, \mathcal{M}_E, \omega_{ZE}\right)$ be feasible for \eqref{eq:commuting_operator_entropy_program_gns} and set $\psi' \coloneqq \psi\circ\Theta \in \mathcal{S}\!\left(\mathcal{U}_A\otimes_{\max}\mathcal{U}_B\right)$. Then $\psi'\!\left(M_{a|x}\otimes N_{b|y}\right) = p(a,b|x,y)$. Since $\Theta$ is surjective, $\xi_\psi$ is cyclic for $\pi_\psi\circ\Theta$, so by uniqueness of the GNS construction up to unitary equivalence \cite[II.~6.4.3]{Blackadar2006} $\left(\pi_\psi\circ\Theta, \mathcal{H}_\psi, \xi_\psi\right)$ is the GNS triple of $\psi'$, and
    \begin{align}
        \mathcal{M}_{AB,\psi'} = \left(\pi_\psi\circ\Theta\!\left(\mathcal{U}_A\otimes_{\max}\mathcal{U}_B\right)\right)'' = \pi_\psi\!\left(\mathcal{P}_A\otimes_{\max}\mathcal{P}_B\right)'' = \mathcal{M}_{AB,\psi},
    \end{align}
    again by density. In particular the commutants agree, so $\left(\psi', \mathcal{M}_E, \omega_{ZE}\right)$ is feasible for \eqref{eq:commuting_operator_entropy_program_gns_povm}, with the same classical--quantum state because $\pi_{\psi'}\!\left(M_{z|\tilde x}\otimes 1\right) = \pi_\psi\!\left(p_{z|\tilde x}\otimes 1\right)$, and hence with the same objective value. Taking the infimum yields $h_{\operatorname{POVM}} \leq h_{\operatorname{PVM}}$.

    We now show $h_{\operatorname{PVM}} \leq h_{\operatorname{POVM}}$. Let $\left(\psi, \mathcal{M}_E, \omega_{ZE}\right)$ be feasible for \eqref{eq:commuting_operator_entropy_program_gns_povm}. By \autoref{thm:dilation_theorem} there exist a Hilbert space $\mathcal{H}$, a unit vector $\xi \in \mathcal{H}$, a unital $\star$-representation $\Pi : \mathcal{P}_A\otimes_{\max}\mathcal{P}_B \to \mathcal{B}\!\left(\mathcal{H}\right)$ and an injective normal unital $\star$-homomorphism $\iota : \mathcal{M}_{AB,\psi}' \to \Pi\!\left(\mathcal{P}_A\otimes_{\max}\mathcal{P}_B\right)'$ such that
    \begin{align}\label{eq:dilation_in_corollary}
        \left\langle \xi_\psi,\, \pi_\psi\!\left(M_{a|x}\otimes N_{b|y}\right) E\, \xi_\psi \right\rangle = \left\langle \xi,\, \Pi\!\left(p_{a|x}\otimes q_{b|y}\right)\iota(E)\,\xi\right\rangle, \qquad a,b,x,y,\ E \in \mathcal{M}_{AB,\psi}'.
    \end{align}
    Define $\psi^{\operatorname{p}} \coloneqq \left\langle \xi, \Pi\!\left(\cdot\right)\xi\right\rangle \in \mathcal{S}\!\left(\mathcal{P}_A\otimes_{\max}\mathcal{P}_B\right)$. Evaluating \eqref{eq:dilation_in_corollary} at $E = 1$ gives $\psi^{\operatorname{p}}\!\left(p_{a|x}\otimes q_{b|y}\right) = p(a,b|x,y)$. Moreover, summing \eqref{eq:dilation_in_corollary} over $b \in \mathcal{B}_y$ with $\sum_b N_{b|y} = 1$ and $\sum_b \Pi\!\left(p_{a|x}\otimes q_{b|y}\right) = \Pi\!\left(p_{a|x}\otimes 1\right)$ yields, for $x = \tilde x$, $a = z$ and $E \in \mathcal{M}_E$,
    \begin{align}\label{eq:cq_equality_corollary}
        \omega_{ZE}\!\left(p_z\otimes E\right) = \left\langle \xi,\, \Pi\!\left(p_{z|\tilde x}\otimes 1\right)\iota(E)\,\xi\right\rangle \eqqcolon \omega_{ZE}^{\operatorname{p}}\!\left(p_z\otimes\iota(E)\right),
    \end{align}
    where $\omega_{ZE}^{\operatorname{p}} \in \mathcal{S}\!\left(\ell^\infty(\mathcal{Z})\,\overline\otimes\,\iota\!\left(\mathcal{M}_E\right)\right)$ is the classical--quantum state induced by the purification $\left(\Pi, \mathcal{H}, \xi\right)$ of $\psi^{\operatorname{p}}$ with side information $\iota\!\left(\mathcal{M}_E\right) \subseteq \Pi\!\left(\mathcal{P}_A\otimes_{\max}\mathcal{P}_B\right)'$. By \autoref{thm:dilation_theorem}, $\iota$ restricts to a normal $\star$-isomorphism from $\mathcal{M}_E$ onto $\iota\!\left(\mathcal{M}_E\right)$, so $\operatorname{id}\otimes\iota$ is a normal $\star$-isomorphism from $\ell^\infty(\mathcal{Z})\,\overline\otimes\,\mathcal{M}_E$ onto $\ell^\infty(\mathcal{Z})\,\overline\otimes\,\iota\!\left(\mathcal{M}_E\right)$. By \eqref{eq:cq_equality_corollary} we have $\omega_{ZE} = \omega_{ZE}^{\operatorname{p}}\circ\left(\operatorname{id}\otimes\iota\right)$ and, for the marginals, $\omega_E = \omega_E^{\operatorname{p}}\circ\iota$. Since the relative entropy is invariant under normal $\star$-isomorphisms, it follows that
    \begin{align}
        D\!\left(\omega_{ZE}\,\middle\|\,\tau_{\mathcal{Z}}\otimes\omega_E\right)
        =
        D\!\left(\omega_{ZE}^{\operatorname{p}}\,\middle\|\,\tau_{\mathcal{Z}}\otimes\omega_E^{\operatorname{p}}\right),
    \end{align}
    and hence $H(Z|E)_{\omega_{ZE}^{\operatorname{p}}} = H(Z|E)_{\omega_{ZE}}$. The tuple $\left(\psi^{\operatorname{p}}, \iota\!\left(\mathcal{M}_E\right), \omega_{ZE}^{\operatorname{p}}\right)$ is a feasible point with respect to the purification $\left(\Pi, \mathcal{H}, \xi\right)$, which need not be the GNS purification of $\psi^{\operatorname{p}}$.

    \rs{Let \(\Omega_{Z\mathcal N_E}\in
\mathcal S(\ell^\infty(\mathcal Z)\,\overline\otimes\,\mathcal N_E)\) denote
the classical--quantum state induced by the same purification
\((\Pi,\mathcal H,\xi)\), where
\(\mathcal N_E:=\Pi(\mathcal P_A\otimes_{\max}\mathcal P_B)'\). Then
\(\omega_{ZE}^{\operatorname p}\) is the restriction of
\(\Omega_{Z\mathcal N_E}\) to
\(\ell^\infty(\mathcal Z)\,\overline\otimes\,\iota(\mathcal M_E)\). Hence,
by monotonicity of relative entropy under this restriction, enlarging Eve's
algebra from \(\iota(\mathcal M_E)\) to \(\mathcal N_E\) can only decrease
the conditional entropy,
\[
    H(Z|\mathcal N_E)_{\Omega}
    \leq
    H(Z|\iota(\mathcal M_E))_{\omega^{\operatorname p}}
    =
    H(Z|E)_{\omega_{ZE}} .
\]
By purification independence of the conditional entropy, the left-hand side
agrees with the value obtained from the GNS purification of
\(\psi^{\operatorname p}\).}
    Let $\left(\psi^{\operatorname{p}}, \mathcal{N}, \omega_{ZE}'\right)$ be a feasible point of \eqref{eq:commuting_operator_entropy_program_gns} with
    \begin{align}
        H(Z|E)_{\omega_{ZE}'} \leq H(Z|E)_{\omega_{ZE}^{\operatorname{p}}} = H(Z|E)_{\omega_{ZE}}.
    \end{align}
    Taking the infimum over all feasible points of \eqref{eq:commuting_operator_entropy_program_gns_povm} yields $h_{\operatorname{PVM}} \leq h_{\operatorname{POVM}}$, which completes the proof.
\end{proof}

We close this section with a simple observation regarding the range of the optimization problem \eqref{eq:commuting_operator_entropy_program_gns}. 
\begin{lemma}\label{lem:domination_and_entropy_bounds_cq}
Let
\begin{align}
    \omega_{ZE}\in \mathcal S\!\bigl(\ell^\infty(\mathcal Z)\,\overline\otimes\,\mathcal M_E\bigr)
\end{align}
be a classical--quantum state, and let $\omega_E$ denote its marginal on
$\mathcal M_E$. Then
\begin{align}\label{eq:domination_cq_state}
    \omega_{ZE}\le \tau_{\mathcal Z}\otimes \omega_E.
\end{align}
In particular,
\begin{align}\label{eq:conditional_entropy_bounds_cq}
    0\le H(Z|E)_\omega \le \log |\mathcal Z|.
\end{align}
\end{lemma}

\begin{proof}
Write $\omega_{ZE}$ in block form as
\begin{align}
    \omega_{ZE}=\sum_{z\in\mathcal Z} p_z\otimes \omega_E^z,
    \qquad
    \omega_E=\sum_{z\in\mathcal Z}\omega_E^z,
\end{align}
with $\omega_E^z\in \mathcal N^+(\mathcal M_E)$. Let
\begin{align}
    x=\sum_{z\in\mathcal Z} p_z\otimes c_z
    \in
    \bigl(\ell^\infty(\mathcal Z)\,\overline\otimes\,\mathcal M_E\bigr)_+
\end{align}
with $c_z\in \mathcal M_E{}_+$. Since $\omega_E^z\le \omega_E$ for every
$z\in\mathcal Z$, we obtain
\begin{align}
    \omega_{ZE}(x)
    =
    \sum_{z\in\mathcal Z}\omega_E^z(c_z)
    \le
    \sum_{z\in\mathcal Z}\omega_E(c_z)
    =
    (\tau_{\mathcal Z}\otimes \omega_E)(x),
\end{align}
which proves \eqref{eq:domination_cq_state}.

By the discussion above,
\begin{align}
    H(Z|E)_\omega
    =
    -D\!\bigl(\omega_{ZE}\,\big\|\,\tau_{\mathcal Z}\otimes \omega_E\bigr).
\end{align}
Using \eqref{eq:domination_cq_state} together with
Lemma~\ref{lem:araki_order_scaling}(i), we get
\begin{align}
    D\!\bigl(\omega_{ZE}\,\big\|\,\tau_{\mathcal Z}\otimes \omega_E\bigr)
    \le
    D\!\bigl(\omega_{ZE}\,\big\|\,\omega_{ZE}\bigr)
    =0,
\end{align}
and hence
\begin{align}
    H(Z|E)_\omega\ge 0.
\end{align}

For the upper bound, let
\begin{align}
    u_{\mathcal Z}\coloneqq \frac{1}{|\mathcal Z|}\tau_{\mathcal Z}.
\end{align}
Then $u_{\mathcal Z}\otimes \omega_E$ is a state, so by nonnegativity of the
Araki relative entropy,
\begin{align}
    D\!\bigl(\omega_{ZE}\,\big\|\,u_{\mathcal Z}\otimes \omega_E\bigr)\ge 0.
\end{align}
Since $\tau_{\mathcal Z}=|\mathcal Z|\,u_{\mathcal Z}$ and $\omega_{ZE}$ is a
state, Lemma~\ref{lem:araki_order_scaling}(ii) yields
\begin{align}
    D\!\bigl(\omega_{ZE}\,\big\|\,\tau_{\mathcal Z}\otimes \omega_E\bigr)
    =
    D\!\bigl(\omega_{ZE}\,\big\|\,u_{\mathcal Z}\otimes \omega_E\bigr)
    -\log |\mathcal Z|.
\end{align}
Therefore
\begin{align}
    D\!\bigl(\omega_{ZE}\,\big\|\,\tau_{\mathcal Z}\otimes \omega_E\bigr)
    \ge -\log |\mathcal Z|,
\end{align}
or equivalently,
\begin{align}
    H(Z|E)_\omega \le \log |\mathcal Z|.
\end{align}
This proves \eqref{eq:conditional_entropy_bounds_cq}.
\end{proof}

\section{Operator layer cake on von Neumann algebras}\label{sec:operator_layer_cake}

In order to solve \eqref{eq:commuting_operator_entropy_program_gns}, we show in the following the main theorem of this work, which is an integral formula for the relative entropy introduced by Frenkel \cite{Frenkel2023} on von Neumann algebras. Given two positive semidefinite matrices $\rho,\sigma \in \mathbb{M}_n(\mathbb{C})$, \cite{Frenkel2023} shows
\begin{align}\label{eq:frenkel_matrices}
    D(\rho \Vert \sigma) = \tr[\rho - \sigma] + \int_{\mathbb{R}} \frac{dt}{\vert t\vert (t-1)^2} \tr^-[(1-t)\rho + t\sigma],
\end{align}
whereby the $\tr^-[\cdot]$ denotes the negative part of an hermitian matrix argument. In subsequent work it has been shown that the formula \eqref{eq:frenkel_matrices} can be seen as a special type of $f$-divergence, which are commonly defined for a twice continuously-differentiable function $f$ with $f(1) = 0$ as
\begin{align}
    D_f(\rho \Vert \sigma)
    \coloneqq \int_1^\infty f^{\prime \prime}(s) \, \tr^+[\rho - s\sigma] \, ds + \int_1^\infty s^{-3} f^{\prime \prime}(s^{-1}) \, \tr^+[\sigma - s\rho] \, ds.
\end{align}
For the specific function $f(s) = s\log s$, a direct integration by parts argument with \eqref{eq:frenkel_matrices} yields
\begin{align}
    D(\rho \Vert \sigma) \equiv D_{s\log s}(\rho\Vert \sigma) = \int_1^\infty \frac{ds}{s} \tr^+[\rho - s\sigma] + \int_1^\infty \frac{ds}{s^2}\tr^+[\sigma - s\rho].
\end{align}
As the goal is to generalize the formula from \eqref{eq:frenkel_matrices}, we introduce the main ingredient for this first, which is the \emph{operator-layer cake} formula, for matrices presented in \cite{liu2025layercakerepresentationsquantum}. The operator layer-cake formula is given by the following expression for the G\^{a}teaux derivative of the logarithm for a matrix $A>0$ and $B=B^\star$
\begin{align}\label{eq:operator_layer_cake_matrices}
\operatorname{D}\log(A)[B]
=
\int_{0}^{\infty} [B-uA>0] \, du
-
\int_{-\infty}^{0} [uA-B>0]\,du,
\end{align}
whereby $[H>0]$ means the unique projection on the positive part of $H$. The matter of the following theorem is to generalize this formula to a separable Hilbert space $\mathcal{H}$.

\begin{theorem}[Operator layer cake]\label{thm:operator_layer_cake}
Let $\mathcal{H}$ be a separable Hilbert space and let $A,B\in\mathcal{B}(\mathcal{H})$ with $A\geq \delta I$ for some $\delta > 0$ and $B=B^\ast$.
Then the following identity holds in strong operator topology:
\begin{align}\label{eq:operator_layer_cake}
\operatorname{D}\log(A)[B]
=
\int_{0}^{\infty} H(B-uA) \, du
-
\int_{-\infty}^{0} \left(I-H(B-uA)\right)du,
\end{align}
where $H(\cdot)$ denotes the Heaviside function defined in \eqref{eq:def_heaviside} below. 
\end{theorem}

\begin{proof}
We define the Heaviside function as
\begin{equation}\label{eq:def_heaviside}
    \begin{aligned}
        H\!\left(x\right) \coloneqq \begin{cases}
            0 & x<0 \\
            \frac{1}{2} & x = 0\\
            1 & x>0.
        \end{cases}
    \end{aligned}
\end{equation}

Fix $\varepsilon>0$ and define bounded continuous functions
\begin{equation}\label{eq:feps_def}
\begin{aligned}
f_{\varepsilon}^+(x)
&\coloneqq  \frac{1}{2}\left(1+\frac{2}{\pi}\arctan\!\left(\frac{x}{\varepsilon}\right)\right),
\\
f_{\varepsilon}^-(x)
&\coloneqq  \frac{1}{2}\left(1-\frac{2}{\pi}\arctan\!\left(\frac{x}{\varepsilon}\right)\right)
=1-f_\varepsilon^+(x).
\end{aligned}
\end{equation}
Then $0\le f_\varepsilon^\pm\le 1$, and $f_\varepsilon^+(x)\to H\!\left(x\right)$, $f_\varepsilon^-(x)\to 1-H\!\left(x\right)$ pointwise as $\varepsilon\downarrow 0$, including at $x = 0$, where $f_\varepsilon^\pm(0) = \frac{1}{2} = H(0)$.

For $r>0$ set
\begin{align}\label{eq:T_eps_r_def}
T_{\varepsilon,r}
\coloneqq
\int_0^{r} f_\varepsilon^+(B-uA)\,du
-
\int_{-r}^{0} f_\varepsilon^-(B-uA)\,du.
\end{align}
Since $u\mapsto f_\varepsilon^\pm(B-uA)$ is continuous and uniformly bounded by $1$, the integrals exist as Bochner integrals in norm topology.

Moreover, for each fixed $u$ and $\xi \in \mathcal{H}$, the pointwise convergence $f_\varepsilon^\pm \to H, 1-H$ on $\mathbb{R}$ together with dominated convergence with respect to the spectral measure of $B-uA$ yields
\begin{align}
f_\varepsilon^+(B-uA)\xi \xrightarrow[\varepsilon \downarrow 0]{} H(B-uA)\xi,
\qquad
f_\varepsilon^-(B-uA)\xi \xrightarrow[\varepsilon \downarrow 0]{} \left(I - H(B-uA)\right)\xi.
\end{align}
Applying dominated convergence once more, now in the variable $u$ with the integrable bound $\Vert f_\varepsilon^\pm(B-uA)\xi\Vert \leq \Vert \xi\Vert$ on the compact interval $[-r,r]$, we obtain the strong limit
\begin{align}\label{eq:T_eps_to_proj}
T_{\varepsilon,r}\xrightarrow[\varepsilon\downarrow 0]{\rm s.o.t.}
\int_0^{r} H(B-uA)\,du
-
\int_{-r}^{0} \left(1-H(B-uA)\right)du.
\end{align}
Note that $f_\varepsilon^{\pm}$ cannot converge uniformly to $H, 1-H$, since a uniform limit of continuous functions is continuous; the convergence therefore only holds strongly.

Using the identity
\begin{align}
\arctan(t)=\frac{1}{2i}\left(\log(1+it)-\log(1-it)\right),
\end{align}
one checks that for real $x$
\begin{align}\label{eq:arctan_log_identity}
f_\varepsilon^+(x)
=
\frac{1}{2}
+\frac{1}{2\pi i}\left(\log(\varepsilon+ix)-\log(\varepsilon-ix)\right),
\qquad
f_\varepsilon^-(x)=1-f_\varepsilon^+(x).
\end{align}
By continuous functional calculus this lifts to selfadjoint operators, and in particular (using \eqref{eq:def_operator_log})
\begin{align}\label{eq:feps_log_operator}
f_\varepsilon^+(X)
=
\frac{1}{2}I
+\frac{1}{2\pi i}\left(\log(\varepsilon I+iX)-\log(\varepsilon I-iX)\right),
\qquad X=X^\ast.
\end{align}
Insert $X=B-uA$ and use that the $\frac{1}{2} I$ contributions cancel in \eqref{eq:T_eps_r_def}, since both domains of integration have length $r$, to obtain
\begin{align}\label{eq:T_eps_r_as_log_integrals}
T_{\varepsilon,r}
=
\frac{1}{2\pi i}
\left(
\int_{-r}^{r}\log\!\left(\varepsilon I+i(B-uA)\right)\,du
-
\int_{-r}^{r}\log\!\left(\varepsilon I-i(B-uA)\right)\,du
\right).
\end{align}

Consider the operator-valued functions of a complex parameter $z$,
\begin{align}
F_+(z)\coloneqq \log\!\left(\varepsilon I+i(B-zA)\right),
\qquad
F_-(z)\coloneqq \log\!\left(\varepsilon I-i(B-zA)\right).
\end{align}
We claim that $F_+$ is holomorphic on the strip
\begin{align}
\Sigma_+\coloneqq \left\{z\in\mathbb{C}:\ \Im z>-\frac{\varepsilon}{2\Vert A\Vert}\right\},
\end{align}
and $F_-$ is holomorphic on
\begin{align}
\Sigma_-\coloneqq \left\{z\in\mathbb{C}:\ \Im z<\frac{\varepsilon}{2\Vert A\Vert}\right\}.
\end{align}
Indeed, for $z=u+iv$ we have
\begin{align}
\varepsilon I+i(B-zA)=\varepsilon I+i(B-uA)+vA,
\end{align}
so
\begin{align}
\Re\!\left(\varepsilon I+i(B-zA)\right)=\varepsilon I+vA.
\end{align}
Since $A\ge 0$, we distinguish two cases. If $v\ge 0$, then $vA\ge 0$, hence
\begin{align}
\Re\!\left(\varepsilon I+i(B-zA)\right)\ge \varepsilon I\ge \frac{\varepsilon}{2}I.
\end{align}
If $-\varepsilon/(2\Vert A\Vert)<v<0$, then $A\le \Vert A\Vert I$ and therefore, because $v<0$,
\begin{align}
vA\ge v\Vert A\Vert I,
\end{align}
so
\begin{align}
\Re\!\left(\varepsilon I+i(B-zA)\right)
=\varepsilon I+vA
\ge (\varepsilon+v\Vert A\Vert)I
> \frac{\varepsilon}{2}I.
\end{align}
Hence for all $z$ with $\Im z=v>-\varepsilon/(2\Vert A\Vert)$ we have
\begin{align}
\Re\!\left(\varepsilon I+i(B-zA)\right)\ge \frac{\varepsilon}{2}I.
\end{align}
In particular,
\begin{align}
\sigma\!\left(\varepsilon I+i(B-zA)\right)\subset \{w\in\mathbb{C}:\Re w\ge \varepsilon/2\}\subset \mathbb{C}\setminus(-\infty,0],
\end{align}
so the principal-branch logarithm is well-defined and $F_+$ is holomorphic on $\Sigma_+$ by the holomorphic functional calculus.

Similarly,
\begin{align}
\varepsilon I-i(B-zA)=\varepsilon I-i(B-uA)-vA,
\end{align}
hence
\begin{align}
\Re\!\left(\varepsilon I-i(B-zA)\right)=\varepsilon I-vA.
\end{align}
If $v\le 0$, then $-vA\ge 0$ and thus
\begin{align}
\Re\!\left(\varepsilon I-i(B-zA)\right)\ge \varepsilon I\ge \frac{\varepsilon}{2}I.
\end{align}
If $0<v<\varepsilon/(2\Vert A\Vert)$, then $A\le \Vert A\Vert I$ implies
\begin{align}
-vA\ge -v\Vert A\Vert I,
\end{align}
and therefore
\begin{align}
\Re\!\left(\varepsilon I-i(B-zA)\right)
=\varepsilon I-vA
\ge (\varepsilon-v\Vert A\Vert)I
> \frac{\varepsilon}{2}I.
\end{align}
Hence for all $z$ with $\Im z=v<\varepsilon/(2\Vert A\Vert)$ we have
\begin{align}
\Re\!\left(\varepsilon I-i(B-zA)\right)\ge \frac{\varepsilon}{2}I,
\end{align}
which implies
\begin{align}
\sigma\!\left(\varepsilon I-i(B-zA)\right)\subset \{w\in\mathbb{C}:\Re w\ge \varepsilon/2\}\subset \mathbb{C}\setminus(-\infty,0].
\end{align}
Therefore the principal-branch logarithm is well-defined and $F_-$ is holomorphic on $\Sigma_-$.

Let $C_{r,+}\coloneqq \{re^{i\theta}:0\le \theta\le \pi\} \subseteq \Sigma_+$ and $C_{r,-}\coloneqq \{re^{i\theta}:-\pi\le \theta\le 0\} \subseteq \Sigma_-$.
By holomorphy and Cauchy's theorem applied to the closed contours formed by $[-r,r]\cup C_{r,+}$ inside $\Sigma_+$ and $[r,-r]\cup C_{r,-}$ inside $\Sigma_-$, we obtain
\begin{align}\label{eq:cauchy_deform}
\int_{-r}^{r}F_+(u)\,du
=
-\int_{C_{r,+}}F_+(z)\,dz,
\qquad
-\int_{-r}^{r}F_-(u)\,du
=
-\int_{C_{r,-}}F_-(z)\,dz.
\end{align}
Substituting into \eqref{eq:T_eps_r_as_log_integrals} yields
\begin{align}\label{eq:T_eps_r_arcs}
T_{\varepsilon,r}
=
\frac{1}{2\pi i}
\left(
-\int_{C_{r,+}}F_+(z)\,dz
-\int_{C_{r,-}}F_-(z)\,dz
\right).
\end{align}

Now invoke $A\ge \delta I$, and choose $r> \Vert B\Vert/\delta$ (fixed). Then $A$ is invertible and for every $z$ with $|z|=r$ we can write
\begin{align}
B-zA
=
-zA\left(I-z^{-1}A^{-1}B\right).
\end{align}
Since
\begin{align}
\left\Vert z^{-1}A^{-1}B\right\Vert
\le \frac{1}{r}\,\Vert A^{-1}\Vert\,\Vert B\Vert
\le \frac{\Vert B\Vert}{r\delta}
<1,
\end{align}
the operator $I-z^{-1}A^{-1}B$ is invertible by the Neumann series argument, hence $B-zA$ is invertible for all $|z|=r$. In particular,
\begin{align}
i(B-zA)\ \text{and}\ -i(B-zA)
\end{align}
are invertible on $C_{r,+}$ and $C_{r,-}$, respectively.

Moreover, if $z=u+iv\in C_{r,+}$, then $v\ge 0$ and
\begin{align}
\Re\!\left(i(B-zA)\right)=vA\ge 0.
\end{align}
Hence
\begin{align}
\sigma\!\left(i(B-zA)\right)\subset \{w\in\mathbb{C}:\Re w\ge 0\}.
\end{align}
Since $i(B-zA)$ is invertible and $z\mapsto (B-zA)^{-1}$ is continuous on the compact arc $C_{r,+}$, there exists $\eta_+>0$ such that
\begin{align}
\sigma\!\left(i(B-zA)\right)\subset \{w\in\mathbb{C}:\Re w\ge 0,\ |w|\ge \eta_+\}
\qquad(z\in C_{r,+}).
\end{align}
Similarly, for $z=u+iv\in C_{r,-}$ we have $v\le 0$ and
\begin{align}
\Re\!\left(-i(B-zA)\right)=-vA\ge 0,
\end{align}
and there exists $\eta_->0$ such that
\begin{align}
\sigma\!\left(-i(B-zA)\right)\subset \{w\in\mathbb{C}:\Re w\ge 0,\ |w|\ge \eta_-\}
\qquad(z\in C_{r,-}).
\end{align}

Set $\eta\coloneqq \min\{\eta_+,\eta_-\}>0$. For $w \in \mathbb{C}$ with $\Re w \geq 0$ and $\varepsilon \geq 0$ we have $|w+\varepsilon|^2 = |w|^2 + 2\varepsilon\Re w + \varepsilon^2 \geq |w|^2$, so the shift by $\varepsilon$ does not decrease the distance to the branch cut. Hence for all $\varepsilon\ge 0$,
\begin{align}
\sigma\!\left(\varepsilon I+i(B-zA)\right)
\subset \{w\in\mathbb{C}:\Re w\ge 0,\ |w|\ge \eta\}
\subset \mathbb{C}\setminus(-\infty,0]
\qquad(z\in C_{r,+}),
\end{align}
and
\begin{align}
\sigma\!\left(\varepsilon I-i(B-zA)\right)
\subset \{w\in\mathbb{C}:\Re w\ge 0,\ |w|\ge \eta\}
\subset \mathbb{C}\setminus(-\infty,0]
\qquad(z\in C_{r,-}).
\end{align}
Therefore the principal logarithm is defined also at $\varepsilon=0$, and by continuity of the holomorphic functional calculus uniformly on the compact arcs we have
\begin{align}
F_+(z)\longrightarrow \log\!\left(i(B-zA)\right)
\quad \text{for} \quad \varepsilon\to 0 \quad \text{uniformly in operator norm on }C_{r,+},
\end{align}
and
\begin{align}
F_-(z)\longrightarrow \log\!\left(-i(B-zA)\right) \quad \text{for} \quad \varepsilon\to 0 \quad
\text{uniformly in operator norm on }C_{r,-}.
\end{align}
Hence we may pass to the limit $\varepsilon\rightarrow 0$ in \eqref{eq:T_eps_r_arcs}, in operator norm, obtaining
\begin{align}\label{eq:T_0_r_arcs_1}
\lim_{\varepsilon\downarrow 0}T_{\varepsilon,r}
=
\frac{1}{2\pi i}
\left(
-\int_{C_{r,+}}\log\!\left(i(B-zA)\right)\,dz
-\int_{C_{r,-}}\log\!\left(-i(B-zA)\right)\,dz
\right).
\end{align}

We next simplify the right-hand side of \eqref{eq:T_0_r_arcs_1}.
For $|z|=r$, write $z=re^{i\theta}$. Since the choice of $r$ implies
\begin{align}
rI>A^{-1/2}BA^{-1/2}>-rI,
\end{align}
we have
\begin{align}
\Re(A-z^{-1}B)
&=
A-r^{-1}\cos(\theta)B \notag\\
&=
A^{1/2}\left(I-r^{-1}\cos(\theta)A^{-1/2}BA^{-1/2}\right)A^{1/2}>0 .
\end{align}
Hence, the spectrum of $A-z^{-1}B$ lies in the open right half-plane. Moreover, on $C_{r,+}$, the scalar $-zi$ lies in the closed right half-plane, while on $C_{r,-}$, the scalar $zi$ lies in the closed right half-plane. Therefore, by the scalar principal-logarithm identity and holomorphic functional calculus,
\begin{align}
\log\!\left(i(B-zA)\right)
&=
\log(-zi)I+\log(A-z^{-1}B),
\qquad z\in C_{r,+},
\label{eq:upper_arc_log_decomp}
\\
\log\!\left(-i(B-zA)\right)
&=
\log(zi)I+\log(A-z^{-1}B),
\qquad z\in C_{r,-}.
\label{eq:lower_arc_log_decomp}
\end{align}

Using \eqref{eq:upper_arc_log_decomp} and \eqref{eq:lower_arc_log_decomp}, we obtain
\begin{align}
&\int_{C_{r,+}}\log\!\left(i(B-zA)\right)\,dz
+
\int_{C_{r,-}}\log\!\left(-i(B-zA)\right)\,dz
\notag\\
&=
\int_{C_{r,+}}\log(A-z^{-1}B)\,dz
+
\int_{C_{r,-}}\log(A-z^{-1}B)\,dz
+
\int_{C_{r,+}}\log(-zi)I\,dz
+
\int_{C_{r,-}}\log(zi)I\,dz .
\end{align}
The scalar logarithmic terms cancel by the substitution $w=-z$:
\begin{align}
\int_{C_{r,+}}\log(-zi)I\,dz
+
\int_{C_{r,-}}\log(zi)I\,dz=\int_{C_{r,+}}\log(-zi)I\,dz
-\int_{C_{r,+}}\log(-wi)I\,dw=0.
\end{align}
Consequently,
\begin{align}
\int_{C_{r,+}}\log\!\left(i(B-zA)\right)\,dz
+
\int_{C_{r,-}}\log\!\left(-i(B-zA)\right)\,dz
=
\oint_{|z|=r}\log(A-z^{-1}B)\,dz .
\label{eq:arc_log_terms_combine}
\end{align}
Combining this with \eqref{eq:T_0_r_arcs_1}, we get
\begin{align}
\lim_{\varepsilon\downarrow0}T_{\varepsilon,r}
=
-\frac{1}{2\pi i}
\oint_{|z|=r}\log(A-z^{-1}B)\,dz .
\label{eq:T_0_r_full_circle}
\end{align}

It remains to evaluate the full-circle integral. Since
\begin{align}
\frac{d}{dz}\log(A-z^{-1}B)
=
\operatorname{D}\log(A-z^{-1}B)[B]\,z^{-2},
\end{align}
integration by parts gives
\begin{align}
\oint_{|z|=r}\log(A-z^{-1}B)\,dz
&=
-\oint_{|z|=r}
z^{-1}\operatorname{D}\log(A-z^{-1}B)[B]\,dz
\notag\\
&=
-\int_0^{2\pi}
\operatorname{D}\log\!\left(A-(re^{i\theta})^{-1}B\right)[B]\cdot i\,d\theta .
\label{eq:full_circle_ibp}
\end{align}
Therefore,
\begin{align}
\lim_{\varepsilon\downarrow0}T_{\varepsilon,r}
=
\frac{1}{2\pi i}
\int_0^{2\pi}
\operatorname{D}\log\!\left(A-(re^{i\theta})^{-1}B\right)[B]\cdot i\,d\theta .
\label{eq:T_0_r_derivative_average}
\end{align}

So far, we have
\begin{align}
\int_0^{r} H(B-uA)\,du
-
\int_{-r}^{0} \left(1-H(B-uA)\right)du=\frac{1}{2\pi i}\int_0^{2\pi} \operatorname{D}\log\!\left(A-(re^{i\theta})^{-1} B\right)[B] \cdot i\,d\theta.
\end{align}
Since $H(B-uA)$ and $1-H(B-uA)$ vanish for $u> r>\Vert B\Vert/\delta$ and $u<- r<-\Vert B\Vert/\delta$, respectively, the integrals on both sides are independent of $r>\Vert B\Vert/\delta$. We can take $r \rightarrow \infty$ and recall that the Fr\'echet derivative $\operatorname{D}\log(\,\cdot\,)[B]$ is continuous to arrive at
\begin{align}
\int_0^{\infty} H(B-uA)\,du
-
\int_{-\infty}^{0} \left(1-H(B-uA)\right)du
    =\frac{1}{2\pi i}\int_0^{2\pi} \operatorname{D}\log(A)[B] \cdot i\,d\theta
    =\operatorname{D}\log(A)[B]
\end{align}
as desired.
\end{proof}
Our first minor result is the layer cake representation of the relative entropy for strictly positive densities in a von Neumann algebra. This is to be compared with \cite[Prop.~4.1]{liu2025layercakerepresentationsquantum}, which shows the result for matrices.  

\begin{proposition}[Layer cake for $D(\rho\Vert\sigma)$ in the finite tracial case]\label{prop:layer_cake_relative_entropy_finite}
Let $(\mathcal{M},\tau)$ be a finite von Neumann algebra and let $\rho,\sigma\in\mathcal{M}_+$ be strictly positive, not necessarily normalized.
Then
\begin{align}\label{eq:relative_entropy_layer_cake}
D(\rho\Vert\sigma)
=
\int_{1}^{\infty}\frac{1}{\gamma}\,
\tau\left(\rho\left(H(\rho-\gamma\sigma)-H(\sigma-\gamma\rho)\right)\right)\,d\gamma,
\end{align}
where $H(X)$ denotes the Heaviside function applied to the selfadjoint operator $X$ via the Borel functional calculus.
\end{proposition}

\begin{proof}
Set $A_t\coloneqq (1-t)\sigma+t\rho$ for $t\in[0,1]$. Then $A_t$ is strictly positive for all $t$, since $\rho$ and $\sigma$ are; note that this and all subsequent steps use only strict positivity and at no point the normalization of $\rho$ or $\sigma$.
By the fundamental theorem of calculus and \eqref{eq:def_relative_entropy_trace},
\begin{align}\label{eq:FTC_relative_entropy}
D(\rho\Vert\sigma)
=
\tau(\rho\log\rho)-\tau(\rho\log\sigma)
=
\int_0^1 \frac{d}{dt}\tau(\rho\log A_t)\,dt
=
\int_0^1 \tau\!\left(\rho\,\operatorname{D}\log(A_t)[\rho-\sigma]\right)\,dt.
\end{align}
Apply Theorem~\ref{thm:operator_layer_cake} inside $\mathcal{M}$ with $A=A_t$ and $B=\rho-\sigma$:
\begin{align}\label{eq:apply_layer_cake_to_At}
\operatorname{D}\log(A_t)[\rho-\sigma]
=
\int_{0}^{\infty} H(\rho-\sigma-uA_t)\,du
-
\int_{-\infty}^{0} \bigl(1-H(\rho-\sigma-uA_t)\bigr)\,du.
\end{align}
Insert into \eqref{eq:FTC_relative_entropy} and use normality of $\tau$ to exchange $\tau$ with the (bounded, finite-range) strong integrals:
\begin{align}\label{eq:double_integral_u_t}
D(\rho\Vert\sigma)
&=
\int_0^1\!\int_0^\infty \tau\left(\rho\,H(\rho-\sigma-uA_t)\right)\,du\,dt
\notag\\
&\quad-
\int_0^1\!\int_{-\infty}^0 \tau\left(\rho\,(1-H(\rho-\sigma-uA_t))\right)\,du\,dt.
\end{align}

We now identify the Heaviside terms. For the first term,
\[
\rho-\sigma-uA_t
=
(\rho-\sigma)-u((1-t)\sigma+t\rho)
=
(1-ut)\rho-(1+u(1-t))\sigma.
\]
For fixed $t\in(0,1]$, the coefficient $1-ut$ is positive precisely when $0\le u<1/t$.
On this interval we may rewrite
\[
(1-ut)\rho-(1+u(1-t))\sigma
=
(1-ut)\bigl(\rho-\gamma\sigma\bigr),
\qquad
\gamma=\gamma(u,t)\coloneqq \frac{1+u(1-t)}{1-ut},
\]
and since $1-ut>0$ we have
\[
H\bigl((1-ut)\rho-(1+u(1-t))\sigma\bigr)
=
H(\rho-\gamma\sigma).
\]
For $u\ge 1/t$ the operator $(1-ut)\rho-(1+u(1-t))\sigma$ is strictly negative (because $\rho,\sigma$ are strictly positive), hence its Heaviside value is $0$.
Moreover, for fixed $t\in(0,1)$, the map $u\mapsto\gamma(u,t)$ is a $C^1$ bijection from $[0,1/t)$ onto $[1,\infty)$.
A direct computation gives the Jacobian
\begin{align}\label{eq:jacobian_du_dgamma}
u=\frac{\gamma-1}{\gamma t+(1-t)},
\qquad
\frac{\partial u}{\partial\gamma}
=\frac{1}{(\gamma t+(1-t))^2}.
\end{align}
Hence
\begin{align}\label{eq:first_term_change_of_variables}
\int_0^1\!\int_0^\infty \tau\left(\rho\,H(\rho-\sigma-uA_t)\right)\,du\,dt
=
\int_0^1\!\int_1^\infty \tau\left(\rho\,H(\rho-\gamma\sigma)\right)\,
\frac{1}{(\gamma t+(1-t))^2}\,d\gamma\,dt.
\end{align}
Now integrate in $t$ explicitly:
\[
\int_0^1 \frac{dt}{(\gamma t+(1-t))^2}
=
\int_0^1 \frac{dt}{(1+t(\gamma-1))^2}
=
\left[\,-\frac{1}{(\gamma-1)(1+t(\gamma-1))}\,\right]_{t=0}^{t=1}
=
\frac{1}{\gamma}.
\]
Therefore
\begin{align}\label{eq:first_term_becomes_gamma_integral}
\int_0^1\!\int_0^\infty \tau\left(\rho\,H(\rho-\sigma-uA_t)\right)\,du\,dt
=
\int_1^\infty \frac{1}{\gamma}\,\tau\left(\rho\,H(\rho-\gamma\sigma)\right)\,d\gamma.
\end{align}

For the second term in \eqref{eq:double_integral_u_t}, use $1-H(X)=H(-X)$ and then substitute $u=-v$ with $v\ge 0$:
\begin{align}
1-H(\rho-\sigma-uA_t)=H\bigl(-(\rho-\sigma-uA_t)\bigr),
\end{align}
so
\[
1-H(\rho-\sigma-uA_t)\Big|_{u=-v}
=
H(\sigma-\rho-vA_t).
\]
Repeating the same change-of-variables computation as above (with $\rho,\sigma$ interchanged) yields
\begin{align}\label{eq:second_term_becomes_gamma_integral}
\int_0^1\!\int_{-\infty}^0 \tau\left(\rho\,(1-H(\rho-\sigma-uA_t))\right)\,du\,dt
=
\int_1^\infty \frac{1}{\gamma}\,\tau\left(\rho\,H(\sigma-\gamma\rho)\right)\,d\gamma.
\end{align}
Subtracting \eqref{eq:second_term_becomes_gamma_integral} from \eqref{eq:first_term_becomes_gamma_integral} gives \eqref{eq:relative_entropy_layer_cake}.
\end{proof}

\subsection{The Frenkel--Jen\v{c}ov\'a integral formula}

In this subsection we derive from \autoref{prop:layer_cake_relative_entropy_finite} an integral representation of the Umegaki relative entropy for arbitrary positive normal functionals on arbitrary von Neumann algebras. This will be an extension of the integral formula by Frenkel \cite{Frenkel2023}, which was previously just known on matrix algebras and limits of matrix algebras.

We begin by recalling the Jordan decomposition for normal functionals on a von Neumann algebra. Let $\mathcal{M}$ be a von Neumann algebra and let $\phi\in\mathcal{M}^\star_{\mathrm{sa}}$ be a self-adjoint normal functional. Then $\phi$ admits a unique decomposition
\begin{align}
    \phi = \phi_+ - \phi_-,
\end{align}
where $\phi_+,\phi_-\in\mathcal{M}_+^\star$ are positive normal functionals with mutually orthogonal support projections (cf.~\cite[III.4.2]{Takesaki1979}). In particular,
\begin{align}\label{eq:positive_part_as_sup}
    \phi_+(1) = \sup_{p\in\mathcal{P}(\mathcal{M})}\phi(p) = \sup_{x\in\mathcal{M},\ 0\leq x\leq 1}\phi(x),
\end{align}
where $\mathcal{P}(\mathcal{M})$ denotes the lattice of projections in $\mathcal{M}$. For brevity we write
\begin{align}\label{eq:jordan_shorthand}
    \left(s\sigma - \rho\right)_+(1) \coloneqq \sup_{p\in\mathcal{P}(\mathcal{M})}\left(s\,\sigma(p)-\rho(p)\right)
\end{align}
for $\rho,\sigma\in\mathcal{M}_+^\star$ and $s\geq 0$. If $(\mathcal{M},\tau)$ carries a faithful normal trace, we identify $\rho\in\mathcal{M}_+$ with the functional $\tau\!\left(\rho\,\cdot\right)\in\mathcal{M}_+^\star$; in this identification \eqref{eq:positive_part_as_sup} reads
\begin{align}\label{eq:variational_positive_part_tracial}
    \tau\!\left(X_+\right) = \sup_{p\in\mathcal{P}(\mathcal{M})}\tau\!\left(pX\right), \qquad X = X^\star \in \mathcal{M},
\end{align}
with the supremum attained at the support projection $\mathds{1}_{(0,\infty)}(X)$ of $X_+$. We further remark that if $A,B\in\mathcal{M}$ satisfy the assumptions of \autoref{thm:operator_layer_cake}, then all spectral projections of $B-uA$ belong to $\mathcal{M}$ by Borel functional calculus, and the strong integrals in \eqref{eq:operator_layer_cake} converge to an element of $\mathcal{M}$. In the following lemma we show that the value of the Heavyside function at $0$ does not affect the integrals.

\begin{lemma}\label{lem:vanishing_terms}
    Let $(\mathcal{M},\tau)$ be a von Neumann algebra with faithful normal semifinite trace. Let $A,B \in \mathcal{M}$ with $A\geq \delta \mathds{1}$, $\delta >0$, and $B = B^\star$, and let $\rho \in L^1\!\left(\mathcal{M},\tau\right)_+$. Then the set
    \begin{align}
        \left\{u\in \mathbb{R} \ \middle\vert \ \tau\!\left(\rho\,\mathds{1}_{\{0\}}\!\left(B-uA\right)\right)>0\right\}
    \end{align}
    is countable, hence a Lebesgue null set. In particular, defining $H_{\theta}(T) \coloneqq \mathds{1}_{(0,\infty)}(T) + \theta\,\mathds{1}_{\{0\}}(T)$ and
    \begin{align}
        \Phi_\theta \coloneqq \int_0^\infty H_\theta\!\left(B-uA\right) du - \int_{-\infty}^0 \left(1-H_\theta\!\left(B-uA\right)\right) du, \quad \theta \in [0,1],
    \end{align}
    the value $\tau\!\left(\rho \, \Phi_\theta\right)$ is independent of $\theta$.
\end{lemma}
\begin{proof}
    Define $P(u)\coloneqq \mathds{1}_{\{0\}}\!\left(B-uA\right) \in \mathcal{M}$. For a self-adjoint $T$ we have
    \begin{align}
        H_{\theta_1}\!\left(T\right) - H_{\theta_0}\!\left(T\right) = \left(\theta_1-\theta_0\right) \mathds{1}_{\{0\}}\!\left(T\right),
    \end{align}
    hence
    \begin{align}
        \Phi_{\theta_1} - \Phi_{\theta_0} = \left(\theta_1-\theta_0\right) \int_{\mathbb{R}} P(u) \, du.
    \end{align}
    By positivity and monotone convergence we may interchange the integration with $\tau\!\left(\rho\,\cdot\right)$, so the second claim follows from the first, since a countable set has Lebesgue measure zero and thus $\int_{\mathbb{R}} \tau\!\left(\rho P(u)\right) du = 0$.

    For the first claim consider the GNS-representation of $\tau$ on $L^2\!\left(\mathcal{M},\tau\right)$ and define $\xi \coloneqq \rho^{1/2} \in L^2\!\left(\mathcal{M},\tau\right)$. Since $P(u)$ is a projection, we have
    \begin{align}
        \tau\!\left(\rho P(u)\right) = \langle \xi, P(u)\xi\rangle = \Vert P(u)\xi\Vert_2^2.
    \end{align}
    As $A\geq \delta \mathds{1}$, we can define the bounded and self-adjoint operator
    \begin{align}
        C \coloneqq A^{-1/2} B A^{-1/2} \in \mathcal{M}.
    \end{align}
    Then $\zeta \coloneqq A^{-1/2}\xi \in L^2\!\left(\mathcal{M},\tau\right)$ and we define the finite Borel measure
    \begin{align}
        \mu(\Omega) \coloneqq \langle \zeta,\mathds{1}_{\Omega}(C) \zeta\rangle, \quad \Omega \subseteq \mathbb{R} \ \text{Borel}.
    \end{align}
    Since $\mu$ is finite, its set of atoms
    \begin{align}
        S \coloneqq \left\{u\in \mathbb{R} \ \middle\vert \ \mu\!\left(\{u\}\right) >0\right\} = \left\{u \in \mathbb{R} \ \middle\vert \ \Vert \mathds{1}_{\{u\}}(C)\zeta\Vert_2 >0\right\}
    \end{align}
    is countable. We prove that $P(u)\xi \neq 0$ implies $\mathds{1}_{\{u\}}(C) \zeta \neq 0$. Assume $P(u)\xi \neq 0$ and set $\eta \coloneqq P(u)\xi$. Then $\eta \in \ker\!\left(B-uA\right)$, where $B-uA$ acts on $L^2\!\left(\mathcal{M},\tau\right)$ via left-multiplication, i.e.
    \begin{align}
        (B-uA)\eta = 0.
    \end{align}
    Indeed, as $A$ has a bounded inverse, we have
    \begin{align}
        0 = A^{-1/2} (B-uA)\eta = \left(A^{-1/2}BA^{-1/2} -u\right)A^{1/2} \eta = (C-u) A^{1/2}\eta,
    \end{align}
    so $A^{1/2}\eta \in \ker (C-u) = \operatorname{Ran}\!\left(\mathds{1}_{\{u\}}(C)\right)$. Now compute
    \begin{align}
        \langle A^{1/2}\eta, \zeta\rangle = \langle A^{1/2} \eta, A^{-1/2}\xi \rangle = \langle \eta,\xi\rangle = \langle P(u) \xi,\xi \rangle = \Vert P(u)\xi\Vert_2^2 >0
    \end{align}
    by assumption, so $A^{1/2}\eta$ is not orthogonal to $\zeta$, which forces $\mathds{1}_{\{u\}}(C)\zeta \neq 0$, i.e.~$u \in S$. Hence $\left\{u \ \middle\vert \ \tau\!\left(\rho P(u)\right)>0\right\} \subseteq S$ is countable.
\end{proof}
Taking \autoref{lem:vanishing_terms} and the operator layer cake formula \autoref{thm:operator_layer_cake} together yields the following integral formula for the Umegaki relative entropy on finite von Neumann algebras. 

\begin{proposition}[Frenkel's formula in the finite tracial case]\label{prop:frenkel_finite}
Let $(\mathcal{M},\tau)$ be a finite von Neumann algebra with faithful normal tracial state $\tau$ and let $\rho,\sigma\in\mathcal{M}_+$ with $s(\rho)\leq s(\sigma)$, not necessarily normalized or invertible. Then
\begin{align}\label{eq:frenkel_jencova_finite}
D(\rho\Vert\sigma)
=
\tau\!\left(\rho-\sigma\right)
+
\int_0^1 \frac{1}{s}\,\tau\!\left(\left(s\sigma-\rho\right)_+\right)ds
+
\int_1^\infty \frac{1}{s}\,\tau\!\left(\left(\rho-s\sigma\right)_+\right)ds,
\end{align}
with both sides possibly infinite.
\end{proposition}

\begin{proof}
Assume first that $\rho,\sigma$ are strictly positive. By \autoref{prop:layer_cake_relative_entropy_finite} and \autoref{lem:vanishing_terms}, applied with the functional $\tau\!\left(\rho\,\cdot\right)$ and the pairs $(A,B)=(\sigma,\rho)$ and $(A,B)=(\rho,\sigma)$ to replace the convention $H(0)=\tfrac{1}{2}$ by the strict spectral projections on a set of full measure, we have
\begin{align}\label{eq:layercake_strict}
D(\rho\Vert\sigma)
=
\int_1^\infty \frac{1}{s}\,\tau\!\left(\rho\left(p_s-q_s\right)\right)ds,
\qquad
p_s\coloneqq \mathds{1}_{(0,\infty)}\!\left(\rho-s\sigma\right),
\quad
q_s\coloneqq \mathds{1}_{(0,\infty)}\!\left(\sigma-s\rho\right).
\end{align}
Since $p_s$ and $q_s$ are the support projections of $\left(\rho-s\sigma\right)_+$ and $\left(\sigma-s\rho\right)_+$, we have $\tau\!\left(\left(\rho-s\sigma\right)_+\right)=\tau\!\left(\rho p_s\right)-s\,\tau\!\left(\sigma p_s\right)$ and $\tau\!\left(\left(\sigma-s\rho\right)_+\right)=\tau\!\left(\sigma q_s\right)-s\,\tau\!\left(\rho q_s\right)$, whence
\begin{align}\label{eq:integrand_split}
\frac{1}{s}\,\tau\!\left(\rho\left(p_s-q_s\right)\right)
=
\frac{1}{s}\,\tau\!\left(\left(\rho-s\sigma\right)_+\right)
+\frac{1}{s^2}\,\tau\!\left(\left(\sigma-s\rho\right)_+\right)
+\tau\!\left(\sigma p_s\right)-\frac{1}{s^2}\,\tau\!\left(\sigma q_s\right).
\end{align}
The substitution $s\mapsto 1/s$ together with $\left(\sigma-s^{-1}\rho\right)_+=s^{-1}\left(s\sigma-\rho\right)_+$ turns the second term into the first integral of \eqref{eq:frenkel_jencova_finite},
\begin{align}\label{eq:substitution_second_term}
\int_1^\infty \frac{1}{s^2}\,\tau\!\left(\left(\sigma-s\rho\right)_+\right)ds
=
\int_0^1 \frac{1}{s}\,\tau\!\left(\left(s\sigma-\rho\right)_+\right)ds,
\end{align}
so it remains to show
\begin{align}\label{eq:mass_balance}
\int_1^\infty \tau\!\left(\sigma p_s\right)ds-\int_1^\infty \frac{1}{s^2}\,\tau\!\left(\sigma q_s\right)ds
=
\tau\!\left(\rho-\sigma\right).
\end{align}
Consider $F(s)\coloneqq \tau\!\left(\left(\rho-s\sigma\right)_+\right)$ for $s\geq 0$. By the variational formula $\tau\!\left(X_+\right)=\sup_{p\in\mathcal{P}(\mathcal{M})}\tau\!\left(pX\right)$, $F$ is a supremum of affine functions of $s$, hence convex, with $F(0)=\tau\!\left(\rho\right)$ and $F(s)=0$ for $s\geq\left\Vert\sigma^{-1/2}\rho\,\sigma^{-1/2}\right\Vert$. Evaluating the variational formula at the projections $p_s$ and $p_{s+h}$, which are optimal at $s$ and $s+h$ respectively, gives for $h>0$
\begin{align}
-\tau\!\left(\sigma p_s\right)
\leq
\frac{F(s+h)-F(s)}{h}
\leq
-\tau\!\left(\sigma p_{s+h}\right),
\end{align}
so $s\mapsto\tau\!\left(\sigma p_s\right)$ is non-increasing and $F'(s)=-\tau\!\left(\sigma p_s\right)$ at almost every $s$. As a finite convex function, $F$ is locally absolutely continuous, and the fundamental theorem of calculus yields
\begin{align}\label{eq:ftc_F}
\int_a^b \tau\!\left(\sigma p_s\right)ds=F(a)-F(b), \qquad 0\leq a\leq b,
\end{align}
in particular $\int_1^\infty \tau\!\left(\sigma p_s\right)ds=F(1)$. For the second term in \eqref{eq:mass_balance}, the substitution $s\mapsto 1/s$ and the fact that $\sigma-s^{-1}\rho$ and $s\sigma-\rho$ have the same strictly positive spectral projection give
\begin{align}
\int_1^\infty \frac{1}{s^2}\,\tau\!\left(\sigma q_s\right)ds
=
\int_0^1 \tau\!\left(\sigma\,\mathds{1}_{(0,\infty)}\!\left(s\sigma-\rho\right)\right)ds.
\end{align}
Since $\mathds{1}_{(0,\infty)}\!\left(s\sigma-\rho\right)=\mathds{1}-p_s-\mathds{1}_{\{0\}}\!\left(\rho-s\sigma\right)$ and, by \autoref{lem:vanishing_terms} applied with the functional $\tau\!\left(\sigma\,\cdot\right)$, the kernel term vanishes for almost every $s$, we obtain with \eqref{eq:ftc_F}
\begin{align}
\int_1^\infty \frac{1}{s^2}\,\tau\!\left(\sigma q_s\right)ds
=
\int_0^1 \left(\tau\!\left(\sigma\right)-\tau\!\left(\sigma p_s\right)\right)ds
=
\tau\!\left(\sigma\right)-\tau\!\left(\rho\right)+F(1),
\end{align}
which proves \eqref{eq:mass_balance}. Integrating \eqref{eq:integrand_split} over $s\in[1,\infty)$ and inserting \eqref{eq:layercake_strict}, \eqref{eq:substitution_second_term} and \eqref{eq:mass_balance} establishes \eqref{eq:frenkel_jencova_finite} for strictly positive $\rho,\sigma$; note that all integrals are finite in this case, since the integrands vanish outside a compact subset of $(0,\infty)$ and are bounded.

Now let $\rho,\sigma\in\mathcal{M}_+$ be arbitrary with $s(\rho)\leq s(\sigma)$, and set
\begin{align}
\rho_\varepsilon \coloneqq \frac{\rho+\varepsilon\mathds{1}}{1+\varepsilon},
\qquad
\sigma_\varepsilon \coloneqq \frac{\sigma+\varepsilon\mathds{1}}{1+\varepsilon},
\qquad \varepsilon>0,
\end{align}
which are strictly positive, so that \eqref{eq:frenkel_jencova_finite} holds for $\left(\rho_\varepsilon,\sigma_\varepsilon\right)$. We now justify the passage to the limit \(\varepsilon\downarrow0\) in the
entropy term. We use joint convexity and lower semicontinuity of the relative
entropy for positive normal functionals, which follow from Kosaki's variational
formula \cite{kosaki_1986_variationalformula}. Put
\[
\lambda_\varepsilon:=\frac{1}{1+\varepsilon}.
\]
Then, 
\[
(\rho_\varepsilon,\sigma_\varepsilon)
=
\lambda_\varepsilon(\rho,\sigma)
+
(1-\lambda_\varepsilon)(\mathds{1},\mathds{1}).
\]
Since \(D(\mathds{1}\Vert\mathds{1})=0\), joint convexity gives
\[
D(\rho_\varepsilon\Vert\sigma_\varepsilon)
\le
\lambda_\varepsilon D(\rho\Vert\sigma).
\]
Consequently,
\[
\limsup_{\varepsilon\downarrow0}
D(\rho_\varepsilon\Vert\sigma_\varepsilon)
\le
D(\rho\Vert\sigma)
\quad\text{in }(-\infty,+\infty].
\]
Moreover,
\[
\tau\bigl(|\rho_\varepsilon-\rho|\bigr)
=
\frac{\varepsilon}{1+\varepsilon}\tau\bigl(|\mathds{1}-\rho|\bigr)
\le
\frac{\varepsilon}{1+\varepsilon}\bigl(\tau(\mathds{1})+\tau(\rho)\bigr)
\to0,
\]
and similarly
\[
\tau\bigl(|\sigma_\varepsilon-\sigma|\bigr)
\le
\frac{\varepsilon}{1+\varepsilon}\bigl(\tau(\mathds{1})+\tau(\sigma)\bigr)
\to0.
\]
Thus the positive normal functionals represented by \(\rho_\varepsilon\) and
\(\sigma_\varepsilon\) converge in norm to those represented by \(\rho\) and
\(\sigma\), respectively. Hence lower semicontinuity gives
\[
D(\rho\Vert\sigma)
\le
\liminf_{\varepsilon\downarrow0}
D(\rho_\varepsilon\Vert\sigma_\varepsilon).
\]
Combining the two inequalities yields
\[
D(\rho_\varepsilon\Vert\sigma_\varepsilon)
\to
D(\rho\Vert\sigma)
\quad\text{in }(-\infty,+\infty].
\]
For the right-hand side, $\tau\!\left(\rho_\varepsilon-\sigma_\varepsilon\right)=\tau\!\left(\rho-\sigma\right)/(1+\varepsilon)\to\tau\!\left(\rho-\sigma\right)$, and since $X\mapsto\tau\!\left(X_+\right)$ satisfies $\left\vert\tau\!\left(X_+\right)-\tau\!\left(Y_+\right)\right\vert\leq\left\Vert X-Y\right\Vert_1$, both integrands converge pointwise in $s$. On $(0,1]$ the integrands are uniformly bounded, since $\frac{1}{s}\,\tau\!\left(\left(s\sigma_\varepsilon-\rho_\varepsilon\right)_+\right)\leq\tau\!\left(\sigma_\varepsilon\right)\leq\tau\!\left(\sigma\right)+\varepsilon\tau\!\left(\mathds{1}\right)$, so dominated convergence applies. On $[1,\infty)$ we have $\rho_\varepsilon-s\sigma_\varepsilon\leq\left(\rho-s\sigma\right)/(1+\varepsilon)$, hence by monotonicity of $X\mapsto\tau\!\left(X_+\right)$
\begin{align}
\frac{1}{s}\,\tau\!\left(\left(\rho_\varepsilon-s\sigma_\varepsilon\right)_+\right)
\leq
\frac{1}{s}\,\tau\!\left(\left(\rho-s\sigma\right)_+\right).
\end{align}
If the right-hand side is integrable on $[1,\infty)$, dominated convergence applies; if not, Fatou's lemma forces the integrals to diverge along $\varepsilon\downarrow 0$ as well. In either case the right-hand side of \eqref{eq:frenkel_jencova_finite} for $\left(\rho_\varepsilon,\sigma_\varepsilon\right)$ converges in $(-\infty,+\infty]$ to the right-hand side for $\left(\rho,\sigma\right)$, and the identity extends.
\end{proof}

To generalize the formula in \autoref{prop:frenkel_finite} to arbitrary von Neumann algebras, we use the concrete form of Haagerup's reduction theorem. The point requiring care is that the finite approximating states must have bounded densities with respect to the finite traces before \autoref{prop:frenkel_finite} can be applied; this is ensured below by choosing the reference state in the reduction to dominate the two states under consideration.
\begin{theorem}[Frenkel formula for general von Neumann algebras]\label{thm:frenkel_general}
Let $\mathcal{M}$ be a von Neumann algebra and let $\rho,\sigma\in\mathcal{M}_+^\star$ with $s_{\mathcal{M}}(\rho)\leq s_{\mathcal{M}}(\sigma)$. Then
\begin{align}\label{eq:frenkel_jencova_general}
D(\rho\Vert\sigma)
=
\left(\rho-\sigma\right)(\mathds{1})
+
\int_0^1 \frac{1}{s}\,\left(s\sigma-\rho\right)_+(\mathds{1})\,ds
+
\int_1^\infty \frac{1}{s}\,\left(\rho-s\sigma\right)_+(\mathds{1})\,ds,
\end{align}
with both sides possibly infinite.
\end{theorem}

\begin{proof}
If $\rho=0$, then $D(0\Vert\sigma)=0$ and the right-hand side of \eqref{eq:frenkel_jencova_general} is
$-\sigma(\mathds{1})+\int_0^1 \sigma(\mathds{1})\,ds=0$. We therefore assume $\rho\neq0$. The support assumption then implies $\sigma\neq0$.

For a unital von Neumann algebra $\mathcal A$ and positive normal functionals $\alpha,\beta\in\mathcal A_+^\star$, write
\begin{align*}
F_{\mathcal A}(\alpha,\beta)
&\coloneqq
(\alpha-\beta)(\mathds{1}_{\mathcal A})
+
\int_0^1 \frac{1}{s}\,(s\beta-\alpha)_+(\mathds{1}_{\mathcal A})\,ds
+
\int_1^\infty \frac{1}{s}\,(\alpha-s\beta)_+(\mathds{1}_{\mathcal A})\,ds.
\end{align*}
We first reduce the identity to the case of states. Put
\begin{align*}
    a\coloneqq\rho(\mathds{1}),\qquad b\coloneqq\sigma(\mathds{1}),\qquad
    \bar\rho\coloneqq a^{-1}\rho,\qquad \bar\sigma\coloneqq b^{-1}\sigma,
    \qquad c\coloneqq b/a.
\end{align*}
The scaling in the first argument follows directly from the standard-form definition of Araki relative entropy: $\xi_{a\bar\rho}=a^{1/2}\xi_{\bar\rho}$ and $\Delta(\sigma,a\bar\rho)=a^{-1}\Delta(\sigma,\bar\rho)$, hence
$D(a\bar\rho\Vert\sigma)=aD(\bar\rho\Vert\sigma)+a\log a$. Combining this with the scaling in the second argument from \autoref{lem:araki_order_scaling} gives
\begin{align}\label{eq:frenkel_general_lhs_scaling}
    D(\rho\Vert\sigma)=aD(\bar\rho\Vert\bar\sigma)+a\log\frac{a}{b}.
\end{align}
On the right-hand side, the substitution $w=cs$ and the homogeneity of the positive part give
\begin{align*}
F_{\mathcal M}(\rho,\sigma)
&=a-b
+a\left(
\int_0^{c}\frac{1}{w}\,(w\bar\sigma-\bar\rho)_+(\mathds{1})\,dw
+
\int_{c}^{\infty}\frac{1}{w}\,(\bar\rho-w\bar\sigma)_+(\mathds{1})\,dw
\right).
\end{align*}
For every $w>0$, the Jordan decomposition identity
$\omega_+(\mathds{1})-(-\omega)_+(\mathds{1})=\omega(\mathds{1})$, applied to $\omega=w\bar\sigma-\bar\rho$, gives
\begin{align*}
(w\bar\sigma-\bar\rho)_+(\mathds{1})-(\bar\rho-w\bar\sigma)_+(\mathds{1})=w-1.
\end{align*}
Thus, with the integral over $[1,c]$ interpreted as an oriented integral when $c<1$,
\begin{align}\label{eq:frenkel_general_rhs_scaling}
F_{\mathcal M}(\rho,\sigma)-aF_{\mathcal M}(\bar\rho,\bar\sigma)
&=a-b+a\int_1^c\frac{w-1}{w}\,dw
=a\log\frac{a}{b}.
\end{align}
It is therefore enough, by \eqref{eq:frenkel_general_lhs_scaling} and \eqref{eq:frenkel_general_rhs_scaling}, to prove the formula when $\rho$ and $\sigma$ are states.

Assume from now on that $\rho,\sigma\in\mathcal S(\mathcal M)$. Let $p=s_{\mathcal M}(\sigma)$. Since $s_{\mathcal M}(\rho)\le p$, both functionals are supported on $p$. Passing from $\mathcal M$ to the corner $p\mathcal M p$ does not change either side of \eqref{eq:frenkel_jencova_general}: for the Jordan terms, if $\omega$ is any self-adjoint normal functional supported on $p$, then
\begin{align*}
\omega_+^{\mathcal M}(\mathds{1})
=\sup_{0\le x\le \mathds{1}_{\mathcal M}}\omega(x)
=\sup_{0\le y\le p}\omega(y)
=\omega_+^{p\mathcal M p}(p),
\end{align*}
and the equality of Araki relative entropy under this support reduction is standard from the definition in standard form \cite{Araki1975,Takesaki1979}. After this replacement, and relabelling $p\mathcal M p$ as $\mathcal M$, we may suppose that $\sigma$ is faithful; in particular $\mathcal M$ is $\sigma$-finite. Set
\begin{align*}
    \varphi\coloneqq\frac{1}{2}(\rho+\sigma).
\end{align*}
Then $\varphi$ is a faithful normal state and $\rho,\sigma\le 2\varphi$.

Apply Haagerup's reduction theorem to $(\mathcal M,\varphi)$ in its concrete crossed-product form \cite[Thm.~2.1, Lem.~2.4]{Haagerup2009}. Thus $\mathcal M$ embeds into an ambient von Neumann algebra $\widetilde{\mathcal M}$, there is a normal faithful conditional expectation $E:\widetilde{\mathcal M}\to\mathcal M$, and there is an increasing sequence of finite von Neumann subalgebras $(\mathcal M_n)_{n\ge1}$ whose union is $\sigma$-weakly dense in $\widetilde{\mathcal M}$. Let
\begin{align*}
    \widetilde\rho\coloneqq\rho\circ E,
    \qquad
    \widetilde\sigma\coloneqq\sigma\circ E,
    \qquad
    \widetilde\varphi\coloneqq\varphi\circ E,
    \qquad
    \rho_n\coloneqq\widetilde\rho\vert_{\mathcal M_n},
    \qquad
    \sigma_n\coloneqq\widetilde\sigma\vert_{\mathcal M_n}.
\end{align*}
The same concrete reduction writes $\mathcal M_n$ as the centralizer of a normal faithful positive functional $\varphi_n$, whose restriction
$\tau_n\coloneqq\varphi_n\vert_{\mathcal M_n}$ is a normal faithful finite trace, and there is a bounded positive element $a_n\in\mathcal M_n$ such that
\begin{align*}
    \varphi_n(x)=\widetilde\varphi(e^{-a_n}x),\qquad 0\le a_n\le \alpha_n\mathds{1},\qquad \alpha_n\coloneqq 2^{n+1}\pi.
\end{align*}
Since $\widetilde\rho,\widetilde\sigma\le 2\widetilde\varphi$, for $x\in(\mathcal M_n)_+$ we obtain
\begin{align}\label{eq:frenkel_general_bounded_density_bound}
\rho_n(x)
&\le 2\widetilde\varphi(x)
=2\tau_n(e^{a_n}x)
=2\tau_n(e^{a_n/2}xe^{a_n/2})
\le 2e^{\alpha_n}\tau_n(x),
\end{align}
and the same bound holds for $\sigma_n$. Hence, by the Radon--Nikodym theorem for normal functionals dominated by a finite trace \cite{Takesaki1979}, the functionals $\rho_n$ and $\sigma_n$ are represented by bounded positive densities in $\mathcal M_n$ with respect to $\tau_n$. Moreover, $s_{\mathcal M_n}(\rho_n)\le s_{\mathcal M_n}(\sigma_n)$: if $x\in(\mathcal M_n)_+$ and $\sigma_n(x)=0$, then $\sigma(E(x))=0$, whence $\rho(E(x))=0$ by $s_{\mathcal M}(\rho)\le s_{\mathcal M}(\sigma)$, and therefore $\rho_n(x)=0$.

We may now apply \autoref{prop:frenkel_finite} inside $(\mathcal M_n,\tau_n)$ to these bounded densities. Since $\rho_n$ and $\sigma_n$ are states, the mass term vanishes and
\begin{align}\label{eq:frenkel_general_level_n}
D(\rho_n\Vert\sigma_n)
=
\int_0^1 \frac{1}{s}\,(s\sigma_n-\rho_n)_+(\mathds{1})\,ds
+
\int_1^\infty \frac{1}{s}\,(\rho_n-s\sigma_n)_+(\mathds{1})\,ds.
\end{align}
By the entropy approximation part of Haagerup reduction \cite[Prop.~2.2]{Fawzi_2025},
\begin{align}\label{eq:frenkel_general_haagerup_entropy_limit}
    \lim_{n\to\infty}D(\rho_n\Vert\sigma_n)
    =D(\widetilde\rho\Vert\widetilde\sigma)
    =D(\rho\Vert\sigma),
\end{align}
where the last equality follows because $E$ is a conditional expectation onto the embedded copy of $\mathcal M$.

It remains to pass to the limit in the two integrals in \eqref{eq:frenkel_general_level_n}. For a self-adjoint normal functional $\omega$ on a von Neumann algebra, the Jordan decomposition gives
\begin{align}\label{eq:frenkel_general_jordan_norm_identity}
    \omega_+(\mathds{1})=\frac{1}{2}\left(\Vert\omega\Vert+\omega(\mathds{1})\right)
\end{align}
(cf.~\cite[III.4.2]{Takesaki1979}). Fix $s>0$ and put $\omega_s\coloneqq s\widetilde\sigma-\widetilde\rho$. Since $\omega_s$ is normal and the unit ball of $\bigcup_n\mathcal M_n$ is $\sigma$-weakly dense in the unit ball of $\widetilde{\mathcal M}$, we have
\begin{align}\label{eq:frenkel_general_norm_convergence}
\Vert s\sigma_n-\rho_n\Vert
&=
\sup_{x\in\mathcal M_n,\,\Vert x\Vert\le1}\left\vert\omega_s(x)\right\vert
\nearrow
\Vert\omega_s\Vert_{\widetilde{\mathcal M}_\star}.
\end{align}
Furthermore, because $\omega_s=(s\sigma-\rho)\circ E$ and $E\vert_{\mathcal M}=\operatorname{id}_{\mathcal M}$, contractivity of $E$ gives
\begin{align*}
    \Vert\omega_s\Vert_{\widetilde{\mathcal M}_\star}=\Vert s\sigma-\rho\Vert_{\mathcal M_\star}.
\end{align*}
Using \eqref{eq:frenkel_general_jordan_norm_identity}, together with
$(s\sigma_n-\rho_n)(\mathds{1})=s-1$ and $(\rho_n-s\sigma_n)(\mathds{1})=1-s$, we get pointwise monotone convergence of both integrands:
\begin{align*}
    (s\sigma_n-\rho_n)_+(\mathds{1})&\nearrow(s\sigma-\rho)_+(\mathds{1}),\\
    (\rho_n-s\sigma_n)_+(\mathds{1})&\nearrow(\rho-s\sigma)_+(\mathds{1}).
\end{align*}
The monotone convergence theorem therefore applies to both non-negative integrals in \eqref{eq:frenkel_general_level_n}. Combining the resulting limit with \eqref{eq:frenkel_general_haagerup_entropy_limit} proves \eqref{eq:frenkel_jencova_general} for states. The scaling reduction at the beginning of the proof proves the asserted formula for arbitrary positive normal functionals satisfying the support condition.
\end{proof}

Similarly to \cite[Cor.~1]{Jencova2024}, we prove the following corollary.
\begin{corollary}[Finite-range integral formula]\label{coro:jencova_finite_range}
Let $\mathcal{M}$ be a von Neumann algebra and let $\rho,\sigma\in\mathcal{M}_+^\star$. Suppose there exist $\mu\ge 0$ and $\lambda> 0$ such that
\begin{align}\label{eq:sandwich_condition}
    \mu\,\sigma \le \rho \le \lambda\,\sigma
\end{align}
in the ordering of $\mathcal{M}_+^\star$. Then
\begin{align}\label{eq:jencova_finite_range}
    D(\rho\Vert\sigma)
    =
    \left(\rho-\sigma\right)(1)
    +
    \int_\mu^\lambda \frac{1}{s}\,
    \left(s\,\sigma-\rho\right)_+(1)\,ds
    +
    \rho(1)\,\ln\lambda
    -
    (\lambda-1)\,\sigma(1).
\end{align}
\end{corollary}

\begin{proof}
We first note that \eqref{eq:sandwich_condition} implies $s_{\mathcal{M}}(\rho)\le s_{\mathcal{M}}(\sigma)$: if $x\in\mathcal{M}_+$ with $\sigma(x)=0$, then $0\le\rho(x)\le\lambda\,\sigma(x)=0$. Hence \autoref{thm:frenkel_general} applies and we start from the integral formula \eqref{eq:frenkel_jencova_general}.

Since $\mu\,\sigma\le\rho$, for every $s\in[0,\mu]$ one has $s\,\sigma\le\mu\,\sigma\le\rho$, whence $\left(s\,\sigma-\rho\right)_+(1)=0$. Thus the first integral in \eqref{eq:frenkel_jencova_general} reduces to $\int_\mu^1$. Likewise, $\rho\le\lambda\,\sigma$ implies $\left(\rho-s\,\sigma\right)_+(1)=0$ for all $s\ge\lambda$, and the second integral reduces to $\int_1^\lambda$. In particular, both remaining integrands are bounded by $\left(s+1\right)\left(\rho(1)+\sigma(1)\right)/s$ on compact subsets of $(0,\infty)$, so all integrals below are finite.

For any self-adjoint normal functional $\phi\in\mathcal{M}_\star^{\mathrm{sa}}$, the Jordan decomposition gives
\begin{align}
    \phi_+(1) = \phi(1) + \phi_-(1)
    = \phi(1) + (-\phi)_+(1).
\end{align}
Applied to $\phi = \rho - s\,\sigma$ with $s\in[1,\lambda]$, this yields
\begin{align}\label{eq:positive_negative_identity}
    \left(\rho - s\,\sigma\right)_+(1)
    =
    \rho(1) - s\,\sigma(1)
    +
    \left(s\,\sigma - \rho\right)_+(1).
\end{align}
Substituting \eqref{eq:positive_negative_identity} into the truncated second integral gives
\begin{align}
    \int_1^\lambda \frac{1}{s}\,\left(\rho - s\,\sigma\right)_+(1)\,ds
    &=
    \int_1^\lambda \frac{1}{s}\,
    \left(s\,\sigma - \rho\right)_+(1)\,ds
    +
    \int_1^\lambda \frac{\rho(1) - s\,\sigma(1)}{s}\,ds.
\end{align}
The elementary integral evaluates to
\begin{align}\label{eq:elementary_integral}
    \int_1^\lambda \frac{\rho(1) - s\,\sigma(1)}{s}\,ds
    =
    \rho(1)\,\ln\lambda - (\lambda - 1)\,\sigma(1).
\end{align}
Combining the truncated first integral over $[\mu,1]$ with the rewritten second integral over $[1,\lambda]$ collapses the two $\left(s\,\sigma-\rho\right)_+(1)$-integrals into a single integral over $[\mu,\lambda]$. Together with \eqref{eq:elementary_integral} and the $\left(\rho-\sigma\right)(1)$ term from \eqref{eq:frenkel_jencova_general}, this yields \eqref{eq:jencova_finite_range}.

Finally, the result is independent of the particular choice of $\mu$ and $\lambda$ satisfying \eqref{eq:sandwich_condition}: if $\lambda'\ge\lambda$, the integrand $\left(s\,\sigma-\rho\right)_+(1)$ vanishes on $(\lambda,\lambda']$, and the change in $\rho(1)\ln\lambda'-(\lambda'-1)\sigma(1)$ is exactly compensated by the additional piece of the elementary integral \eqref{eq:elementary_integral}. An analogous argument applies at the lower limit $\mu$.
\end{proof}
\section{Approximations and NPA-programs}\label{sec:approx}
In this section we aim to combine the tools from \autoref{sec:operator_layer_cake} with the remaining optimization problem for a security proof in DIQKD. We recall the optimization problem after applying the dilation argument from \autoref{thm:dilation_theorem} in the following (cf.~\autoref{eq:commuting_operator_entropy_program_gns})
\begin{equation}\label{eq:commuting_operator_entropy_program_repeat}
    \begin{aligned}
        \inf\ & H(Z|E)_{\omega_{ZE}} \\
        \text{s.t. }&
        \psi \in \mathcal S\!\bigl(\mathcal P_A\otimes_{\max}\mathcal P_B\bigr),\\
        &\psi\!\left(p_{a|x}\otimes q_{b|y}\right)=p(a,b|x,y),
        \qquad
        a,b,x,y,\\
        &\mathcal M_E \subseteq \mathcal M_{AB,\psi}',\\
        &\omega_{ZE}\in
        \mathcal S\!\bigl(\ell^\infty(\mathcal Z)\,\overline\otimes\,\mathcal M_E\bigr),\\
        &\omega_{ZE}(p_z\otimes c)
        =
        \bigl\langle \xi_\psi,\,
        \pi_\psi(p_{z|\tilde x})\,c\,\xi_\psi
        \bigr\rangle,
        \qquad
        z\in\mathcal Z,\ c\in \mathcal M_E.
    \end{aligned}
\end{equation}
The idea is to replace the conditional von Neumann entropy with a relative entropy and to apply then subsequently the integral formula from \autoref{coro:jencova_finite_range}.
\subsection{Discretization and upper bounds for relative entropy}

The integral in \autoref{coro:jencova_finite_range} can be approximated from above by a finite sum of suprema over projections as shown in \cite{Komann2026,kossmann2025reliableentropyestimationobserved}. The resulting upper bound for the relative entropy translates directly into a lower bound for the conditional von Neumann entropy and, ultimately, into a semidefinite program via the NPA hierarchy.

\begin{theorem}[Upper bounds for relative entropy]\label{thm:upper_bound_relative_entropy}
\rausdamit{Let $\mathcal{M}$ be a von Neumann algebra and let $\rho,\sigma\in\mathcal{M}_+^\star$ with
\begin{align}
    \mu\,\sigma \le \rho \le \lambda\,\sigma
\end{align}
for some $\mu\ge 0$ and $\lambda>0$. Let $r\ge 1$ and fix a partition $\mu = t_1 < t_2 < \cdots < t_r = \lambda$. Then there exist real numbers $\alpha_0,\ldots,\alpha_r$ and $\beta_0,\ldots,\beta_r$, \rs{(given in \eqref{def:coeffs})} depending only on $\mu,\lambda$ and the partition, such that
}

\rs{Let $\mathcal{M}$ be a von Neumann algebra and let $\rho,\sigma\in\mathcal{M}_+^\star$ and $\mu,\lambda$ a pair of cut off parameters with
\begin{align}
    \rho \le \lambda\,\sigma \text{ and } \lambda>\mu>0
\end{align}. Let $r\ge 1$ and fix a partition $\mu = t_1 < t_2 < \cdots < t_r = \lambda$. Then there exist real numbers $\alpha_0,\ldots,\alpha_r$ and $\beta_0,\ldots,\beta_r$, (given in \eqref{def:coeffs}) depending only on $\mu,\lambda$ and the partition, such that
}
\begin{align}\label{eq:integral_upper_bound}
    \int_0^\lambda \frac{1}{s}\,
    \left(s\,\sigma - \rho\right)_+(1)\,ds
    \le
    \sum_{k=0}^r
    \sup_{p_k\in\mathcal{P}(\mathcal{M})}
    \left(\alpha_k\,\rho + \beta_k\,\sigma\right)(p_k).
\end{align}
In particular,
\begin{align}\label{eq:relative_entropy_upper_bound}
    D(\rho\Vert\sigma)
    \le
    \left(\rho-\sigma\right)(1)
    +
    \sum_{k=0}^r
    \sup_{p_k\in\mathcal{P}(\mathcal{M})}
    \left(\alpha_k\,\rho + \beta_k\,\sigma\right)(p_k)
    +
    \rho(1)\,\ln\lambda
    -
    (\lambda-1)\,\sigma(1).
\end{align}
Moreover, the sequence of upper bounds converges to $D(\rho\Vert\sigma)$ when the mesh of the partition tends to cover the interval $(0,\lambda)$.

\end{theorem}

\begin{proof} 
Define the function
\begin{align}
    F(s)
    \coloneqq
    \left(s\,\sigma - \rho\right)_+(1)
    =
    \sup_{p\in\mathcal{P}(\mathcal{M})}
    \left(s\,\sigma(p) - \rho(p)\right),
    \qquad s\ge 0.
\end{align}
Since $F$ is the pointwise supremum of a family of affine functions in $s$, it is convex and continuous on $[0,\infty)$. \rausdamit{Moreover, $F(0)=0$, as $\rho\ge 0$, and $F(s)=0$ for $s\in[0,\mu]$, since $\mu\,\sigma\le\rho$.}
\rs{Moreover, $F(0)=0$, since $\rho$ is positive. We emphasize that no lower domination assumption of the form $\mu\,\sigma\le\rho$ is used here; the parameter $\mu>0$ is only an auxiliary lower cut-off which keeps the logarithmic kernel away from the singular endpoint.}
By convexity and the boundary values $F(0)=0$ and $F(\mu)\ge 0$, for every $s\in[0,\mu]$,
\begin{align}
    F(s)
    \le
    \left(1 - \frac{s}{\mu}\right)F(0)
    +
    \frac{s}{\mu}\,F(\mu)
    =
    \frac{s}{\mu}\,F(\mu).
\end{align}
Hence
\begin{align}\label{eq:lower_tail_bound}
    \int_0^\mu \frac{F(s)}{s}\,ds
    \le
    \int_0^\mu \frac{F(\mu)}{\mu}\,ds
    =
    F(\mu)
    =
    \sup_{p_0\in\mathcal{P}(\mathcal{M})}
    \left(\mu\,\sigma(p_0) - \rho(p_0)\right),
\end{align}
which corresponds to the choice $\alpha_0 \coloneqq -1$ and $\beta_0 \coloneqq \mu$. For each grid point $t_k$, define
\begin{align}
    y_k
    \coloneqq
    F(t_k)
    =
    \sup_{p\in\mathcal{P}(\mathcal{M})}
    \left(t_k\,\sigma(p) - \rho(p)\right).
\end{align}
Since $s\mapsto F(s)/s$ is the ratio of a convex function to a linear function on $(0,\infty)$, the piecewise-linear interpolant of the values $(t_k,y_k)$ yields, by the methods of \cite{kossmann2025reliableentropyestimationobserved}, the estimate
\begin{align}\label{eq:quadrature_estimate}
    \int_\mu^\lambda \frac{F(s)}{s}\,ds
    \le
    \sum_{k=1}^r c_k\,y_k,
\end{align}
where the quadrature weights $c_k\ge 0$ are given by
\begin{align}
    c_1
    &\coloneqq
    1 + \frac{t_1}{t_2-t_1}
    \left(\ln\frac{t_2}{t_1} - 1\right),
    \\
    c_r
    &\coloneqq
    1 - \frac{t_{r-1}}{t_r - t_{r-1}}
    \ln\frac{t_r}{t_{r-1}},
    \\
    c_k
    &\coloneqq
    1
    + \frac{t_k}{t_{k+1}-t_k}
    \ln\frac{t_{k+1}}{t_k}
    - \frac{t_{k-1}}{t_k - t_{k-1}}
    \ln\frac{t_k}{t_{k-1}},
    \qquad
    2\le k\le r-1.
\end{align}
Combining \eqref{eq:lower_tail_bound} and \eqref{eq:quadrature_estimate}, and writing $y_k = \sup_{p_k}\left(t_k\,\sigma(p_k) - \rho(p_k)\right)$, we obtain
\begin{align}
    \int_0^\lambda \frac{F(s)}{s}\,ds
    \le
    \sup_{p_0}\left(-\rho(p_0)+\mu\,\sigma(p_0)\right)
    +
    \sum_{k=1}^r
    \sup_{p_k}\left(-c_k\,\rho(p_k) + c_k t_k\,\sigma(p_k)\right).
\end{align}
Defining
\begin{align}\label{def:coeffs}
    \alpha_k
    &\coloneqq
    \begin{cases}
    -1, & k=0,\\
    -c_k, & 1\le k\le r,
    \end{cases}
    &
    \beta_k
    &\coloneqq
    \begin{cases}
    \mu, & k=0,\\
    c_k\,t_k, & 1\le k\le r,
    \end{cases}
\end{align}
yields \eqref{eq:integral_upper_bound}. \rausdamit{ Inserting this into \autoref{coro:jencova_finite_range} produces \eqref{eq:relative_entropy_upper_bound}. Convergence as the mesh tends to zero follows from the fact that the piecewise-linear upper bound for $F$ converges uniformly on compact sets to the convex function $F$.}

\rs{`
Inserting this into \autoref{coro:jencova_finite_range} produces \eqref{eq:relative_entropy_upper_bound}. It remains to explain the convergence as the lower cut-off is removed.
Fix $\mu>0$ and let $\mathcal{T}=\{\mu=t_1<t_2<\cdots<t_r=\lambda\}$ be a partition of $[\mu,\lambda]$. Denote by $L_{\mathcal{T}}F$ the piecewise affine interpolant of $F$ with respect to this partition. Since $F$ is convex, $L_{\mathcal{T}}F\ge F$ on $[\mu,\lambda]$, and the bound obtained above can be written as
\begin{align}
    B_{\mu,\mathcal{T}}
    =
    F(\mu)
    +
    \int_\mu^\lambda \frac{L_{\mathcal{T}}F(s)}{s}\,ds .
\end{align}
For fixed $\mu>0$, the kernel $1/s$ is bounded on $[\mu,\lambda]$, and $L_{\mathcal{T}}F$ converges uniformly to $F$ as the mesh size $|\mathcal{T}|$ tends to zero. Hence
\begin{align}
    \lim_{|\mathcal{T}|\to0} B_{\mu,\mathcal{T}}
    =
    F(\mu)
    +
    \int_\mu^\lambda \frac{F(s)}{s}\,ds .
\end{align}
It remains to let $\mu\downarrow0$. By convexity and $F(0)=0$, for $0\le s\le\mu$ one has
\begin{align}
    F(s)\le \frac{s}{\mu}F(\mu),
\end{align}
and therefore
\begin{align}
    0
    \le
    F(\mu)
    -
    \int_0^\mu \frac{F(s)}{s}\,ds
    \le
    F(\mu).
\end{align}
Since $F$ is continuous at $0$ and $F(0)=0$, the right-hand side tends to zero as $\mu\downarrow0$. Consequently,
\begin{align}
    \lim_{\mu\downarrow0}
    \lim_{|\mathcal{T}|\to0}
    B_{\mu,\mathcal{T}}
    =
    \int_0^\lambda \frac{F(s)}{s}\,ds .
\end{align}
Finally, since $\rho\le\lambda\sigma$, the upper tail in the finite-range representation of the relative entropy vanishes for $s\ge\lambda$. Thus \autoref{coro:jencova_finite_range} gives
\begin{align}
    D(\rho\Vert\sigma)
    =
    \left(\rho-\sigma\right)(1)
    +
    \int_0^\lambda \frac{F(s)}{s}\,ds
    +
    \rho(1)\ln\lambda
    -
    (\lambda-1)\sigma(1).
\end{align}
This proves that the upper bounds in \eqref{eq:relative_entropy_upper_bound} converge to $D(\rho\Vert\sigma)$ by first refining the partition of $[\mu,\lambda]$ and then taking $\mu\downarrow0$.
}
\end{proof}

\subsection{Application to the secrecy program}\label{subsec:npo_formulation}

We now apply \autoref{thm:upper_bound_relative_entropy} to the optimization problem \eqref{eq:commuting_operator_entropy_program_gns}, which, captures the asymptotic secrecy of a DIQKD protocol without loss of generality: the GNS reduction eliminates the optimization over arbitrary purifications, and the dilation theorem (\autoref{thm:dilation_theorem}) together with \autoref{cor:povm_pvm_equality} justifies the formulation in terms of the universal PVM algebras. Recall that, for a state $\psi\in\mathcal{S}\!\left(\mathcal{P}_A\otimes_{\max}\mathcal{P}_B\right)$ with GNS triple $\left(\pi_\psi,\mathcal{H}_\psi,\xi_\psi\right)$ and von Neumann algebra $\mathcal{M}_{AB,\psi} \coloneqq\pi_\psi\!\left(\mathcal{P}_A\otimes_{\max}\mathcal{P}_B\right)''$, the conditional von Neumann entropy is given by
\begin{align}
    H(Z|E)_\omega
    =
    -D\!\left(\omega_{ZE}\,\middle\Vert\,\tau_{\mathcal{Z}}\otimes\omega_E\right),
\end{align}
where $\omega_{ZE}\in\mathcal{S}\!\left(\ell^\infty(\mathcal{Z})\,\overline\otimes\,\mathcal{M}_E\right)$ is the post-measurement state, $\mathcal{M}_E\subseteq\mathcal{M}_{AB,\psi}'$ is the adversary's algebra, and $\omega_E$ denotes the marginal on $\mathcal{M}_E$.

By \autoref{lem:domination_and_entropy_bounds_cq}, the domination $\omega_{ZE}\le\tau_{\mathcal{Z}}\otimes\omega_E$ holds. Setting $\rho\coloneqq\omega_{ZE}$ and $\sigma\coloneqq\tau_{\mathcal{Z}}\otimes\omega_E$, we have $\lambda=1$ and $\mu=0$, so that \autoref{coro:jencova_finite_range} yields
\begin{align}\label{eq:relative_entropy_cq_explicit}
    D\!\left(\omega_{ZE}\,\middle\Vert\,\tau_{\mathcal{Z}}\otimes\omega_E\right)
    =
    \left(1-|\mathcal{Z}|\right)
    +
    \int_0^1 \frac{1}{s}\,
    \left(s\left(\tau_{\mathcal{Z}}\otimes\omega_E\right) - \omega_{ZE}\right)_+(1)
    \,ds,
\end{align}
since $\rho(1)\ln 1 = 0$ and $(\lambda-1)\sigma(1)=0$. Replacing the integral by the discretization of \autoref{thm:upper_bound_relative_entropy} turns \eqref{eq:commuting_operator_entropy_program_gns} into
\begin{equation}\label{eq:commuting_operator_entropy_program_discretized}
    \begin{aligned}
        \inf\ & \left(|\mathcal{Z}|-1\right) -  \sum_{k=0}^r
    \sup_{p_k\in\mathcal{P}\left(\ell^\infty(\mathcal{Z})\,\overline\otimes\,\mathcal{M}_E\right)}
    \left(\alpha_k\,\omega_{ZE} + \beta_k\left(\tau_{\mathcal{Z}}\otimes\omega_E\right)\right)(p_k)\\
        \text{s.t. }&
        \psi \in \mathcal{S}\!\left(\mathcal{P}_A\otimes_{\max}\mathcal{P}_B\right),\\
        &\psi\!\left(p_{a|x}\otimes q_{b|y}\right)=p(a,b|x,y),
        \qquad
        a,b,x,y,\\
        &\mathcal{M}_E \subseteq \mathcal{M}_{AB,\psi}',\\
        &\omega_{ZE}\in
        \mathcal{S}\!\left(\ell^\infty(\mathcal{Z})\,\overline\otimes\,\mathcal{M}_E\right),\\
        &\omega_{ZE}(p_z\otimes c)
        =
        \left\langle \xi_\psi,\,
        \pi_\psi(p_{z|\tilde x})\,c\,\xi_\psi
        \right\rangle,
        \qquad
        z\in\mathcal{Z},\ c\in \mathcal{M}_E.
    \end{aligned}
\end{equation}

\begin{remark}
The relaxation in $(228)$ is monotone in Eve's algebra. If $\mathcal M_E\subseteq \mathcal M'_{AB,\psi}$, then $\operatorname{Proj}(\ell^\infty(Z),\bar\otimes,M_E)\subseteq \operatorname{Proj}(\ell^\infty(Z),\bar\otimes,M'*{AB,\psi})$. Thus the projection suprema in $(228)$ can only increase when $M_E$ is enlarged to the full commutant. Since these terms appear with a minus sign, the objective can only decrease. Hence the infimum over all $\mathcal M_E\subseteq \mathcal M'_{AB,\psi}$ is unchanged by taking $\mathcal M_E=\mathcal M'_{AB,\psi}$.
\end{remark}

\begin{theorem}[Equivalence and convergence of the entropy hierarchy]\label{thm:equivalence_convergence}
Let $p(a,b\vert x,y)$ be a conditional distribution, $\tilde x\in\mathcal{X}$ a key-generation input, and $\mathcal{Z}\coloneqq\mathcal{A}_{\tilde x}$. For each $r\ge 1$, fix a partition $0<\mu_r=t_1<t_2<\cdots<t_r=1$ with mesh $\delta_r\coloneqq\max_{1\le k\le r-1}(t_{k+1}-t_k)$, and let $\alpha_0,\ldots,\alpha_r$, $\beta_0,\ldots,\beta_r$ be the coefficients from \autoref{thm:upper_bound_relative_entropy}. Define the following four quantities:
\begin{enumerate}
\item[\emph{(I)}] The \emph{conditional entropy value} in the commuting-operator model:
\begin{align}
    h^\star
    \coloneqq
    \text{value of } \eqref{eq:commuting_operator_entropy_program_gns}.
\end{align}

\item[\emph{(II)}] The \emph{discretized commuting-operator program} at level $r$:
\begin{align}
    h_r
    \coloneqq
    \text{value of } \eqref{eq:commuting_operator_entropy_program_discretized}.
\end{align}

\item[\emph{(III)}] The \emph{universal $\mathcal{C}^\star$-algebraic NPO} at level $r$: let 
\begin{align}
    \mathcal{Q}_r\coloneqq\mathcal{C}^\star\!\left(\{p_k^{(z)}\}_{0\le k\le r,\,z\in\mathcal{Z}}\;\middle\vert\;(p_k^{(z)})^2=p_k^{(z)}=(p_k^{(z)})^\star\right)  
\end{align}
be the universal $\mathcal{C}^\star$-algebra generated by $(r+1)\cdot|\mathcal{Z}|$ projections. Set
\begin{equation}\label{eq:universal_npo}
    \begin{aligned}
        h_r^{\operatorname{NPO}}
        \coloneqq
        \inf\ &
        |\mathcal{Z}|-1
        +
        \sum_{k=0}^r\sum_{z\in\mathcal{Z}}
        \omega\!\left(
        \left(-\alpha_k\,p_{z|\tilde x}-\beta_k\right)p_k^{(z)}
        \right)
        \\
        \text{s.t. }
        &\omega\in\mathcal{S}\!\left(
        \mathcal{P}_A\otimes_{\max}\mathcal{P}_B
        \otimes_{\max}\mathcal{Q}_r\right),\\
        &\omega\!\left(p_{a|x}\,q_{b|y}\right)=p(a,b|x,y),
        \qquad a,b,x,y.
    \end{aligned}
\end{equation}

\item[\emph{(IV)}] The \emph{NPA relaxation} at level $n$ of the program \eqref{eq:universal_npo}:
\begin{align}
    h_r^{(n)}
    \coloneqq
    \text{level-$n$ NPA outer relaxation of }
    \eqref{eq:universal_npo}.
\end{align}
\end{enumerate}
Then the following hold:
\begin{enumerate}
    \item[\emph{(a)}] $h_r = h_r^{\operatorname{NPO}}$. Moreover, $h_r^{\operatorname{NPO}}$ coincides with the value of the Hilbert-space NPO from \cite[Thm.~II.2]{kossmann2025reliableentropyestimationobserved}, formulated with commuting operators in place of tensor products.

    \item[\emph{(b)}] $h_r\le h^\star$ for every $r\ge 1$.

    \item[\emph{(c)}] $\displaystyle\lim_{r\to\infty}  \rs{\lim_{ \mu_r\to 0} h_r} = h^\star$.

    \item[\emph{(d)}] $h_r^{(n)}\le h_r$ for every $n\ge 1$, and $\displaystyle\lim_{n\to\infty} h_r^{(n)} = h_r$.

    \item[\emph{(e)}] $\displaystyle\lim_{r\to\infty}\lim_{n\to\infty} \rs{\lim_{ \mu_r\to 0}}   h_r^{(n)} = h^\star$.
\end{enumerate}
\end{theorem}

\begin{proof}
We first show $h_r = h_r^{\operatorname{NPO}}$. Let $\omega\in\mathcal{S}\!\left(\mathcal{P}_A\otimes_{\max}\mathcal{P}_B\otimes_{\max}\mathcal{Q}_r\right)$ be feasible for \eqref{eq:universal_npo}. By the universal property of $\otimes_{\max}$ (\autoref{prop:universal_property_maxtensorproduct}), the GNS representation $\left(\pi_\omega,\mathcal{H}_\omega,\xi_\omega\right)$ yields commuting $\star$-representations of $\mathcal{P}_A$, $\mathcal{P}_B$, and $\mathcal{Q}_r$ on $\mathcal{H}_\omega$. Setting $\mathcal{M}_{AB}\coloneqq\pi_\omega\!\left(\mathcal{P}_A\otimes_{\max}\mathcal{P}_B\right)''$, the images $\pi_\omega(p_k^{(z)})$ lie in $\mathcal{M}_{AB}'$ and are projections. With $\psi\coloneqq\omega\vert_{\mathcal{P}_A\otimes_{\max}\mathcal{P}_B}$ and $\mathcal{M}_E\coloneqq W^\star\!\left(\{\pi_\omega(p_k^{(z)})\}_{k,z}\right)\subseteq\mathcal{M}_{AB}'$, the tuple $\left(\psi,\mathcal{M}_E,\{\pi_\omega(p_k^{(z)})\}\right)$ is feasible for \eqref{eq:commuting_operator_entropy_program_discretized}. The block decomposition gives
\begin{align}
    \sup_{p_k\in\mathcal{P}\left(\ell^\infty(\mathcal{Z})\,\overline\otimes\,\mathcal{M}_E\right)}
    \left(\alpha_k\,\omega_{ZE}+\beta_k\left(\tau_{\mathcal{Z}}\otimes\omega_E\right)\right)(p_k)
    =
    \sum_{z\in\mathcal{Z}}
    \sup_{p_k^{(z)}\in\mathcal{P}(\mathcal{M}_E)}
    \left(\alpha_k\,\omega_E^z+\beta_k\,\omega_E\right)(p_k^{(z)}).
\end{align}
The supremum on the right is attained at the support projection of $\left(\alpha_k\,\omega_E^z+\beta_k\,\omega_E\right)_+$, which exists in $\mathcal{M}_E$ by the Jordan decomposition \eqref{eq:positive_part_as_sup}, and the resulting value equals
\begin{align}
    -\omega\!\left(
    \left(-\alpha_k\,p_{z|\tilde x}-\beta_k\right)p_k^{(z)}\right)
\end{align}
precisely when $p_k^{(z)}$ is the optimal projection. Since the program \eqref{eq:commuting_operator_entropy_program_discretized} optimizes over all projections in the commutant, while \eqref{eq:universal_npo} optimizes over all states on the maximal tensor product, both infima range over the same set of expectation values, and hence $h_r=h_r^{\operatorname{NPO}}$. Conversely, any feasible tuple for \eqref{eq:commuting_operator_entropy_program_discretized} defines, via the universal property of $\otimes_{\max}$, a state on $\mathcal{P}_A\otimes_{\max}\mathcal{P}_B\otimes_{\max}\mathcal{Q}_r$ that is feasible for \eqref{eq:universal_npo} with the same objective value. The identification with \cite[Thm.~II.2]{kossmann2025reliableentropyestimationobserved} follows from the same universal property: the commuting-operator constraints in \cite[eq.~(16)]{kossmann2025reliableentropyestimationobserved}
\begin{align}
    \left[M_{a|x},N_{b|y}\right]
    =\left[M_{a|x},P_k^{(z)}\right]
    =\left[N_{b|y},P_k^{(z)}\right]=0
\end{align}
are exactly the relations encoded by the maximal tensor product $\mathcal{P}_A\otimes_{\max}\mathcal{P}_B\otimes_{\max}\mathcal{Q}_r$, and the passage from POVMs to PVMs is justified by \autoref{thm:dilation_theorem} together with \autoref{cor:povm_pvm_equality}. This proves (a).

For (b), let $\psi$ be feasible. By \autoref{lem:domination_and_entropy_bounds_cq} we have $\omega_{ZE}\le\tau_{\mathcal{Z}}\otimes\omega_E$, so $\mu=0$ and $\lambda=1$ are admissible in \autoref{thm:upper_bound_relative_entropy}. By \eqref{eq:relative_entropy_upper_bound},
\begin{align}
    D\!\left(\omega_{ZE}\,\middle\Vert\,\tau_{\mathcal{Z}}\otimes\omega_E\right)
    \le
    \left(1-|\mathcal{Z}|\right)
    +
    \sum_{k=0}^r
    \sup_{p_k\in\mathcal{P}\left(\ell^\infty(\mathcal{Z})\,\overline\otimes\,\mathcal{M}_E\right)}
    \left(\alpha_k\,\omega_{ZE}+\beta_k\left(\tau_{\mathcal{Z}}\otimes\omega_E\right)\right)(p_k),
\end{align}
whence $H(Z|E)_\omega=-D\!\left(\omega_{ZE}\,\middle\Vert\,\tau_{\mathcal{Z}}\otimes\omega_E\right)$ satisfies
\begin{align}
    H(Z|E)_\omega
    \ge
    \left(|\mathcal{Z}|-1\right)
    -
    \sum_{k=0}^r
    \sup_{p_k}
    \left(\alpha_k\,\omega_{ZE}+\beta_k\left(\tau_{\mathcal{Z}}\otimes\omega_E\right)\right)(p_k).
\end{align}
Since this holds for every feasible $\psi$, taking the infimum over $\psi$ on both sides yields $h_r\le h^\star$.

For (c), the discretization error for a single feasible state $\psi$ is
\begin{align}\label{eq:discretization_error}
    0
    \le
    H(Z|E)_\psi - L_r(\psi)
    =
    \sum_{k=0}^r c_k\,F_\psi(t_k)
    -
    \int_{\mu_r}^1\frac{F_\psi(s)}{s}\,ds,
\end{align}
where $L_r(\psi)$ denotes the objective of \eqref{eq:commuting_operator_entropy_program_discretized} evaluated at $\psi$, and
\begin{align}
    F_\psi(s)
    \coloneqq
    \left(s\left(\tau_{\mathcal{Z}}\otimes\omega_E\right)
    -\omega_{ZE}\right)_+(1),
    \qquad s\in[0,1],
\end{align}
is the function appearing in the Frenkel--Jen\v{c}ov\'a formula \eqref{eq:frenkel_jencova_general}. We claim that the right-hand side of \eqref{eq:discretization_error} is bounded uniformly over all feasible $\psi$. 
\rs{Indeed, for any feasible $\psi$, the variational characterization of the positive part of a self-adjoint normal functional gives
$$
F_\psi(s)=\sup_{0\leq q\leq 1}\bigl(s(\tau_Z\otimes\omega_E)-\omega_{ZE}\bigr)(q).
$$
Hence, for $s\geq s'$,
$$
F_\psi(s)\leq F_\psi(s')+(s-s')(\tau_Z\otimes\omega_E)(1)=F_\psi(s')+(s-s')|\mathcal Z|.
$$
Interchanging $s$ and $s'$ gives
$$
|F_\psi(s)-F_\psi(s')|\leq |s-s'|\,|\mathcal Z|
$$
for all $s,s'\in[0,1]$ and all feasible $\psi$. Thus the family $\{F_\psi:[0,1]\to\mathbb R_{\geq0}\mid \psi\ \mathrm{feasible}\}$ is uniformly bounded and equi-Lipschitz.
}

\rausdamit{
Indeed, for any feasible $\psi$, the function $F_\psi$ is convex with $F_\psi(0)=0$. Moreover, for every $s\in[0,1]$ and every projection $p$,
\begin{align}
   s\left(\tau_{\mathcal{Z}}\otimes\omega_E\right)(p)
    -\omega_{ZE}(p)
    \le
    s\left(\tau_{\mathcal{Z}}\otimes\omega_E\right)(p)
    \le
    s\,|\mathcal{Z}|,
\end{align}
so $F_\psi(s)\le s\,|\mathcal{Z}|$ for all $s\in[0,1]$. Therefore, the family
\begin{align}
    \mathscr{F}
    \coloneqq
    \left\{F_\psi : [0,1]\to\mathbb{R}_{\ge 0}
    \ \middle\vert\
    \psi\ \text{feasible}\right\}
\end{align}
is uniformly bounded and equi-Lipschitz:
\begin{align}\label{eq:equi_lipschitz}
    |F_\psi(s)-F_\psi(s')|
    \le
    |s-s'|\cdot|\mathcal{Z}|
    \qquad
    \text{for all }s,s'\in[0,1]
    \ \text{and all feasible }\psi.
\end{align}
}
Since the quadrature from \autoref{thm:upper_bound_relative_entropy} provides an upper bound for $\int_0^1 F(s)/s\,ds$ via a piecewise-linear interpolant of any convex $F$, and the quality of this approximation depends only on the mesh $\delta_r$ and the Lipschitz constant of $F$ (cf.~\cite{kossmann2025reliableentropyestimationobserved}), there exists a function $\varepsilon(\delta_r,|\mathcal{Z}|)\ge 0$ with $\varepsilon(\delta_r,|\mathcal{Z}|)\to 0$ as $\delta_r\to 0$, such that
\begin{align}
    \sup_{\psi\ \text{feasible}}
    \left(H(Z|E)_\psi - L_r(\psi)\right)
    \le
    \varepsilon(\delta_r,|\mathcal{Z}|).
\end{align}
From this uniform bound we conclude: on one hand, $h_r\le h^\star$ by (b); on the other, $L_r(\psi)\ge H(Z|E)_\psi - \varepsilon(\delta_r,|\mathcal{Z}|)$ for every feasible $\psi$, so
\begin{align}
    h_r
    =
    \inf_\psi L_r(\psi)
    \ge
    \inf_\psi H(Z|E)_\psi
    - \varepsilon(\delta_r,|\mathcal{Z}|)
    =
    h^\star - \varepsilon(\delta_r,|\mathcal{Z}|).
\end{align}
Therefore $|h^\star - h_r|\le\varepsilon(\delta_r,|\mathcal{Z}|)\to 0$ as $r\to\infty$ and \rs{$\mu_r\to 0 $}, which proves (c).

For (d), by (a) the value $h_r = h_r^{\operatorname{NPO}}$ is the value of a non-commutative polynomial optimization over states on the $\mathcal{C}^\star$-algebra $\mathcal{P}_A\otimes_{\max}\mathcal{P}_B\otimes_{\max}\mathcal{Q}_r$. Since all generators are contractions, this algebra is Archimedean, and \cite[Lem.~4]{Ligthart_2023} (cf.~also~\cite{Navascus2008}) guarantees that the NPA hierarchy at level $n$ provides outer approximations satisfying $h_r^{(n)}\le h_r$ with $h_r^{(n)}\nearrow h_r$ as $n\to\infty$.

Finally, (e) is immediate from (c) and (d):
\begin{align}
    h^\star =  \lim_{r\to\infty} \rs{\lim_{\mu_r\to 0}} h_r =  \lim_{r\to\infty} \lim_{\mu_r\to 0} \lim_{n\to\infty} h_r^{(n)}.
\end{align}
\end{proof}

\section{Outlook}
In this work, we have provided decisive tools for a possible route towards a fully operator-algebraic formulation of device-independent quantum key distribution. At the same time, the results presented here should be understood only as a first step. They show that the secrecy-relevant entropy optimization can be formulated consistently in the commuting-operator framework and connected to NPA-type relaxations, but they also point to several directions in which the operator-algebraic approach could be developed further.

A first natural direction concerns the cryptographic reduction from many rounds to a single-round entropy quantity. The asymptotic equipartition theorem used in this work is available in the von Neumann algebraic setting and therefore fits naturally into the commuting-operator framework. However, modern DIQKD security proofs often rely on more refined tools, such as generalized entropy accumulation theorems and techniques based on Rényi $\alpha$-norms. The language of von Neumann algebras and noncommutative $L_p$-spaces appears well suited for formulating such quantities beyond the tensor-product Hilbert-space setting. In particular, an analogue of the leftover-hash lemma with the required Rényi-type one-shot quantities in the commuting-operator framework does not seem to be available. Establishing such results would be an important step towards finite-size DIQKD security proofs that are fully intrinsic to the commuting-operator model.

A second direction is the development of parallel-repetition and entropy-accumulation principles directly at the level of commuting-operator strategies. In the language of non-local games, parallel repetition is a central tool for understanding how correlations behave under independent repetitions of an experiment. For DIQKD, an operator-algebraic parallel-repetition theorem would provide a structural bridge between the single-round Bell-type description and the multi-round cryptographic protocol. Such a theorem could clarify to what extent the tensor-product assumptions usually imposed between rounds can be replaced by intrinsic algebraic independence or product-state assumptions.

A third perspective is obtained by reversing the logic used in the present paper. We started from the commuting-operator model and showed how to recover an entropy optimization that can be represented on concrete Hilbert spaces when needed. More generally, any theorem proved at the level of universal $C^*$-algebras or von Neumann algebras immediately yields corresponding device-dependent statements by choosing a concrete representation on Hilbert spaces. This suggests a possible route towards unified security proofs: rather than proving separate theorems for each concrete physical model, one proves a single representation-independent theorem at the algebraic level and obtains the usual Hilbert-space versions as representations. Such an approach could also make security proofs more transparent and certifiable, since the assumptions entering the proof are encoded explicitly in the algebra on which the theorem is formulated.

More broadly, the $C^*$-algebraic viewpoint offers a systematic way of comparing different sets of security assumptions. Each model of a cryptographic experiment corresponds to an algebra generated by the observables that are assumed to exist, together with the relations they are assumed to satisfy. Stronger or weaker security assumptions then become relations between algebras, for instance through quotients, embeddings, or tensor-product choices. From this perspective, a minimal requirement for a mathematically transparent security proof would be that the proof specifies at the outset the algebra in which it is carried out. Once this algebra is made explicit, different device-dependent, semi-device-independent, and device-independent assumptions can be compared and ordered within a common mathematical language.

We therefore view the present work not as the conclusive treatment of the commuting-operator approach to DIQKD, but rather as an initial step towards a unified and systematic view on quantum security.

\section*{Acknowledgements}
\autoref{thm:frenkel_general} has been shown independently by Correa da Silva, Fr{\"o}b, Lechner, and L. Sangaletti \cite{CorreaDaSilva_fDivergences}. 
GK thanks Alexander Stottmeister for a discussion at an early stage of the project. GK acknowledges support from the Excellence Cluster - Matter and Light for Quantum Computing (ML4Q-2) and by the European Research Council (ERC Grant Agreement No. 948139). R.S.\ is supported  by the DFG under Germany's Excellence Strategy - EXC-2123 QuantumFrontiers - 390837967 and  SFB 1227 (DQ-mat), the Quantum Valley Lower Saxony, and the BMBF projects CBQD, SEQUIN, Quics and ATIQ.  
PL and HC are supported by the Emerging Young Scholars Program of the National Science and Technology Council, Taiwan (R.O.C.) under Grant numbers~NSTC 114-2628-E-002-006, NSTC 114-2119-M-001-002, and NSTC 114-2124-M-002-003, by the Yushan Young Scholar Program of the Ministry of Education, Taiwan (R.O.C.) under Grant number~NTU-114V2016-1, and by the research project `Forefront Quantum Computing, Learning, and Engineering in Noisy Intermediate-Scale Quantum Era’ of National Taiwan University under Grant NTU-114L895005. H.-C.~Cheng acknowledges the support from the `Center for Advanced Computing and Imaging in Biomedicine (NTU-115L900702)' through The Featured Areas Research Center Program within the framework of the Higher Education Sprout Project by the Ministry of Education (MOE) in Taiwan.
The authors acknowledge the use of \cite{claude2026}.
\bibliography{main}

\bibliographystyle{ultimate}
\appendix

\end{document}

%% file: abstract2.tex
\begin{abstract}
    Device-independent quantum key distribution (DIQKD) is arguably the gold standard for secure quantum communication, as it aims to rely only on observed input-output statistics of an uncharacterized device which is only assumption to obey the laws of quantum physics. A corresponding security analysis hence demands a description of a quantum experiment from a most general perspective.  Under close inspection, existing proof techniques do not always meet this goal as they tend to rely on subtile assumptions on a tensor product structure of the underlying Hilbert space and a 'hidden but finite' dimensionality.  
    In this work, we collect the tools needed for a full analysis of DIQKD in the commuting operator framework, which avoids these subtilities and provides the arguably most general view on a quantum experiment. 
    We rigorously proof the common assumption that in DIQKD measurements can be w.l.o.g. assumed to be projective. Furthermore, we show that task of computing key rates can be casted as a non-commutative polynomial optimization (NPO) problem to which the Navascu\'es--Pironio--Ac\'in (NPA) hierarchy gives a correct and converging relaxation. 
    
    As a tool, we generalize the integral representation for the relative entropy by \emph{Frenkel} \cite{Frenkel2023} to general von Neumann algebras and apply techniques from \emph{Ko\ss{}mann and Schwonnek} \cite{kossmann2025reliableentropyestimationobserved} for the approximation in an NPO program. 
\end{abstract}

%% file: figure-qkd.tex
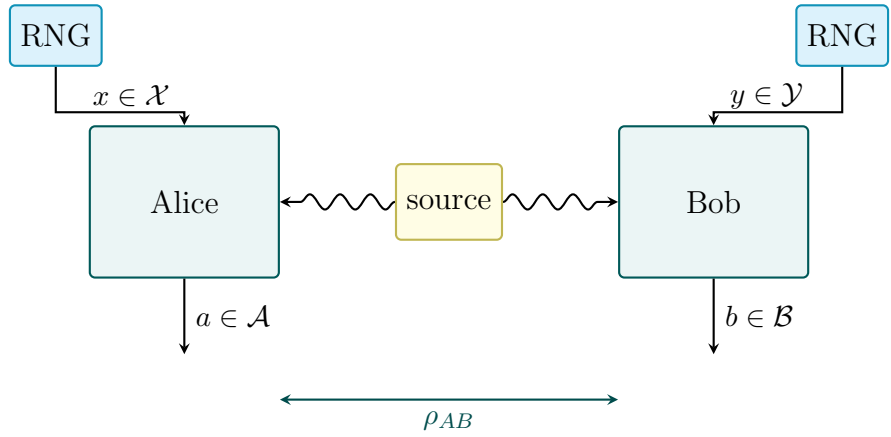
\begin{figure}
    \centering
\begin{tikzpicture}[
    >=stealth,
    box/.style={
        draw, thick, rounded corners=2pt,
        minimum width=2.5cm, minimum height=2cm,
        fill=teal!8, draw=teal!70!black,
    },
    src/.style={
        draw, thick, rounded corners=2pt,
        minimum width=1.4cm, minimum height=1cm,
        fill=yellow!12, draw=yellow!70!black,
    },
    rng/.style={
        draw, thick, rounded corners=2pt,
        minimum width=1.2cm, minimum height=0.8cm,
        fill=cyan!12, draw=cyan!70!black,
    },
]
\node[box] (A) at (-3.5, 0) {Alice};
\node[box] (B) at ( 3.5, 0) {Bob};
\node[src] (S) at (0, 0) {source};
\node[rng] (RA) at (-5.2, 2.2) {RNG};
\node[rng] (RB) at ( 5.2, 2.2) {RNG};
\draw[->, thick, decorate, decoration={snake, amplitude=1mm, segment length=4mm, post length=2mm}]
    (S.west) -- (A.east);
\draw[->, thick, decorate, decoration={snake, amplitude=1mm, segment length=4mm, post length=2mm}]
    (S.east) -- (B.west);
\draw[->, thick] (RA.south) |- ++(0,-0.6) -| (A.north);
\draw[->, thick] (RB.south) |- ++(0,-0.6) -| (B.north);
\node[font=\small] at (-4.2, 1.4) {$x \in \mathcal{X}$};
\node[font=\small] at ( 4.2, 1.4) {$y \in \mathcal{Y}$};
\draw[->, thick] (A.south) -- ([yshift=-1cm]A.south)
    node[midway, right, font=\small] {$a \in \mathcal{A}$};
\draw[->, thick] (B.south) -- ([yshift=-1cm]B.south)
    node[midway, right, font=\small] {$b \in \mathcal{B}$};
\draw[<->, thick, teal!60!black]
    ([yshift=-1.6cm]A.south east) -- ([yshift=-1.6cm]B.south west)
    node[midway, below, font=\small] {$\rho_{AB}$};
\end{tikzpicture}
    \caption{The figure illustrates a device-independent quantum key distribution (DIQKD) experiment. In contrast to the quantum game depicted in \autoref{fig:quantum_game}, the referee is replaced by local random number generators (RNGs), one on each side. While sharing a quantum state $\rho_{AB}$, Alice and Bob use their respective RNGs to draw private and local random inputs, feed them into their devices, and collect the resulting output statistics. }
    \label{fig:diqkd_sketch}
\end{figure}

%% file: figure-games.tex
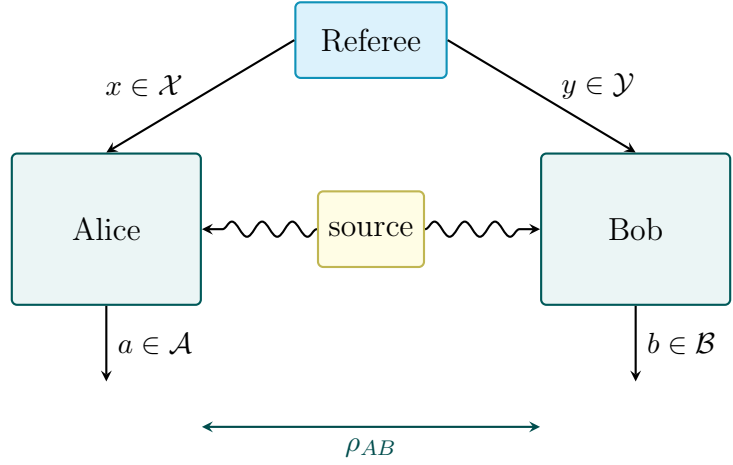
\begin{figure}
   \centering
\begin{tikzpicture}[
    >=stealth,
    box/.style={
        draw, thick, rounded corners=2pt,
        minimum width=2.5cm, minimum height=2cm,
        fill=teal!8, draw=teal!70!black,
    },
    src/.style={
        draw, thick, rounded corners=2pt,
        minimum width=1.4cm, minimum height=1cm,
        fill=yellow!12, draw=yellow!70!black,
    },
    ref/.style={
        draw, thick, rounded corners=2pt,
        minimum width=2cm, minimum height=1cm,
        fill=cyan!12, draw=cyan!70!black,
    },
]
\node[box] (A) at (-3.5, 0) {Alice};
\node[box] (B) at ( 3.5, 0) {Bob};
\node[src] (S) at (0, 0) {source};
\node[ref] (R) at (0, 2.5) {Referee};
\draw[->, thick, decorate, decoration={snake, amplitude=1mm, segment length=4mm, post length=2mm}]
    (S.west) -- (A.east);
\draw[->, thick, decorate, decoration={snake, amplitude=1mm, segment length=4mm, post length=2mm}]
    (S.east) -- (B.west);
\draw[->, thick] (R.west) -- (A.north);
\draw[->, thick] (R.east) -- (B.north);
\node[font=\small] at (-3, 1.9) {$x \in \mathcal{X}$};
\node[font=\small] at ( 3, 1.9) {$y \in \mathcal{Y}$};
\draw[->, thick] (A.south) -- ([yshift=-1cm]A.south)
    node[midway, right, font=\small] {$a \in \mathcal{A}$};
\draw[->, thick] (B.south) -- ([yshift=-1cm]B.south)
    node[midway, right, font=\small] {$b \in \mathcal{B}$};
\draw[<->, thick, teal!60!black]
    ([yshift=-1.6cm]A.south east) -- ([yshift=-1.6cm]B.south west)
    node[midway, below, font=\small] {$\rho_{AB}$};
\end{tikzpicture}
    \caption{The figure illustrates a non-local quantum game. A source prepares a bipartite state $\rho_{AB}$ and distributes it between Alice and Bob. Following a predetermined strategy, specified by the shared state together with their local measurements, Alice and Bob receive questions $x \in \mathcal{X}$ and $y \in \mathcal{Y}$ from a referee, feed them as inputs into their respective devices, and respond with the outcomes $a \in \mathcal{A}$ and $b \in \mathcal{B}$ produced by their local boxes.}
    \label{fig:quantum_game}
\end{figure}

%% file: main.bib
@article{Palazuelos2016,
  title = {Survey on nonlocal games and operator space theory},
  volume = {57},
  ISSN = {1089-7658},
  url = {http://dx.doi.org/10.1063/1.4938052},
  DOI = {10.1063/1.4938052},
  number = {1},
  journal = {Journal of Mathematical Physics},
  publisher = {AIP Publishing},
  author = {Palazuelos,  Carlos and Vidick,  Thomas},
  year = {2016},
  month = jan 
}

@article{Bell1964,
  title = {On the Einstein Podolsky Rosen paradox},
  volume = {1},
  ISSN = {0554-128X},
  url = {http://dx.doi.org/10.1103/PhysicsPhysiqueFizika.1.195},
  DOI = {10.1103/physicsphysiquefizika.1.195},
  number = {3},
  journal = {Physics Physique},
  publisher = {American Physical Society (APS)},
  author = {Bell,  J. S.},
  year = {1964},
  month = nov,
  pages = {195–200}
}

@article{Ekert1991,
  title = {Quantum cryptography based on Bell’s theorem},
  volume = {67},
  ISSN = {0031-9007},
  url = {http://dx.doi.org/10.1103/PhysRevLett.67.661},
  DOI = {10.1103/physrevlett.67.661},
  number = {6},
  journal = {Physical Review Letters},
  publisher = {American Physical Society (APS)},
  author = {Ekert,  Artur K.},
  year = {1991},
  month = aug,
  pages = {661–663}
}

@inproceedings{Mayers,
  series = {SFCS-98},
  title = {Quantum cryptography with imperfect apparatus},
  url = {http://dx.doi.org/10.1109/SFCS.1998.743501},
  DOI = {10.1109/sfcs.1998.743501},
  booktitle = {Proceedings 39th Annual Symposium on Foundations of Computer Science (Cat. No.98CB36280)},
  publisher = {IEEE Comput. Soc},
  author = {Mayers,  D. and Yao,  A.},
  pages = {503–509},
  collection = {SFCS-98},
  year = {1998}
}

@article{Barrett2005,
  title = {No Signaling and Quantum Key Distribution},
  volume = {95},
  ISSN = {1079-7114},
  url = {http://dx.doi.org/10.1103/PhysRevLett.95.010503},
  DOI = {10.1103/physrevlett.95.010503},
  number = {1},
  journal = {Physical Review Letters},
  publisher = {American Physical Society (APS)},
  author = {Barrett,  Jonathan and Hardy,  Lucien and Kent,  Adrian},
  year = {2005},
  month = jun 
}

@article{Acn2006_1,
  title = {From Bell’s Theorem to Secure Quantum Key Distribution},
  volume = {97},
  ISSN = {1079-7114},
  url = {http://dx.doi.org/10.1103/PhysRevLett.97.120405},
  DOI = {10.1103/physrevlett.97.120405},
  number = {12},
  journal = {Physical Review Letters},
  publisher = {American Physical Society (APS)},
  author = {Acín,  Antonio and Gisin,  Nicolas and Masanes,  Lluis},
  year = {2006},
  month = sep 
}

@article{Acn2006_2,
  title = {Efficient quantum key distribution secure against no-signalling eavesdroppers},
  volume = {8},
  ISSN = {1367-2630},
  url = {http://dx.doi.org/10.1088/1367-2630/8/8/126},
  DOI = {10.1088/1367-2630/8/8/126},
  number = {8},
  journal = {New Journal of Physics},
  publisher = {IOP Publishing},
  author = {Acín,  Antonio and Massar,  Serge and Pironio,  Stefano},
  year = {2006},
  month = aug,
  pages = {126–126}
}

@article{Pironio2009,
  title = {Device-independent quantum key distribution secure against collective attacks},
  volume = {11},
  ISSN = {1367-2630},
  url = {http://dx.doi.org/10.1088/1367-2630/11/4/045021},
  DOI = {10.1088/1367-2630/11/4/045021},
  number = {4},
  journal = {New Journal of Physics},
  publisher = {IOP Publishing},
  author = {Pironio,  Stefano and Acín,  Antonio and Brunner,  Nicolas and Gisin,  Nicolas and Massar,  Serge and Scarani,  Valerio},
  year = {2009},
  month = apr,
  pages = {045021}
}

@article{Vazirani2014,
  title = {Fully Device-Independent Quantum Key Distribution},
  volume = {113},
  ISSN = {1079-7114},
  url = {http://dx.doi.org/10.1103/PhysRevLett.113.140501},
  DOI = {10.1103/physrevlett.113.140501},
  number = {14},
  journal = {Physical Review Letters},
  publisher = {American Physical Society (APS)},
  author = {Vazirani,  Umesh and Vidick,  Thomas},
  year = {2014},
  month = sep 
}

@article{Acn2007,
  title = {Device-Independent Security of Quantum Cryptography against Collective Attacks},
  volume = {98},
  ISSN = {1079-7114},
  url = {http://dx.doi.org/10.1103/PhysRevLett.98.230501},
  DOI = {10.1103/physrevlett.98.230501},
  number = {23},
  journal = {Physical Review Letters},
  publisher = {American Physical Society (APS)},
  author = {Acín,  Antonio and Brunner,  Nicolas and Gisin,  Nicolas and Massar,  Serge and Pironio,  Stefano and Scarani,  Valerio},
  year = {2007},
  month = jun 
}

@article{Lobo2024,
   title={Certifying long-range quantum correlations through routed Bell tests},
   volume={8},
   ISSN={2521-327X},
   url={http://dx.doi.org/10.22331/q-2024-05-02-1332},
   DOI={10.22331/q-2024-05-02-1332},
   journal={Quantum},
   publisher={Verein zur Forderung des Open Access Publizierens in den Quantenwissenschaften},
   author={Lobo, Edwin Peter and Pauwels, Jef and Pironio, Stefano},
   year={2024},
   month=may, pages={1332} }

@article{Boca1991,
  title = {Free products of completely positive maps and spectral sets},
  volume = {97},
  ISSN = {0022-1236},
  url = {http://dx.doi.org/10.1016/0022-1236(91)90001-L},
  DOI = {10.1016/0022-1236(91)90001-l},
  number = {2},
  journal = {Journal of Functional Analysis},
  publisher = {Elsevier BV},
  author = {Boca,  Florin},
  year = {1991},
  month = may,
  pages = {251–263}
}

@article{Hensen2015,
  title = {Loophole-free Bell inequality violation using electron spins separated by 1.3 kilometres},
  volume = {526},
  ISSN = {1476-4687},
  url = {http://dx.doi.org/10.1038/nature15759},
  DOI = {10.1038/nature15759},
  number = {7575},
  journal = {Nature},
  publisher = {Springer Science and Business Media LLC},
  author = {Hensen,  B. and Bernien,  H. and Dréau,  A. E. and Reiserer,  A. and Kalb,  N. and Blok,  M. S. and Ruitenberg,  J. and Vermeulen,  R. F. L. and Schouten,  R. N. and Abellán,  C. and Amaya,  W. and Pruneri,  V. and Mitchell,  M. W. and Markham,  M. and Twitchen,  D. J. and Elkouss,  D. and Wehner,  S. and Taminiau,  T. H. and Hanson,  R.},
  year = {2015},
  month = oct,
  pages = {682–686}
}

@book{Blackadar2006,
  title = {Operator Algebras},
  ISBN = {9783540285175},
  ISSN = {0938-0396},
  url = {http://dx.doi.org/10.1007/3-540-28517-2},
  DOI = {10.1007/3-540-28517-2},
  journal = {Encyclopaedia of Mathematical Sciences},
  publisher = {Springer Berlin Heidelberg},
  author = {Blackadar,  Bruce},
  year = {2006}
}

@article{Davidson2018,
  title = {A proof of Boca’s Theorem},
  volume = {149},
  ISSN = {1473-7124},
  url = {http://dx.doi.org/10.1017/prm.2018.50},
  DOI = {10.1017/prm.2018.50},
  number = {04},
  journal = {Proceedings of the Royal Society of Edinburgh: Section A Mathematics},
  publisher = {Cambridge University Press (CUP)},
  author = {Davidson,  Kenneth R. and Kakariadis,  Evgenios T. A.},
  year = {2018},
  month = dec,
  pages = {869–876}
}

@article{Tsirelson1993,
  author  = {Tsirelson, B. S.},
  title   = {Some results and problems on quantum {B}ell-type inequalities},
  journal = {Hadronic Journal Supplement},
  year    = {1993},
  volume  = {8},
  pages   = {329--345},
}

@inproceedings{umegaki1962conditional,
  title={Conditional expectation in an operator algebra, IV (entropy and information)},
  author={Umegaki, Hisaharu},
  booktitle={Kodai Mathematical Seminar Reports},
  volume={14},
  pages={59--85},
  year={1962},
  organization={Department of Mathematics, Tokyo Institute of Technology}
}

@misc{claude2026,
  author       = {{Anthropic}},
  title        = {Claude (Claude Fable 5) [Large language model]},
  year         = {2026},
  howpublished = {\url{https://claude.ai}},
  note         = {Zugriff am 11. Juni 2026}
}

@article{kosaki_1986_variationalformula,
 ISSN = {03794024, 18417744},
 URL = {http://www.jstor.org/stable/24714805},
 author = {HIDEKI KOSAKI},
 journal = {Journal of Operator Theory},
 number = {2},
 pages = {335--348},
 publisher = {Theta Foundation},
 title = {Relative entropy of states: a variational expression},
 urldate = {2025-06-12},
 volume = {16},
 year = {1986}
}

@book{Ohya1993,
  title = {Quantum Entropy and Its Use},
  ISBN = {9783642579974},
  url = {http://dx.doi.org/10.1007/978-3-642-57997-4},
  DOI = {10.1007/978-3-642-57997-4},
  publisher = {Springer Berlin Heidelberg},
  author = {Ohya,  Masanori and Petz,  Dénes},
  year = {1993}
}

@misc{MIPstar=RE,
  doi = {10.48550/ARXIV.2001.04383},
  url = {https://arxiv.org/abs/2001.04383},
  author = {Ji,  Zhengfeng and Natarajan,  Anand and Vidick,  Thomas and Wright,  John and Yuen,  Henry},
  title = {MIP*=RE},
  publisher = {arXiv},
  year = {2020},
  copyright = {arXiv.org perpetual,  non-exclusive license}
}

@book{Takesaki1979,
  author    = {Masamichi Takesaki},
  title     = {Theory of Operator Algebras I},
  ISBN      = {9781461261889},
  url       = {http://dx.doi.org/10.1007/978-1-4612-6188-9},
  DOI       = {10.1007/978-1-4612-6188-9},
  publisher = {Springer New York},
  year      = {1979}
}

@book{Takesaki2003,
  title = {Theory of Operator Algebras II},
  ISBN = {9783662104514},
  ISSN = {0938-0396},
  url = {http://dx.doi.org/10.1007/978-3-662-10451-4},
  DOI = {10.1007/978-3-662-10451-4},
  journal = {Encyclopaedia of Mathematical Sciences},
  publisher = {Springer Berlin Heidelberg},
  author = {Takesaki,  Masamichi},
  year = {2003}
}

@misc{cheng2025errorexponentsquantumpacking,
      title={Error Exponents for Quantum Packing Problems via An Operator Layer Cake Theorem}, 
      author={Hao-Chung Cheng and Po-Chieh Liu},
      year={2025},
      eprint={2507.06232},
      archivePrefix={arXiv},
      primaryClass={quant-ph},
      url={https://arxiv.org/abs/2507.06232}, 
}

@misc{liu2025layercakerepresentationsquantum,
      title={Layer Cake Representations for Quantum Divergences}, 
      author={Po-Chieh Liu and Christoph Hirche and Hao-Chung Cheng},
      year={2025},
      eprint={2507.07065},
      archivePrefix={arXiv},
      primaryClass={quant-ph},
      url={https://arxiv.org/abs/2507.07065}, 
}

@article{Araki1975,
  author  = {Huzihiro Araki},
  title   = {Relative Entropy of States of von Neumann Algebras},
  journal = {Publications of the Research Institute for Mathematical Sciences},
  volume  = {11},
  number  = {3},
  pages   = {809--833},
  year    = {1975},
  doi     = {10.2977/PRIMS/1195191148},
  url     = {https://ems.press/journals/prims/articles/2800}
}

@article{Bennett1992a,
  author = {Bennett, Charles H. and Bessette, Fran\c{c}ois and Brassard, Gilles and Salvail, Louis and Smolin, John},
  title = {Experimental quantum cryptography},
  journal = {Journal of Cryptology},
  year = {1992},
  volume = {5},
  number = {1},
  pages = {3–28},
  publisher = {Springer Science and Business Media LLC},
  doi = {10.1007/bf00191318},
  url = {http://dx.doi.org/10.1007/BF00191318},
  issn = {1432-1378},
  month = {jan}
}

@article{Bennett1992b,
  title = {Quantum cryptography without Bell’s theorem},
  volume = {68},
  ISSN = {0031-9007},
  url = {http://dx.doi.org/10.1103/PhysRevLett.68.557},
  DOI = {10.1103/physrevlett.68.557},
  number = {5},
  journal = {Physical Review Letters},
  publisher = {American Physical Society (APS)},
  author = {Bennett,  Charles H. and Brassard,  Gilles and Mermin,  N. David},
  year = {1992},
  month = feb,
  pages = {557–559}
}

@article{Scarani2009,
  title = {The security of practical quantum key distribution},
  volume = {81},
  ISSN = {1539-0756},
  url = {http://dx.doi.org/10.1103/RevModPhys.81.1301},
  DOI = {10.1103/revmodphys.81.1301},
  number = {3},
  journal = {Reviews of Modern Physics},
  publisher = {American Physical Society (APS)},
  author = {Scarani,  Valerio and Bechmann-Pasquinucci,  Helle and Cerf,  Nicolas J. and Dušek,  Miloslav and L\"{u}tkenhaus,  Norbert and Peev,  Momtchil},
  year = {2009},
  month = sep,
  pages = {1301–1350}
}

@book{vonNeumann1996,
  title = {Mathematische Grundlagen der Quantenmechanik},
  ISBN = {9783642614095},
  url = {http://dx.doi.org/10.1007/978-3-642-61409-5},
  DOI = {10.1007/978-3-642-61409-5},
  publisher = {Springer Berlin Heidelberg},
  author = {von Neumann,  John},
  year = {1996}
}

@book{BaezSegalZhou1992,
  author    = {John C. Baez and Irving E. Segal and Zhengfang Zhou},
  title     = {Introduction to Algebraic and Constructive Quantum Field Theory},
  publisher = {Princeton University Press},
  address   = {Princeton, NJ},
  year      = {1992},
  isbn      = {978-0-691-08546-3},
  isbn13    = {9780691085463}
}

@article{RENNER2008,
  author = {Renner, Renato},
  title = {SECURITY OF QUANTUM KEY DISTRIBUTION},
  journal = {International Journal of Quantum Information},
  year = {2008},
  volume = {06},
  number = {01},
  pages = {1–127},
  publisher = {World Scientific Pub Co Pte Lt},
  doi = {10.1142/s0219749908003256},
  url = {http://dx.doi.org/10.1142/S0219749908003256},
  issn = {1793-6918},
  month = {feb}
}

@article{Navascus2008,
  author = {Navascués, Miguel and Pironio, Stefano and Acín, Antonio},
  title = {A convergent hierarchy of semidefinite programs characterizing the set of quantum correlations},
  journal = {New Journal of Physics},
  year = {2008},
  volume = {10},
  number = {7},
  pages = {073013},
  publisher = {IOP Publishing},
  doi = {10.1088/1367-2630/10/7/073013},
  url = {http://dx.doi.org/10.1088/1367-2630/10/7/073013},
  issn = {1367-2630},
  month = {jul}
}

@article{Navascus2007,
  author = {Navascués, Miguel and Pironio, Stefano and Acín, Antonio},
  title = {Bounding the Set of Quantum Correlations},
  journal = {Physical Review Letters},
  year = {2007},
  volume = {98},
  number = {1},
  publisher = {American Physical Society (APS)},
  doi = {10.1103/physrevlett.98.010401},
  url = {http://dx.doi.org/10.1103/PhysRevLett.98.010401},
  issn = {1079-7114},
  month = {jan}
}

@misc{ji2022mipre,
      title={MIP*=RE}, 
      author={Zhengfeng Ji and Anand Natarajan and Thomas Vidick and John Wright and Henry Yuen},
      year={2022},
      eprint={2001.04383},
      archivePrefix={arXiv},
      primaryClass={quant-ph},
      url={https://arxiv.org/abs/2001.04383}, 
}

@article{OZAWA2004,
  title = {ABOUT THE QWEP CONJECTURE},
  volume = {15},
  ISSN = {1793-6519},
  url = {http://dx.doi.org/10.1142/S0129167X04002417},
  DOI = {10.1142/s0129167x04002417},
  number = {05},
  journal = {International Journal of Mathematics},
  publisher = {World Scientific Pub Co Pte Ltd},
  author = {OZAWA,  NARUTAKA},
  year = {2004},
  month = jul,
  pages = {501–530}
}

@article{Kirchberg1993,
  title = {On non-semisplit extensions,  tensor products and exactness of groupC *-algebras},
  volume = {112},
  ISSN = {1432-1297},
  url = {http://dx.doi.org/10.1007/BF01232444},
  DOI = {10.1007/bf01232444},
  number = {1},
  journal = {Inventiones Mathematicae},
  publisher = {Springer Science and Business Media LLC},
  author = {Kirchberg,  Eberhard},
  year = {1993},
  month = dec,
  pages = {449–489}
}

@misc{CorreaDaSilva_fDivergences,
  author = {R. Correa da Silva and M. B. Fr{\"o}b and G. Lechner and L. Sangaletti},
  title  = {Integral representations of $f$-divergences for general {von Neumann} algebras},
  year   = {2026},
  note   = {The paper appears on the same day on the arXiv}
}

@article{CerveroMartn2025,
  title = {Device independent security of quantum key distribution from monogamy-of-entanglement games},
  volume = {9},
  ISSN = {2521-327X},
  url = {http://dx.doi.org/10.22331/q-2025-03-05-1652},
  DOI = {10.22331/q-2025-03-05-1652},
  journal = {Quantum},
  publisher = {Verein zur Forderung des Open Access Publizierens in den Quantenwissenschaften},
  author = {Cervero-Martin,  Enrique and Tomamichel,  Marco},
  year = {2025},
  month = mar,
  pages = {1652}
}

@article{SLOFSTRA_2019_not_closed, 
title={THE SET OF QUANTUM CORRELATIONS IS NOT CLOSED}, 
volume={7}, 
DOI={10.1017/fmp.2018.3},
journal={Forum of Mathematics, Pi}, 
author={SLOFSTRA, WILLIAM}, 
year={2019}, 
pages={e1}}

@article{Arveson1969Subalgebras,
  author  = {Arveson, W. B.},
  title   = {Subalgebras of ${C}^*$-algebras},
  journal = {Acta Mathematica},
  volume  = {123},
  year    = {1969},
  pages   = {141--224}
}

@article{Paulsen2016,
  title = {Estimating quantum chromatic numbers},
  volume = {270},
  ISSN = {0022-1236},
  url = {http://dx.doi.org/10.1016/j.jfa.2016.01.010},
  DOI = {10.1016/j.jfa.2016.01.010},
  number = {6},
  journal = {Journal of Functional Analysis},
  publisher = {Elsevier BV},
  author = {Paulsen,  Vern I. and Severini,  Simone and Stahlke,  Daniel and Todorov,  Ivan G. and Winter,  Andreas},
  year = {2016},
  month = mar,
  pages = {2188–2222}
}

@article{PrezGarca2008,
  title = {Unbounded Violation of Tripartite Bell Inequalities},
  volume = {279},
  ISSN = {1432-0916},
  url = {http://dx.doi.org/10.1007/s00220-008-0418-4},
  DOI = {10.1007/s00220-008-0418-4},
  number = {2},
  journal = {Communications in Mathematical Physics},
  publisher = {Springer Science and Business Media LLC},
  author = {Pérez-García,  D. and Wolf,  M. M. and Palazuelos,  C. and Villanueva,  I. and Junge,  M.},
  year = {2008},
  month = feb,
  pages = {455–486}
}

@article{Junge2010,
  title = {Operator Space Theory: A Natural Framework for Bell Inequalities},
  volume = {104},
  ISSN = {1079-7114},
  url = {http://dx.doi.org/10.1103/PhysRevLett.104.170405},
  DOI = {10.1103/physrevlett.104.170405},
  number = {17},
  journal = {Physical Review Letters},
  publisher = {American Physical Society (APS)},
  author = {Junge,  M. and Palazuelos,  C. and Pérez-García,  D. and Villanueva,  I. and Wolf,  M. M.},
  year = {2010},
  month = apr 
}

@article{Slofstra_2019,
   title={Tsirelson’s problem and an embedding theorem for groups arising from non-local games},
   volume={33},
   ISSN={1088-6834},
   url={http://dx.doi.org/10.1090/jams/929},
   DOI={10.1090/jams/929},
   number={1},
   journal={Journal of the American Mathematical Society},
   publisher={American Mathematical Society (AMS)},
   author={Slofstra, William},
   year={2019},
   month=sep, pages={1–56} }

@article{Ligthart_2023,
   title={A Convergent Inflation Hierarchy for Quantum Causal Structures},
   volume={401},
   ISSN={1432-0916},
   url={http://dx.doi.org/10.1007/s00220-023-04697-7},
   DOI={10.1007/s00220-023-04697-7},
   number={3},
   journal={Communications in Mathematical Physics},
   publisher={Springer Science and Business Media LLC},
   author={Ligthart, Laurens T. and Gachechiladze, Mariami and Gross, David},
   year={2023},
   month=jun, pages={2673–2714} }

@book{Tomamichel2016,
  title = {Quantum Information Processing with Finite Resources},
  ISBN = {9783319218915},
  ISSN = {2197-1765},
  url = {http://dx.doi.org/10.1007/978-3-319-21891-5},
  DOI = {10.1007/978-3-319-21891-5},
  journal = {SpringerBriefs in Mathematical Physics},
  publisher = {Springer International Publishing},
  author = {Tomamichel,  Marco},
  year = {2016}
}

@article{Uhlmann1976,
  title = {The “transition probability” in the state space of a $\star$-algebra},
  volume = {9},
  ISSN = {0034-4877},
  url = {http://dx.doi.org/10.1016/0034-4877(76)90060-4},
  DOI = {10.1016/0034-4877(76)90060-4},
  number = {2},
  journal = {Reports on Mathematical Physics},
  publisher = {Elsevier BV},
  author = {Uhlmann,  A.},
  year = {1976},
  month = apr,
  pages = {273–279}
}

@article{vanLuijk2026,
  title = {Uniqueness of Purifications Is Equivalent to Haag Duality},
  volume = {136},
  ISSN = {1079-7114},
  url = {http://dx.doi.org/10.1103/d7nm-gx37},
  DOI = {10.1103/d7nm-gx37},
  number = {4},
  journal = {Physical Review Letters},
  publisher = {American Physical Society (APS)},
  author = {van Luijk,  Lauritz and Stottmeister,  Alexander and Wilming,  Henrik},
  year = {2026},
  month = jan 
}

@article{Ekert2014,
  title = {The ultimate physical limits of privacy},
  volume = {507},
  ISSN = {1476-4687},
  url = {http://dx.doi.org/10.1038/nature13132},
  DOI = {10.1038/nature13132},
  number = {7493},
  journal = {Nature},
  publisher = {Springer Science and Business Media LLC},
  author = {Ekert,  Artur and Renner,  Renato},
  year = {2014},
  month = mar,
  pages = {443–447}
}

@inproceedings{Pfitzmann2000,
  series = {CCS00},
  title = {Composition and integrity preservation of secure reactive systems},
  url = {http://dx.doi.org/10.1145/352600.352639},
  DOI = {10.1145/352600.352639},
  booktitle = {Proceedings of the 7th ACM conference on Computer and Communications Security},
  publisher = {ACM},
  author = {Pfitzmann,  Birgit and Waidner,  Michael},
  year = {2000},
  month = nov,
  pages = {245–254},
  collection = {CCS00}
}

@misc{ferradini2025definingsecurityquantumkey,
      title={Defining Security in Quantum Key Distribution}, 
      author={Carla Ferradini and Martin Sandfuchs and Ramona Wolf and Renato Renner},
      year={2025},
      eprint={2509.13405},
      archivePrefix={arXiv},
      primaryClass={quant-ph},
      url={https://arxiv.org/abs/2509.13405}, 
}

@article{Tomamichel2011,
  title = {Leftover Hashing Against Quantum Side Information},
  volume = {57},
  ISSN = {1557-9654},
  url = {http://dx.doi.org/10.1109/TIT.2011.2158473},
  DOI = {10.1109/tit.2011.2158473},
  number = {8},
  journal = {IEEE Transactions on Information Theory},
  publisher = {Institute of Electrical and Electronics Engineers (IEEE)},
  author = {Tomamichel,  Marco and Schaffner,  Christian and Smith,  Adam and Renner,  Renato},
  year = {2011},
  month = aug,
  pages = {5524–5535}
}

@article{Dupuis2023,
  title = {Privacy Amplification and Decoupling Without Smoothing},
  volume = {69},
  ISSN = {1557-9654},
  url = {http://dx.doi.org/10.1109/TIT.2023.3301812},
  DOI = {10.1109/tit.2023.3301812},
  number = {12},
  journal = {IEEE Transactions on Information Theory},
  publisher = {Institute of Electrical and Electronics Engineers (IEEE)},
  author = {Dupuis,  Frédéric},
  year = {2023},
  month = dec,
  pages = {7784–7792}
}

@misc{regula2026rethinkingquantumsmoothentropies,
      title={Rethinking quantum smooth entropies: Tight one-shot analysis of quantum privacy amplification}, 
      author={Bartosz Regula and Marco Tomamichel},
      year={2026},
      eprint={2603.04493},
      archivePrefix={arXiv},
      primaryClass={quant-ph},
      url={https://arxiv.org/abs/2603.04493}, 
}

@article{Dupuis_2020,
   title={Entropy Accumulation},
   volume={379},
   ISSN={1432-0916},
   url={http://dx.doi.org/10.1007/s00220-020-03839-5},
   DOI={10.1007/s00220-020-03839-5},
   number={3},
   journal={Communications in Mathematical Physics},
   publisher={Springer Science and Business Media LLC},
   author={Dupuis, Frédéric and Fawzi, Omar and Renner, Renato},
   year={2020},
   month=sep, pages={867–913} }

@article{Dupuis_2019,
   title={Entropy Accumulation With Improved Second-Order Term},
   volume={65},
   ISSN={1557-9654},
   url={http://dx.doi.org/10.1109/TIT.2019.2929564},
   DOI={10.1109/tit.2019.2929564},
   number={11},
   journal={IEEE Transactions on Information Theory},
   publisher={Institute of Electrical and Electronics Engineers (IEEE)},
   author={Dupuis, Frederic and Fawzi, Omar},
   year={2019},
   month=nov, pages={7596–7612} }

@article{Metger2024,
  title = {Generalised Entropy Accumulation},
  volume = {405},
  ISSN = {1432-0916},
  url = {http://dx.doi.org/10.1007/s00220-024-05121-4},
  DOI = {10.1007/s00220-024-05121-4},
  number = {11},
  journal = {Communications in Mathematical Physics},
  publisher = {Springer Science and Business Media LLC},
  author = {Metger,  Tony and Fawzi,  Omar and Sutter,  David and Renner,  Renato},
  year = {2024},
  month = oct 
}

@article{Tomamichel2009,
  title = {A Fully Quantum Asymptotic Equipartition Property},
  volume = {55},
  ISSN = {1557-9654},
  url = {http://dx.doi.org/10.1109/TIT.2009.2032797},
  DOI = {10.1109/tit.2009.2032797},
  number = {12},
  journal = {IEEE Transactions on Information Theory},
  publisher = {Institute of Electrical and Electronics Engineers (IEEE)},
  author = {Tomamichel,  Marco and Colbeck,  Roger and Renner,  Renato},
  year = {2009},
  month = dec,
  pages = {5840–5847}
}

@article{Fawzi_2025,
   title={Asymptotic Equipartition Theorems in von Neumann Algebras},
   volume={27},
   ISSN={1424-0661},
   url={http://dx.doi.org/10.1007/s00023-025-01545-3},
   DOI={10.1007/s00023-025-01545-3},
   number={1},
   journal={Annales Henri Poincaré},
   publisher={Springer Science and Business Media LLC},
   author={Fawzi, Omar and Gao, Li and Rahaman, Mizanur},
   year={2025},
   month=feb, pages={95–141} }

@article{Berta_2015,
   title={The smooth entropy formalism for von Neumann algebras},
   volume={57},
   ISSN={1089-7658},
   url={http://dx.doi.org/10.1063/1.4936405},
   DOI={10.1063/1.4936405},
   number={1},
   journal={Journal of Mathematical Physics},
   publisher={AIP Publishing},
   author={Berta, Mario and Furrer, Fabian and Scholz, Volkher B.},
   year={2015},
   month=dec }

@article{kossmann2025routed,
  title={Routed Bell tests with arbitrarily many local parties},
  author={Ko{\ss}mann, Gereon and Berta, Mario and Schwonnek, Ren{\'e}},
  journal={arXiv preprint arXiv:2510.08405},
  year={2025}
}

@article{bluhm2026device,
  title={Device independent quantum key distribution with robust self-tests},
  author={Bluhm, Andreas and Ko{\ss}mann, Gereon and Schwonnek, Ren{\'e}},
  journal={arXiv preprint arXiv:2603.28085},
  year={2026}
}

@article{van2024schmidt,
  title={The Schmidt rank for the commuting operator framework},
  author={van Luijk, Lauritz and Schwonnek, Ren{\'e} and Stottmeister, Alexander and Werner, Reinhard F},
  journal={Communications in Mathematical Physics},
  volume={405},
  number={7},
  pages={152},
  year={2024},
  publisher={Springer}
}

@article{Junge2011,
  author        = {Junge, Marius and Navascu{\'e}s, Miguel and Palazuelos, Carlos and P{\'e}rez-Garc{\'i}a, David and Scholz, Volkher B. and Werner, Reinhard F.},
  title         = {Connes' Embedding Problem and Tsirelson's Problem},
  journal       = {Journal of Mathematical Physics},
  volume        = {52},
  number        = {1},
  pages         = {012102},
  year          = {2011},
  doi           = {10.1063/1.3514538},
  archivePrefix = {arXiv},
  eprint        = {1008.1142},
  primaryClass  = {quant-ph}
}

@article{lu2026device,
  title={Device-independent quantum key distribution over 100 km with single atoms},
  author={Lu, Bo-Wei and Yang, Chao-Wei and Wang, Run-Qi and Gao, Bo-Feng and Zhen, Yi-Zheng and Wang, Zhen-Gang and Shi, Jia-Kai and Ren, Zhong-Qi and Hahn, Thomas A and Tan, Ernest Y-Z and others},
  journal={Science},
  volume={391},
  number={6785},
  pages={592--597},
  year={2026},
  publisher={American Association for the Advancement of Science}
}

@misc{he2025qicsquantuminformationconic,
      title={QICS: Quantum Information Conic Solver}, 
      author={Kerry He and James Saunderson and Hamza Fawzi},
      year={2025},
      eprint={2410.17803},
      archivePrefix={arXiv},
      primaryClass={math.OC},
      url={https://arxiv.org/abs/2410.17803}, 
}

@article{Komann2026,
  title = {Optimising the relative entropy under semidefinite constraints},
  volume = {12},
  ISSN = {2056-6387},
  url = {http://dx.doi.org/10.1038/s41534-026-01184-4},
  DOI = {10.1038/s41534-026-01184-4},
  number = {1},
  journal = {npj Quantum Information},
  publisher = {Springer Science and Business Media LLC},
  author = {Koßmann,  Gereon and Schwonnek,  René},
  year = {2026},
  month = jan 
}

@article{Hu_2022,
   title={Robust Interior Point Method for Quantum Key Distribution Rate Computation},
   volume={6},
   ISSN={2521-327X},
   url={http://dx.doi.org/10.22331/q-2022-09-08-792},
   DOI={10.22331/q-2022-09-08-792},
   journal={Quantum},
   publisher={Verein zur Forderung des Open Access Publizierens in den Quantenwissenschaften},
   author={Hu, Hao and Im, Jiyoung and Lin, Jie and Lütkenhaus, Norbert and Wolkowicz, Henry},
   year={2022},
   month=sep, pages={792} }

@misc{kossmann2025reliableentropyestimationobserved,
      title={Reliable Entropy Estimation from Observed Statistics for Device-Independent Quantum Cryptography}, 
      author={Gereon Koßmann and René Schwonnek},
      year={2025},
      eprint={2411.04858},
      archivePrefix={arXiv},
      primaryClass={quant-ph},
      url={https://arxiv.org/abs/2411.04858}, 
}

@article{Bennett_2014,
   title={Quantum cryptography: Public key distribution and coin tossing},
   volume={560},
   ISSN={0304-3975},
   url={http://dx.doi.org/10.1016/j.tcs.2014.05.025},
   DOI={10.1016/j.tcs.2014.05.025},
   journal={Theoretical Computer Science},
   publisher={Elsevier BV},
   author={Bennett, Charles H. and Brassard, Gilles},
   year={2014},
   month=Dec, pages={7–11} }

@article{Brown_2024,
   title={Device-independent lower bounds on the conditional von Neumann entropy},
   volume={8},
   ISSN={2521-327X},
   url={http://dx.doi.org/10.22331/q-2024-08-27-1445},
   DOI={10.22331/q-2024-08-27-1445},
   journal={Quantum},
   publisher={Verein zur Forderung des Open Access Publizierens in den Quantenwissenschaften},
   author={Brown, Peter and Fawzi, Hamza and Fawzi, Omar},
   year={2024},
   month=aug, pages={1445} }

@article{Tan_2021,
   title={Computing secure key rates for quantum cryptography with untrusted devices},
   volume={7},
   ISSN={2056-6387},
   url={http://dx.doi.org/10.1038/s41534-021-00494-z},
   DOI={10.1038/s41534-021-00494-z},
   number={1},
   journal={npj Quantum Information},
   publisher={Springer Science and Business Media LLC},
   author={Tan, Ernest Y.-Z. and Schwonnek, René and Goh, Koon Tong and Primaatmaja, Ignatius William and Lim, Charles C.-W.},
   year={2021},
   month=oct }

@article{Frenkel2023,
  title = {Integral formula for quantum relative entropy implies data processing inequality},
  volume = {7},
  ISSN = {2521-327X},
  url = {http://dx.doi.org/10.22331/q-2023-09-07-1102},
  DOI = {10.22331/q-2023-09-07-1102},
  journal = {Quantum},
  publisher = {Verein zur Forderung des Open Access Publizierens in den Quantenwissenschaften},
  author = {Frenkel,  Péter E.},
  year = {2023},
  month = sep,
  pages = {1102}
}

@article{FRITZ_2012,
   title={TSIRELSON’S PROBLEM AND KIRCHBERG’S CONJECTURE},
   volume={24},
   ISSN={1793-6659},
   url={http://dx.doi.org/10.1142/S0129055X12500122},
   DOI={10.1142/s0129055x12500122},
   number={05},
   journal={Reviews in Mathematical Physics},
   publisher={World Scientific Pub Co Pte Lt},
   author={FRITZ, TOBIAS},
   year={2012},
   month=may, pages={1250012} }

@book{ArnonFriedman2020,
  title = {Device-Independent Quantum Information Processing: A Simplified Analysis},
  ISBN = {9783030602314},
  ISSN = {2190-5061},
  url = {http://dx.doi.org/10.1007/978-3-030-60231-4},
  DOI = {10.1007/978-3-030-60231-4},
  journal = {Springer Theses},
  publisher = {Springer International Publishing},
  author = {Arnon-Friedman,  Rotem},
  year = {2020}
}

@article{Stinespring1955,
  title = {Positive Functions on $C^\star$-Algebras},
  volume = {6},
  ISSN = {0002-9939},
  url = {http://dx.doi.org/10.2307/2032342},
  DOI = {10.2307/2032342},
  number = {2},
  journal = {Proceedings of the American Mathematical Society},
  publisher = {JSTOR},
  author = {Stinespring,  W. Forrest},
  year = {1955},
  month = apr,
  pages = {211}
}

@article{Haagerup2009,
  title = {A reduction method for noncommutative $L_p$-spaces and applications},
  volume = {362},
  ISSN = {0002-9947},
  url = {http://dx.doi.org/10.1090/S0002-9947-09-04935-6},
  DOI = {10.1090/s0002-9947-09-04935-6},
  number = {4},
  journal = {Transactions of the American Mathematical Society},
  publisher = {American Mathematical Society (AMS)},
  author = {Haagerup,  Uffe and Junge,  Marius and Xu,  Quanhua},
  year = {2009},
  month = oct,
  pages = {2125–2165}
}

@book{Paulsen2003,
  title = {Completely Bounded Maps and Operator Algebras},
  ISBN = {9780511546631},
  url = {http://dx.doi.org/10.1017/CBO9780511546631},
  DOI = {10.1017/cbo9780511546631},
  publisher = {Cambridge University Press},
  author = {Paulsen,  Vern},
  year = {2003},
  month = feb 
}

@book{Stratila2019,
  title = {Lectures on von Neumann Algebras},
  ISBN = {9781108654975},
  url = {http://dx.doi.org/10.1017/9781108654975},
  DOI = {10.1017/9781108654975},
  publisher = {Cambridge University Press},
  author = {Stratila,  Serban-Valentin and Zsido,  Laszlo},
  year = {2019},
  month = apr 
}

@article{Jencova2024, 
    title={Recoverability of quantum channels via hypothesis testing}, volume={114}, ISSN={1573-0530}, DOI={10.1007/s11005-024-01775-2}, number={1}, journal={Letters in Mathematical Physics}, publisher={Springer Science and Business Media LLC}, author={Jenčová, Anna}, year={2024}, month=feb }

@misc{tomamichel2013frameworknonasymptoticquantuminformation,
      title={A Framework for Non-Asymptotic Quantum Information Theory}, 
      author={Marco Tomamichel},
      year={2013},
      eprint={1203.2142},
      archivePrefix={arXiv},
      primaryClass={quant-ph},
      url={https://arxiv.org/abs/1203.2142}, 
}

@article{Jennewein2000,
  title = {Quantum Cryptography with Entangled Photons},
  volume = {84},
  ISSN = {1079-7114},
  url = {http://dx.doi.org/10.1103/PhysRevLett.84.4729},
  DOI = {10.1103/physrevlett.84.4729},
  number = {20},
  journal = {Physical Review Letters},
  publisher = {American Physical Society (APS)},
  author = {Jennewein,  Thomas and Simon,  Christoph and Weihs,  Gregor and Weinfurter,  Harald and Zeilinger,  Anton},
  year = {2000},
  month = May,
  pages = {4729–4732}
}

@article{Naik2000,
  title = {Entangled State Quantum Cryptography: Eavesdropping on the Ekert Protocol},
  volume = {84},
  ISSN = {1079-7114},
  url = {http://dx.doi.org/10.1103/PhysRevLett.84.4733},
  DOI = {10.1103/physrevlett.84.4733},
  number = {20},
  journal = {Physical Review Letters},
  publisher = {American Physical Society (APS)},
  author = {Naik,  D. S. and Peterson,  C. G. and White,  A. G. and Berglund,  A. J. and Kwiat,  P. G.},
  year = {2000},
  month = May,
  pages = {4733–4736}
}

@article{Tittel2000,
  title = {Quantum Cryptography Using Entangled Photons in Energy-Time Bell States},
  volume = {84},
  ISSN = {1079-7114},
  url = {http://dx.doi.org/10.1103/PhysRevLett.84.4737},
  DOI = {10.1103/physrevlett.84.4737},
  number = {20},
  journal = {Physical Review Letters},
  publisher = {American Physical Society (APS)},
  author = {Tittel,  W. and Brendel,  J. and Zbinden,  H. and Gisin,  N.},
  year = {2000},
  month = May,
  pages = {4737–4740}
}

@article{Nadlinger2022,
  title = {Experimental quantum key distribution certified by Bell’s theorem},
  volume = {607},
  ISSN = {1476-4687},
  url = {http://dx.doi.org/10.1038/s41586-022-04941-5},
  DOI = {10.1038/s41586-022-04941-5},
  number = {7920},
  journal = {Nature},
  publisher = {Springer Science and Business Media LLC},
  author = {Nadlinger,  D. P. and Drmota,  P. and Nichol,  B. C. and Araneda,  G. and Main,  D. and Srinivas,  R. and Lucas,  D. M. and Ballance,  C. J. and Ivanov,  K. and Tan,  E. Y.-Z. and Sekatski,  P. and Urbanke,  R. L. and Renner,  R. and Sangouard,  N. and Bancal,  J.-D.},
  year = {2022},
  month = {7},
  pages = {682–686}
}

@article{Zhang2022,
  title = {A device-independent quantum key distribution system for distant users},
  volume = {607},
  ISSN = {1476-4687},
  url = {http://dx.doi.org/10.1038/s41586-022-04891-y},
  DOI = {10.1038/s41586-022-04891-y},
  number = {7920},
  journal = {Nature},
  publisher = {Springer Science and Business Media LLC},
  author = {Zhang,  Wei and van Leent,  Tim and Redeker,  Kai and Garthoff,  Robert and Schwonnek,  René and Fertig,  Florian and Eppelt,  Sebastian and Rosenfeld,  Wenjamin and Scarani,  Valerio and Lim,  Charles C.-W. and Weinfurter,  Harald},
  year = {2022},
  month = {7},
  pages = {687–691}
}

@article{Liu2022,
  title = {Toward a Photonic Demonstration of Device-Independent Quantum Key Distribution},
  volume = {129},
  ISSN = {1079-7114},
  url = {http://dx.doi.org/10.1103/PhysRevLett.129.050502},
  DOI = {10.1103/physrevlett.129.050502},
  number = {5},
  journal = {Physical Review Letters},
  publisher = {American Physical Society (APS)},
  author = {Liu,  Wen-Zhao and Zhang,  Yu-Zhe and Zhen,  Yi-Zheng and Li,  Ming-Han and Liu,  Yang and Fan,  Jingyun and Xu,  Feihu and Zhang,  Qiang and Pan,  Jian-Wei},
  year = {2022},
  month = {7} 
}

@article{Ren05,
	doi = "10.3929/ethz-a-005115027",
	author = "Renato Renner",
	title = "Security of quantum key distribution",
	journal= "Ph.D.~Thesis (ETH Zurich), arXiv:quant-ph/0512258",
	year = "2005",
}
